\documentclass[11pt]{article}
\usepackage{amsthm,amsgen,amsmath,amstext,amsbsy,amsopn,amssymb,latexsym,mathrsfs}
\usepackage{array}
\usepackage{cite}
\usepackage{multirow}
\usepackage{graphicx}
\usepackage{amssymb,bm,bbm}
\usepackage{titling}
\usepackage{url,bm}
\usepackage{algpseudocode, float, algorithm}
\usepackage[round]{natbib}
\usepackage{bbm}
\usepackage{tabularx}
\usepackage{authblk}
\usepackage{hyperref}

\hypersetup{
    colorlinks=true,
    linkcolor=blue,
    filecolor=blue,
    urlcolor=blue,
    citecolor=blue,
}

 \usepackage[dvipsnames,table,dvipsnames*, svgnames*]{xcolor}

\allowdisplaybreaks

\newcommand{\beq}{\begin{equation}}
\newcommand{\eeq}{\end{equation}}
\newcommand{\beas}{\begin{eqnarray*}}
\newcommand{\eeas}{\end{eqnarray*}}
\newcommand{\bea}{\begin{eqnarray}}
\newcommand{\eea}{\end{eqnarray}}
\newcommand{\bei}{\begin{itemize}}
\newcommand{\eei}{\end{itemize}}
\newcommand{\ben}{\begin{enumerate}}
\newcommand{\een}{\end{enumerate}}
\newcommand{\argmin}{\mathop{\rm arg\min}}
\newcommand{\argmax}{\mathop{\rm arg\max}}

\let\emptyset\varnothing

\newcommand{\norm}[1]{\left\|#1\right\|}
\newcommand{\abs}[1]{\left|#1\right|}

\newtheorem{Corollary}{Corollary}

\newtheorem{Lemma}{Lemma}

\newtheorem{Theorem}{Theorem}

\newtheorem{Remark}{Remark}

\newtheorem{Assumption}{Assumption}

\newcommand{\R}{\mathbb{R}}
\newcommand{\E}{{\mathbb{E}}}
\newcommand{\Prob}{{\mathbb{P}}}
\newcommand{\1}{{\mathbbm{1}}}

\newcommand{\supp}{{\rm supp}}
\newcommand{\Var}{{\rm Var}}

\newcommand{\normm}[1]{{\left\vert\kern-0.25ex\left\vert\kern-0.25ex\left\vert #1 
    \right\vert\kern-0.25ex\right\vert\kern-0.25ex\right\vert}}

\usepackage{tikz}
\newcommand{\solidcirc}{\tikz\draw[black,fill=black] (0,0) circle (.5ex);}
\newcommand{\hollowcirc}{\tikz\draw[black] (0,0) circle (.5ex);}

\graphicspath{{Figures/}}

\begin{document}
\part*{}
\title{Repro Samples Method for High-dimensional Logistic Model}
\author[1]{Xiaotian Hou}
\author[1]{Linjun Zhang}
\author[2]{Peng Wang}
\author[1]{Min-ge Xie}
\affil[1]{Department of Statistics, Rutgers University}
\affil[2]{Department of Operations, Business Analytics, and Information Systems, University of Cincinnati}
\date{ }
\maketitle

\vspace{-20mm}

\begin{abstract}
     This paper presents a novel method to make statistical inferences for both the model support and regression coefficients in a high-dimensional logistic regression model. Our method is based on the repro samples framework, in which we conduct statistical inference by generating artificial samples mimicking the actual data-generating process. The proposed method has two major advantages. Firstly, for model support, we introduce the first method for constructing model confidence set in a high-dimensional setting and the proposed method only requires a weak signal strength assumption. Secondly, in terms of regression coefficients, we establish confidence sets for any group of linear combinations of regression coefficients. 
    Our simulation results demonstrate that the proposed method produces valid and small model confidence sets and achieves better coverage for regression coefficients than the state-of-the-art debiasing methods. 
    Additionally, we analyze single-cell RNA-seq data on the immune response. 
    Besides identifying genes previously proved as relevant in the literature, our method also discovers a significant gene that has not been studied before, revealing a potential new direction in understanding cellular immune response mechanisms. 
\end{abstract}

\section{Introduction}

Logistic regression is a widely used method for classification \citep[e.g.,][]{hosmer2013applied}. In the high-dimensional setting where the number of variables $p$ exceeds the sample size $n$, existing inference methods for logistic regression mainly focus on the inference of regression coefficients and their functions, while the uncertainty quantification of model selection remains relatively unexplored. The main challenge is that the model support is discrete, and therefore classical methods for statistical inference such as the central limit theorem cannot be applied. Some recent works have studied the problem of constructing model confident sets \citep{hansen2011model,ferrari2015confidence,zheng2019model,li2019model}. However, these methods are either designed for low-dimensional models or require a model selection procedure with a sure screening property \citep{fan2010sure} to reduce the dimension, which relies on strong signal strength assumptions. 


To address this problem, we provide several further developments on the so-called repro samples method and extend the work by \cite{wang2022finite} on high-dimensional linear regression models to high-dimensional logistic regression models. Our work differs with \cite{wang2022finite} in several aspects. First, unlike linear regression, the response in logistic regression is binary and the information of regression coefficients is highly compressed, making it hard to recover the linear discriminant function in finite-sample as \cite{wang2022finite} did. 
Second, Gaussian noises are assumed in \cite{wang2022finite}, in which case one can get rid of the nuisance parameters using the corresponding sufficient statistics and construct finite-sample pivot statistics for statistical inference.  In logistic models, however, such pivot statistics are no longer available. We thus use asymptotic approximations to characterize the distribution of test statistics and use a profiling method to handle nuisance parameters. To facilitate the asymptotic approximation, we also consider the random design setting instead of the fixed design as studied in \cite{wang2022finite}.


A key step of our method is to search for a relatively small set of candidate models that include the ground-truth support with high probability, which can be done using an inversion method by noting that the model space is discrete. 
Then, a likelihood ratio test can be applied to each candidate model to infer the regression coefficients. We use the following intuition to construct a confidence set for the model support -- for each candidate model, we generate artificial samples using that model, then compare the summary statistics calculated based on the generated artificial data to that obtained using the observed data; when these two statistics are very different, we reject the current candidate model. We provide rigorous mathematical theories to support our developments. 

Our contributions are as follows:
\begin{itemize}
    \item[(1)] We propose a novel method to construct the model candidate set that provably contains the true model with high probability as long as the model is identifiable. Here, we only assume that the signal is sufficient to identify the model given the realization of the error term. 
    \item[(2)] Based on the model candidate set, we further construct a confidence set for the model support with a desired confidence level. 
    To the best of our knowledge, this is the first approach for constructing model confidence sets in the high-dimensional logistic regression setting. 
    \item[(3)] Besides the model confidence set, we also develop a comprehensive approach that allows for inference on individual regression coefficients, subsets of coefficients, and even any groups of linear combinations of these coefficients. 
    This general result also enables us to efficiently infer nonlinear transformations of the regression coefficients, such as the case probabilities of a set of new observations.
    Existing works in the literature only focus on inferring a constrained group of linear combinations of regression coefficients, e.g., see \cite{van2014asymptotically,zhang2017simultaneous,shi2019linear}.
    \item[(4)] Existing approaches for high-dimensional inference often depend on either data splitting, feature screening, consistent model or parameter estimations, or sparse inverse Hessian matrix. 
    However, our methods do not rely on these procedures or assumptions. 
    Therefore, the proposed methods require weaker theoretical assumptions and attain the full sample efficiency.  
\end{itemize}


\subsection{Related works}

 The inference problem in high-dimensional generalized linear models has attracted a lot of attention in recent years, almost all of which however focus on the regression coefficients. For instance, \cite{van2014asymptotically} constructed confidence intervals for $\beta_j, j\in[p]$ based on the debiased estimator by inverting the KKT condition of the $\ell_1$ penalized regression problems. \cite{dezeure2015high} proposed the multi sample-splitting approach by repeatedly splitting the data into two parts. For each single splitting, they apply variable selection on one part of the data and apply a low dimensional inference procedure on the other part of the data with the selected variables to produce a $p$-value. Then they aggregate all the $p$-values to make the final decision. \cite{belloni2016post} and \cite{chernozhukov2018double} make inferences for single parameters by constructing Neyman-orthogonal estimating equations, so that the proposed estimators are immunized against high-dimensional nuisance parameters and variable selection mistakes. \cite{ning2017general} tested single parameters by constructing a decorrelated score function to handle the impact of high-dimensional nuisance parameters. \cite{shi2019linear} considered the problem of linear hypothesis testing using the partial penalized likelihood where only the high-dimensional parameters not involved in the hypothesis are penalized. \cite{sur2019modern} studied the likelihood ratio test for single parameters under the setting where $\frac{p}{n}\rightarrow \kappa,\kappa<\frac{1}{2}$. \cite{ma2021global} considered the global and multiple testing problem by a two-step standardization procedure. \cite{cai2021statistical} studied high-dimensional GLM with binary outcomes using the debiased procedure together with a link-specific weighting to ensure the bias of the estimator is dominated by the stochastic error. \cite{shi2021statistical} generalized the decorrelated score method of \cite{ning2017general} by recursively conducting model selection and constructing score equations based on the selected variables in an online manner. \cite{fei2021estimation} proposed the splitting and smoothing method for inferring single coefficients. The work of \cite{dezeure2015high,shi2019linear,shi2021statistical,fei2021estimation} require either variable selection methods with sure screening property \citep{fan2010sure} or a uniform signal strength condition. The work of \cite{van2014asymptotically,belloni2016post,ning2017general,ma2021global} relies on either the sparse inverse Hessian assumption or the sparse precision assumption which can be too stringent in practice. The method in \cite{cai2021statistical} uses sample splitting which could cause a loss of efficiency for inference. 
 Our method uses generated random noise to approximate the data generating noise, therefore we only require a model selection consistency in the oracle setting, resulting in a much weaker signal strength assumption than the aforementioned methods.
 Our method requires neither the sparse inverse Hessian, sparse precision matrix assumptions nor data splitting.

It is also of interest to quantify the uncertainty and make inference for the model support, a task that we can do. This inference problem is more difficult than the inference problem for the regression coefficients since the model space is discrete. 
There are several works to construct a model confidence set, but most of them are in the classical regression settings with $p < n$. For instance, 
\cite{hansen2011model} constructed the model confidence set by a sequence of equivalence tests and eliminations. Specifically, starting from a set of candidate models, they apply the equivalence tests to all the remaining model pairs in the candidate set, when any test is rejected, they eliminate the worst model from the candidate set. Then they repeat the procedure until all models in the candidate set are tested as equivalent. \cite{ferrari2015confidence} constructed the variable selection confidence set for linear regression based on $F$-testing. Starting from a relatively small full model (that is assumed to be available), they compare each of the sub-models to the full model using $F$-testing, then collect all the accepted sub-models as the variable selection confidence set. \cite{zheng2019model} extended the linear regression models in \cite{ferrari2015confidence} to more general models by comparing the sub-models to the full model using the likelihood ratio test. \cite{li2019model} introduced the model confidence bounds as two nested models such that the true model is between these two models with a certain level of confidence. To achieve this, they construct several Bootstrap samples and apply a model selection procedure to each of them to form a set of selected models. Then they choose the model confidence bounds such that it achieves the pre-specified coverage level on the set of selected models. The work of \cite{hansen2011model,ferrari2015confidence,zheng2019model} requires either the dimension of the data to be less than the sample size, or a variable screening procedure with sure screening property and thus a uniform signal strength condition. The work of \cite{li2019model} relies on a consistent model selection procedure where uniform signal strength is again necessary. Our proposed method does not have these constraints and it directly applies for high-dimensional models with $p \gg n$.  

A very recent work by \cite{wang2022finite} uses the repro samples method proposed in \cite{xie2022repro} to address the statistical inference for both regression coefficients and model support in a high-dimensional Gaussian linear regression model with finite sample coverage guarantee. Their artificial-sample-based method mimics the data generating process by sampling from the known noise distribution and generating synthetic response variables using the generated random noises. Noticing that if one knows the noise that generates the observed data, they could calculate all the possible values of the parameters that are able to generate the observed data using the noise, then the uncertainty of identifying the parameters merely comes from the inversion of the data generating process. However the data generating noise is unobservable, the repro samples method then incorporates both the uncertainty of the inversion of the data generating process and the uncertainty of the random noise to construct a confidence set for the parameters. Our developments use the same idea of the repro sample method to develop our inference. However, \cite{wang2022finite} focuses on the much easier setting of Gaussian regression models and their development can not be directly applied to high-dimensional logistic regression models. 

\subsection{Organization}
The paper is organized as follows. Section \ref{sec_preliminaries} includes notations, model setup, and a brief review of the repro samples method. Section \ref{sec_method} is the main section that introduces our methods for constructing the model candidate set, model confidence set, and inference for the regression coefficients. Theoretical properties of the proposed methods are also included. 
Section \ref{sec_numerical} provides numerical illustrations of the proposed method, which includes both simulation and real data examples. Section \ref{sec:disc} contains further remarks and discussions.
The proof of the theoretical results 
are in Appendix (Section \ref{sec_proof}).

\section{Preliminaries}\label{sec_preliminaries}
In this section, we introduce notation and the data generative model in Section \ref{sec_notation} and \ref{sec_model}, respectively. We also briefly review the general repro samples method in Section \ref{sec_repro}. 

\subsection{Notation}\label{sec_notation}
For any $p\in\mathbb{N}_+$, we denote $[p]$ to be the set $\{1,\ldots,p\}$. For a vector $v\in\R^p$ and a subset of indexes $\tau\subset[p]$, we denote $v_\tau$ to be the sub-vector of $v$ with indexes in $\tau$, denote $\norm{v}_k=(\sum_{j\in[p]}\abs{v_j}^k)^{1/k}$ for $k\ge 0$ with $\norm{v}_0=\sum_{j\in[p]}\1\{v_j\ne 0\}$ to be the number of nonzero elements in $v$ and $\norm{v}_\infty=\max_{j\in[p]}\abs{v_j}$. We also denote $\abs{\tau}=\sum_{j\in[p]}\1\{j\in\tau\}$ to be the cardinality of $\tau$. For matrix $A\in\R^{q\times p}$ and $\tau\subset[p]$, we denote $A_{\cdot,\tau}$ to be a submatrix of $A$ consisting of all the columns of $A$ with column indexes in $\tau$ and $\norm{A}_{\rm op}=\sup_{a\in\R^q,b\in\R^p}a^\top Ab$ is the operator norm of $A$. For a symmetric matrix $A$, $\lambda_{\min}(A)$ and $\lambda_{\max}(A)$ denote respectively the smallest and largest eigenvalues of $A$. We use $c$ and $C$ to denote absolute positive constants that may vary from place to place. For two positive sequences $a_n$ and $b_n$, $a_n\lesssim b_n$ means $a_n\le Cb_n$ for all $n$ and $a_n\gtrsim b_n$ if $b_n\lesssim a_n$ and $a_n\asymp b_n$ if $a_n\lesssim b_n$ and $b_n\lesssim a_n$, and $a_n\ll b_n$ if $\lim\sup_{n\rightarrow\infty}\frac{a_n}{b_n}=0$ and $a_n\gg b_n$ if $b_n\ll a_n$.

\subsection{Model Set-up}\label{sec_model}

In this work, we consider the logistic regression with independent observations $\{(X_i,y_i):i\in[n] \}$ generated from the following  distribution
\[\Prob(Y=1\mid X)=1-\Prob(Y=0\mid X)=\frac{1}{1+e^{-X^\top{\bm{\beta}}_0}},\quad X\sim \Prob_X,\]
where $X\in\R^p, Y\in\{0,1\}$ and ${\bm{\beta}}_0\in\R^p$ is the linear coefficients in logistic regression with $\|{\bm{\beta}}_0\|_0=s$. We can rewrite this model in the form of a data generating model
\begin{equation}\label{eq_model}
    Y = \mathbbm{1}\{X^\top{\bm{\beta}}_0 +\epsilon> 0 \},
\end{equation}
with $X\sim\Prob_X$, $\epsilon\sim{\rm Logistic}$ is a logistic random variable with cumulative distribution function $\Prob(\epsilon\le t)=(1+e^{-t})^{-1}$. Here we assume ${\bm{\beta}}_0$ defined in model \eqref{eq_model} is unique and this assumption will be satisfied if $\sup_{a\in\R^p,\norm{a}_2=1}\Prob(a^\top X=0)=0$. To highlight the observed data and the correspondence between $X_i, y_i$ and $\epsilon_i,$ for $ i\in[n]$, we use $\{(X_i^{obs}, y_i^{obs}, \epsilon_i^{rel}):i\in[n]\}$ to denote the oracle data, which consists of the \textit{observed data} and the corresponding \textit{realization} of the data generating noises. Denote $\bm{X}=(X_1, \ldots, X_n)^\top$, $\bm{X}^{obs}=(X_1^{obs}, \ldots, X_n^{obs})$, $\bm{y}=(y_1, \ldots, y_n)^\top$, $\bm{y}^{obs}=(y_1^{obs}, \ldots, y_n^{obs})^\top$, $\bm{\epsilon}=(\epsilon_1, \ldots, \epsilon_n)^\top$, $\bm{\epsilon}^{rel}=(\epsilon_1^{rel}, \ldots, \epsilon_n^{rel})^\top$. Throughout the paper, we use $\bm{X}, \bm{y}$ and $\bm{\epsilon}$ to denote the random copy of data and corresponding random noises, respectively. We use $\bm{X}^{obs}, \bm{y}^{obs}$ and $\bm{\epsilon}^{rel}$ when the observed data is treated as given (or realized). 

In this model, we have two sets of unknown parameters, one is the support of ${\bm{\beta}}_0$ denoted by $\tau_0=\supp({\bm{\beta}}_0)=\{j\in[p]: {\beta}_{0,j}\ne 0\}$ and the other one is the vector of nonzero coefficients ${\bm{\beta}}_{0,\tau_0}$.  We write $\bm{\theta}=(\tau,{\bm{\beta}}_{\tau})$, $\bm{\theta}_0=(\tau_0,{\bm{\beta}}_{0,\tau_0})$ and denote $\Prob_{\bm{\theta}_0}$ to be the joint distribution of $(X,Y)$ defined in Equation \eqref{eq_model} with parameters equal to $\bm\theta_0$. 
The log-likelihood function of $\bm{\theta}$ is then given by 
$$l(\bm{\theta}\mid\bm{X}^{obs},\bm{y}^{obs})=-\sum_{i=1}^n \log(1+e^{-(2y_i^{obs}-1)X_{i,\tau}^{obs\top}{\bm{\beta}}_{\tau}}).$$
Most work of statistical inference on high-dimensional logistic models only focuses on the parameters ${\bm{\beta}}_0$. 
In this paper, we are interested in making inferences for both the true model~$\tau_0$~and linear coefficients~${\bm{\beta}}_0$.

\subsection{Repro Samples Method}\label{sec_repro}

In this subsection, we briefly review the general repro samples framework for statistical inference proposed by \cite{xie2022repro}. This artificial-sample-based method 
can be applied to
construct confidence regions for a variety of parameters that take values in either continuum or discrete sets. Assume we observe $n$ samples $\bm{y}^{obs}=\{y_1^{obs},\ldots,y_n^{obs}\}$ from the population $\bm{Y}=G(\bm{U},\bm{\theta}_0)$, where $\bm{U}\in\mathcal{U}$ is a random vector from a known distribution $P_U$, $\bm{\theta}_0\in\Theta$ is the unknown model parameter and $G: \mathcal{U}\times\Theta\rightarrow \R^n$ is a known mapping. The observed data $\bm{y}^{obs}$ satisfies $\bm{y}^{obs}=G(\bm{u}^{rel},\bm{\theta}_0)$ with $\bm{u}^{rel}\in\mathcal{U}$ is a realization of the random vector $\bm{U}$. 

The repro samples method makes inference for the parameter $\bm{\theta}_0$ by mimicking the model generating process. Intuitively, if we have observed $\bm{u}^{rel}$, then for any parameter $\bm{\theta}\in\Theta$, we generate an artificial data $\bm{y}'=G(\bm{u}^{rel},\bm{\theta})$. If the artificial data matches the observed samples, i.e., $\bm{y}' =\bm{y}^{obs}$, then $\bm{\theta}$ is a potential value of $\bm{\theta}_0$ 
and, if there is any ambiguity, it
comes only from the inversion of $G(\bm{u}^{rel},\cdot)$. However, the data generating noises $\bm{u}^{rel}$ is unobserved, we need to incorporate its uncertainty for which we do it by considering a Borel set $B_\alpha$ with $\Prob(\bm{U}\in B_\alpha)\ge\alpha$. For any $\bm{u}^*\in B_\alpha$ and $\bm{\theta}\in\Theta$, we create an artificial data $\bm{y}^*=G(\bm{u}^{*},\bm{\theta})$ called repro sample. We keep $\bm{\theta}$ as a potential value of $\bm{\theta}_0$ if $\bm{y}^* = \bm{y}^{obs}$. All the retained values of $\bm{\theta}$ form a level-$\alpha$ confidence set for $\bm{\theta}_0$. Therefore the total uncertainty of the confidence region comes from both the possible ambiguity of the inversion of $G(\bm{u}^{rel},\cdot)$ and the uncertainty of the unobservability of $\bm{u}^{rel}$. Note that throughout the paper, we use $\alpha$ instead of $1-\alpha$ to denote the confidence level. For example, $\alpha = .90, .95$, or $.99$.

More generally, we consider a Borel set $B_\alpha(\bm{\theta})$ with $\Prob(T(\bm{U},\bm{\theta})\in B_\alpha(\bm{\theta}))\ge\alpha$. Then define the confidence region of $\bm{\theta}_0$ as
\[\Gamma_\alpha^{\bm{\theta}_0}(\bm{y}^{obs})=\{\bm{\theta}:\exists \bm{u}^*~{\rm s.t.}~\bm{y}^{obs}=G(\bm{u}^*,\bm{\theta}), T(\bm{u}^*,\bm{\theta})\in B_\alpha(\bm{\theta})\}.\]
It follows
\[\Prob(\bm{\theta}_0\in\Gamma_\alpha^{\bm{\theta}_0}(\bm{Y}))\ge\Prob\big(T(\bm{U},\bm{\theta}_0)\in B_\alpha(\bm{\theta}_0)\big)\ge\alpha.\]
Here $T:\mathcal{U}\times\Theta\rightarrow \R^d$ is called the nuclear mapping. Clearly, there might be multiple choices of $T$ that all lead to valid confidence regions. One choice is $T(\bm{u},\bm{\theta})=\bm{u}$ for any $\bm{\theta}\in\Theta$ and $B_\alpha(\bm{\theta})=D_\alpha$ is a level-$\alpha$ Borel set of $P_U$ with $\Prob(\bm{U}\in D_\alpha)\ge\alpha$. However, this naive nuclear statistic could lead to an oversized confidence region. Therefore $T$ is similar to a test statistic under the hypothesis testing framework and should be designed properly, see \cite{xie2022repro} for more details. Note that if $T$ depends on $\bm{u}^*$ through $G(\bm{u}^*,\bm{\theta})$, i.e., $T(\bm{u}^*,\bm{\theta})=\tilde T(G(\bm{u}^*,\bm{\theta}), \bm{\theta})$ for some $\tilde T$, then $\Gamma_\alpha^{\bm{\theta}_0}$ can be equivalently expressed as
\begin{equation}\label{eq_nuclear_test}
    \begin{aligned}
    \Gamma_{\alpha}^{\bm{\theta}_0}(\bm{y}^{obs})=&\{\bm{\theta}:\exists\bm{u}^*{\rm~s.t.~}\bm{y}^{obs}=G(\bm{u}^*,\bm{\theta}),\tilde T(\bm{y}^{obs},\bm{\theta})\in B_\alpha(\bm{\theta})\}\\
    \subseteq&\{\bm{\theta}:\tilde T(\bm{y}^{obs},\bm{\theta})\in B_\alpha(\bm{\theta})\}
    = \tilde\Gamma_\alpha^{\bm{\theta}_0}(\bm{y}^{obs}).
\end{aligned}
\end{equation}
Specifically, if $\tilde T$ is a test statistic under the Neyman-Pearson framework, by the property of test duality, $\tilde\Gamma_\alpha^{\bm{\theta}_0}(\bm{y}^{obs})$ is a level-$\alpha$ confidence set and the confidence set $\Gamma_\alpha^{\bm{\theta}_0}(\bm{y}^{obs})$ constructed by repro samples method becomes a subset of $\tilde\Gamma_\alpha^{\bm{\theta}_0}(\bm{y}^{obs})$.

In cases when nuisance parameters are present, \cite{xie2022repro} proposes a nuclear mapping function to make inferences for the parameters of interest. However, this approach encounters specific challenges when applied to high-dimensional logistic regression models. First, when making inference for $\tau_0,$ the regression coefficients $\*\beta_0$ are unknown nuisance parameters. Since it is infeasible to sample from the conditional distribution given the sufficient statistics, we would need to tackle the computational challenge by designing a nuclear mapping that can profile out all possible values of $\*\beta_0.$ On the other hand, when making inferences for the regression coefficients, we would need to profile out $\tau_0,$ this would require a viable model candidate set. Unlike the linear model setting in \cite{wang2022finite}, in the logistic regression framework,  multiple values of $\bm\theta$ may satisfy \eqref{eq_model} given $\bm y^{obs}$ and $\bm \epsilon^{rel}. $ This aspect significantly complicates the task of establishing a candidate set, both from computational and theoretical standpoints. See Section~\ref{sec_method} for a detailed explanation of the strategies to address the above challenges.


\section{Method and Theory}\label{sec_method}

This section presents the methods and theoretical results for statistical inference in high-dimensional logistic regression models.
Specifically, in Section~\ref{sec_candidate}, we construct the model candidate set to reduce the computation cost, by selecting a subset of the model space for inference.
In Section \ref{sec_tau}, we prune the model candidate set to get a model confidence set.  
During both of these subsections, the parameter of interest is $\tau_0$, while $\bm{\beta}_0$ is a nuisance parameter. Subsequently, in Section \ref{sec_Abeta}, we delve into techniques for inferring $q$ linear combinations of $\bm{\beta}_0$, denoted as $A\bm{\beta}_0$, utilizing any predetermined constant matrix $A \in \R^{q \times p}$ where $1 \leq q \leq p$. 

\subsection{Model candidate set}\label{sec_candidate}

As mentioned in Section \ref{sec_repro}, we need a Borel set $B_\alpha(\bm{\theta})$ for $\bm{\theta}=(\tau,\bm{\beta}_\tau)$ to incorporate the uncertainty of $\bm{\epsilon}^{rel}$. We will see in later sections that if we fix a model $\tau$, the set $B_\alpha(\bm{\theta})$ can be easily constructed for any $\bm{\beta}_\tau$. However, we still need to search over all model $\tau$'s, which can be computationally expensive. Therefore we introduce the notion of model candidate set to constrain the potential values of $\tau_0$ to only a small set of models without loss of much confidence and propose an efficient procedure for constructing such a candidate set. 

To demonstrate our construction of the model candidate set, we start from the oracle scenario when $\bm{\epsilon}^{rel}$ is known. With this oracle data, we show that $\tau_0$ can be identified under a weak signal strength assumption. However, the noise $\bm{\epsilon}^{rel}$ is not observable in practice, so we randomly generate $d$ noises $\{\bm{\epsilon}^{*(j)}: j\in[d]\}$ to approximate $\bm{\epsilon}^{rel}$. For each $\bm{\epsilon}^{*(j)}$, we construct an estimator $\hat\tau(\bm{\epsilon}^{*(j)})$ of $\tau_0$, then we collect these $d$ estimators to form the model candidate set $\mathcal{C}=\{\hat\tau(\bm{\epsilon}^{*(j)}):j\in[d]\}$.

Since the estimators $\hat\tau(\bm{\epsilon}^{*(j)})$ are constructed using empirical risk minimization, we define the following notations at first. Given any noises $\tilde{\bm{\epsilon}}=\{\tilde\epsilon_{i}:i\in[n]\}$, if we apply the following function $(X,\tilde\epsilon)\mapsto\1\{X_\tau^\top\bm{\beta}_\tau+\sigma\tilde\epsilon>0\}$ to $(\bm{X}^{obs},\tilde{\bm{\epsilon}})$ to predict $\bm{y}^{obs}$, we denote the empirical prediction error as
\begin{align*}
L_n(\tau,\bm{\beta}_\tau,\sigma|\bm{X}^{obs},\bm{y}^{obs},\tilde{\bm{\epsilon}}) 
=&\dfrac{1}{n}\sum_{i=1}^n\mathbbm{1}\big\{y_i^{obs}\ne\mathbbm{1}\{X_{i,\tau}^{obs\top}\bm{\beta}_{\tau}+\sigma\tilde\epsilon_{i}>0\}\big\} \\
=&\dfrac{1}{n}\sum_{i=1}^n\mathbbm{1}\big\{\1\{X_{i,\tau_0}^{obs\top}\bm{\beta}_{0,\tau_0}+\epsilon_{i}^{rel} > 0\}\ne\mathbbm{1}\{X_{i,\tau}^{obs\top}\bm{\beta}_{\tau}+\sigma\tilde\epsilon_{i}>0\}\big\}.
\end{align*}
In this work, we will choose $\tilde{\bm{\epsilon}}$ to be either $\bm{\epsilon}^{rel}$ or an artificially generated $\bm{\epsilon}^*$ that is independent of the oracle data $(\bm{X}^{obs},\bm{y}^{obs},\bm{\epsilon}^{rel})$. Now we define the expected prediction error using these two choices of $\tilde{\bm{\epsilon}}$ respectively. For a random copy $(\bm{X},\bm{y},\bm{\epsilon})$ of the oracle data, if we choose $\tilde {\bm{\epsilon}}=\bm{\epsilon}$, the expected prediction error is denoted as
\begin{align*}
    L_{\bm{\theta_0}}(\tau,\bm{\beta}_\tau,\sigma)
    = &\E  L_n(\tau,\bm{\beta}_\tau,\sigma|\bm{X},\bm{y},\bm{\epsilon})
    = 
 \Prob\big(\1\{X_{\tau_0}^\top\bm{\beta}_{0,\tau_0}+\epsilon>0\}\ne\1\{X_\tau^\top\bm{\beta}_\tau+\sigma\epsilon>0\}\big),
\end{align*}
where the expectation $\E$ is over the randomness of $\bm{X}$, $\bm{\epsilon}$ and $\bm{y}$ (or equivalently $\bm{X}$ and $\bm{\epsilon}$).
When we set $\tilde{\bm{\epsilon}}=\bm{\epsilon}^*$ which is independent of $(\bm{X},\bm{y},\bm{\epsilon})$, we denote the expected prediction error as
\begin{align*}
    L_{\bm{\theta}_0}^*(\tau,\bm{\beta}_\tau,\sigma)
    =&\E L_n(\tau,\bm{\beta}_\tau,\sigma|\bm{X},\bm{y},\bm{\epsilon}^*)
    =\Prob\big(\1\{X_{\tau_0}^\top\bm{\beta}_{0,\tau_0}+\epsilon>0\}\ne\1\{X_\tau^\top\bm{\beta}_\tau+\sigma\epsilon^*>0\}\big).
\end{align*}
Here the expectation $\E$ above is over the randomness of $\bm{X},\bm{y}$ and ${\bm{\epsilon}^*}$ (or $\bm{X},\bm{\epsilon}$ and ${\bm{\epsilon}^*}$).

\subsubsection{Signal strength condition and recovery under oracle setting}\label{sec_candidate_oracle}

As outlined in Section~\ref{sec_candidate}. Our intuition for constructing the model candidate set contains two stages. At first, we show $\tau_0$ can be recovered given $\bm{\epsilon}^{rel}$. Then we generate independent copies of $\bm{\epsilon}$ to approximate $\bm{\epsilon}^{rel}$. This subsection considers the first stage, investigating the sufficient condition for recovering $\tau_0$ given the knowledge of $\bm{\epsilon}^{rel}$. Then we will show in Section~\ref{sec_candidate_practice} that under this sufficient condition, the recovery of $\tau_0$ still holds providing $\bm{\epsilon}^{rel}$ aligns with some of the generated synthetic noises. 

Note that $L_{\bm{\theta}_0}(\tau,\bm{\beta}_{\tau},\sigma)$ attains its minimum value 0 at $(\tau_0,\bm{\beta}_{0,\tau_0},1)$. Therefore, supposing $\bm{\epsilon}^{rel}$ is known, we could estimate $\tau_0$ by minimizing $L_n(\tau,\bm{\beta}_\tau,\sigma|\bm{X}^{obs},\bm{y}^{obs},\bm{\epsilon}^{rel})$. However, when $\bm{\beta}_{0,\tau_0}$ has weak signals, excluding those weak signals from $\tau_0$ may not increase $L_{\bm{\theta}_0}$ substantially. Consequently, the minimizer of $L_n$ may differ from $\tau_0$, making it hard to identify the true model using the oracle data $(\bm{X}^{obs},\bm{y}^{obs},\bm{\epsilon}^{rel})$. Therefore, 
to identify the true model $\tau_0$, we need the following assumption on the signal strength to separate $\tau_0$ from all other wrong models using data.

\begin{Assumption}\label{ass_signal}
For all $\tau\subset[p]$ with $\abs{\tau}\le s,\tau\ne\tau_0$,
\begin{equation}\label{eq_signal}
    \begin{aligned}
        &\inf_{\bm{\beta}_\tau\in\R^{\abs{\tau}},\sigma\ge 0}L_{\bm{\theta}_0}(\tau,\bm{\beta}_\tau,\sigma)
        \gtrsim &\min\bigg\{\abs{\tau_0\setminus\tau}\dfrac{\log p}{n}+(\abs{\tau}+1)\dfrac{\log \frac{n}{\abs{\tau}+1}}{n},(\abs{\tau}+1)\dfrac{\log p}{n}\bigg\}.
    \end{aligned}
\end{equation}
\end{Assumption}

Note that when $\bm{\epsilon}^{rel}$ is known, the prediction error under the true parameter is 0, $L_{\bm{\theta}_0}(\tau_0,\bm{\beta}_{0,\tau_0},1)=0$. Then Assumption \ref{ass_signal} links model selection to prediction in the sense that at the population level, any wrong model $\abs{\tau}\le s,\tau\ne\tau_0$ has a positive prediction error gap from the true model. As we will show in Remark \ref{rem_cmin} and \ref{rem_betamin} later, this assumption is weaker than other commonly used signal strength conditions in the literature.

Under Assumption \ref{ass_signal}, all the wrong models have a relatively large prediction error while the true model $\tau_0$ has a prediction error equal to 0, therefore if we solve the constrained empirical risk minimization problem
\begin{equation}\label{eq_oracle}
    \hat \tau(\bm{\epsilon}^{rel})=\argmin_{|\tau|\le s}\min_{\bm{\beta}\in\R^p,\sigma\ge 0}L_n(\tau,\bm{\beta}_\tau,\sigma|\bm{X}^{obs},\bm{y}^{obs},\bm{\epsilon}^{rel}),
\end{equation}
$\hat\tau(\bm{\epsilon}^{rel})$ is likely to equal to $\tau_0$. Formally, we have the following Lemma \ref{lem_oracle} which states that as long as Assumption \ref{ass_signal} is satisfied, we can identify $\tau_0$ using $(\bm{X}^{obs},\bm{y}^{obs},\bm{\epsilon}^{rel})$ in probability. A proof is given in the Appendix. In Lemma~\ref{lem_oracle}, we denote
$\hat \tau(\bm{\epsilon}) = \argmin_{|\tau|\le s}\min_{\bm{\beta}\in\R^p,\sigma\ge 0}L_n(\tau,\bm{\beta}_\tau,\sigma|\bm{X},\bm{y},\bm{\epsilon})$ to be a random copy of $\hat\tau(\bm{\epsilon}^{rel})$.

\begin{Lemma}\label{lem_oracle}

For $\hat\tau$ defined in Equation \eqref{eq_oracle}, denote
\[\tilde c_{\min}=\min_{|\tau|\le s,\tau\ne\tau_0,\bm{\beta}_\tau\in\R^{\abs{\tau}},\sigma\ge 0}\dfrac{L_{\bm{\theta}_0}(\tau,\bm{\beta}_\tau,\sigma)-\frac{2\abs{\tau}+2}{n}\log_2\frac{2en}{\abs{\tau}+1}}{\abs{\tau_0\setminus\tau}},\]
\[c_{\min}=\min_{|\tau|\le s,\tau\ne\tau_0,\bm{\beta}_\tau\in\R^{\abs{\tau}},\sigma\ge 0}\dfrac{L_{\bm{\theta}_0}(\tau,\bm{\beta}_\tau,\sigma)-\frac{2\abs{\tau}+2}{n}\log_2\frac{2en}{\abs{\tau}+1}}{\abs{\tau}\vee 1},\]
then

\[\Prob(\hat\tau(\bm{\epsilon})\ne\tau_0)\lesssim 2^{-\frac{1}{2}n\tilde c_{\min}+2\log_2 p}\wedge 2^{-\frac{1}{2}nc_{\min}+\log_2 p}.\]
Here the probability is taken with respect to $(\bm{X}, \bm{y}, \bm{\epsilon})$. Furthermore, if Assumption \ref{ass_signal} holds, 
\[\Prob(\hat\tau(\bm{\epsilon})\ne\tau_0)\lesssim 2^{-cn\tilde c_{\min}}\wedge 2^{-cnc_{\min}}.\]
\end{Lemma}

\begin{Remark}\label{rem_cmin}
    Assumption \ref{ass_signal} can be shown to be weaker than the $C_{\min}$ condition in \cite{shen2012likelihood}. Note that the $C_{\min}$ condition requires
    \[\inf_{\abs{\tau}\le s,\tau\ne\tau_0,\bm{\beta}_\tau\in\R^{\abs{\tau}}}\frac{[H(\Prob_{\bm{\theta}_0},\Prob_{(\tau,\bm{\beta}_\tau)})]^2}{\abs{\tau_0\setminus\tau}}\gtrsim\frac{\log p}{n},\]
    where $\Prob_{(\tau,\bm{\beta}_\tau)}$ is the joint distribution of $(X,Y)$ with $X\sim\Prob_X$, $\epsilon\sim\text{Logistic}$ and $Y=\1\{X_{\tau}^\top\bm{\beta}_\tau+\epsilon>0\}$, $H(\Prob_1,\Prob_2)$ is the Hellinger distance between $\Prob_1,\Prob_2$. However as we will show in Lemma \ref{lem_cmin} of Section \ref{sec_proof}, when $\sigma>0$, 
    \[L_{\bm{\theta}_0}(\tau,\bm{\beta}_\tau,\sigma)={\rm TV}(\Prob_{\bm{\theta}_0},\Prob_{(\tau,\frac{\bm{\beta}_\tau}{\sigma})}),\]
    where ${\rm TV}(\Prob_1,\Prob_2)=\sup_A\abs{\Prob_1(A)-\Prob_2(A)}$ is the total variation distance between $\Prob_1,\Prob_2$. If for any $\tau\subset[p]$ with $\abs{\tau}\le s,\tau\ne\tau_0$, the minimizer $(\bm{\beta}_\tau,\sigma)$ of Equation \eqref{eq_signal} satisfies $\sigma>0$, and if we further assume $s\log\frac{n}{s}\lesssim \log p$, then a sufficient condition for Assumption \ref{ass_signal} is
    \[\inf_{\abs{\tau}\le s,\tau\ne\tau_0,\bm{\beta}_\tau\in\R^{\abs{\tau}}}\frac{{\rm TV}(\Prob_{\bm{\theta}_0},\Prob_{(\tau,\bm{\beta}_\tau)})}{\abs{\tau_0\setminus\tau}}\gtrsim\frac{\log p}{n}.\]
    Since $\{H(\Prob_1,\Prob_2)\}^2\lesssim{\rm TV}(\Prob_1,\Prob_2)\lesssim H(\Prob_1,\Prob_2)$, Assumption \ref{ass_signal} is weaker than the $C_{\min}$ condition in \cite{shen2012likelihood}.
\end{Remark}
\begin{Remark}\label{rem_betamin}
  Assumption \ref{ass_signal} is also weaker than the commonly used $\beta$-min condition \citep{bunea2008honest,zhang2010nearly,zhao2006model}. Denote $\beta_{\min}=\min_{j\in\tau_0}\abs{\beta_{0,j}}$, then the $\beta$-min condition assumes
    \[\beta_{\min}\gtrsim\sqrt{\frac{\log p}{n}}.\]
    As we will show in Lemma \ref{lem_betamin} of Section \ref{sec_proof}, suppose $X$ is sub-Gaussian with $\norm{X}_{\psi_2}\lesssim 1$, $\norm{\bm{\beta}_0}_2\lesssim 1$ and $\lambda_{\min}(\E XX^\top)\gtrsim 1$,
    then 
    \[\inf_{\abs{\tau}\le s,\tau\ne\tau_0,\bm{\beta}_\tau\in\R^{\abs{\tau}}}\frac{{\rm TV}(\Prob_{\bm{\theta}_0},\Prob_{(\tau,\bm{\beta}_\tau)})}{\sqrt{\abs{\tau_0\setminus\tau}}}\gtrsim\beta_{\min}.\]
    Therefore another sufficient condition for Assumption \ref{ass_signal} is
    $\beta_{\min}\gtrsim\frac{\sqrt{s}\log p}{n}+\frac{s\log \frac{n}{s}}{n}.$ 
    When $\frac{s\log p}{n}+\frac{s^2\log^2\frac{n}{s}}{n\log p}\lesssim 1$, we have Assumption \ref{ass_signal} is weaker than the $\beta$-min condition.
\end{Remark}

\subsubsection{Candidate set construction in the practical setting}\label{sec_candidate_practice}

In practice, although the oracle noise $\bm{\epsilon}^{rel}$ is unobservable, we can generate a vector $\bm{\epsilon}^*$ independently from logistic distribution and calculate $\hat \tau(\bm{\epsilon}^*)$ as
\[\hat \tau(\bm{\epsilon}^*)=\argmin_{|\tau|\le k}\min_{\bm{\beta}\in\R^p,\sigma\ge 0}L_n(\tau,\bm{\beta}_\tau,\sigma|\bm{X}^{obs},\bm{y}^{obs},\bm{\epsilon}^*).\]
We expect that as long as $\bm{\epsilon}^*$ and $\bm{\epsilon}^{rel}$ are close enough, we would have $\hat\tau(\bm{\epsilon}^*)=\hat\tau(\bm{\epsilon}^{rel})$. Therefore, we generate $d$ i.i.d. random noises $\{\bm{\epsilon}^{*(j)}:j\in[d]\}$ and calculate their corresponding $\hat\tau(\bm{\epsilon}^{*(j)})$. Then we collect all the estimated models into the model candidate set $\mathcal{C}$ as
\[\mathcal{C}=\{\hat\tau(\bm{\epsilon}^{*(j)}):\epsilon_i^{*(j)}\overset{{\rm i.i.d.}}{\sim}{\rm Logistic},i\in[n],j\in[d]\}.\] 
We summarize the above procedure in Algorithm \ref{alg_candidate}.

\begin{algorithm}
\caption{Model Candidate Set}\label{alg_candidate}
\begin{algorithmic}[1]
\State{\bf Input:} Observed data $(\bm{X}^{obs},\bm{y}^{obs})$, the number of repro samples $d$. 
\State{\bf Output:} Model candidate set $\mathcal{C}$.
\State Generate $d$ copies of logistic random noises $\{\bm{\epsilon}^{*(j)}:\epsilon_i^{*(j)}\overset{{\rm i.i.d.}}{\sim}{\rm Logistic},i\in[n],j\in[d]\}$.
\State Compute $\hat\tau(\bm{\epsilon}^{*(j)})=\argmin_{|\tau|\le k}\min_{\bm{\beta}\in\R^p,\sigma\ge 0}L_n(\tau,\bm{\beta}_\tau,\sigma|\bm{X}^{obs},\bm{y}^{obs},\bm{\epsilon}^{*(j)}),$ for $j\in[d]$ and some $k$. 
\State Construct $\mathcal{C}=\{\hat\tau(\bm{\epsilon}^{*(j)}):j\in[d]\}$.
\end{algorithmic}
\end{algorithm}

\begin{Remark}[Practical implementation of Algorithm \ref{alg_candidate}]
    Step 2 in Algorithm \ref{alg_candidate} involves optimization for 0-1 loss function with $\ell_0$ constraint which can be hard to calculate. In practice, we use the hinge loss or logistic loss as surrogates for the 0-1 loss, then replace the $\ell_0$ constraint by the adaptive LASSO penalty and apply the extended BIC \citep{chen2008extended} to choose the tuning parameter for the penalty.
\end{Remark}

In the following theorem, we show that as long as the number of Monte Carlo copies, $d$, is large enough, there will be at least one $\bm{\epsilon}^{*(j)}$ that is closed to $\bm{\epsilon}^{rel}$, then the model candidate set $\mathcal{C}$ contains the true model $\tau_0$ with high probability. A proof is given in the Appendix.

\begin{Theorem}\label{thm_candidate}
Using the same notation  as in Lemma \ref{lem_oracle}, if we further denote $F_\epsilon(z)=(1+e^{-z})^{-1}$ to be the CDF of logistic distribution, we have
\begin{align*}
    \Prob(\tau_0\not\in\mathcal{C})\lesssim 2^{-\frac{1}{2}n\tilde c_{\min}+2\log_2p}\wedge2^{-\frac{1}{2}nc_{
    \min}+\log_2p}+(1-\{\E \big|F_\epsilon(\epsilon)-F_\epsilon(-X^\top_{\tau_0}\bm{\beta}_{0,\tau_0})\big|\}^n)^d.
\end{align*}
If Assumption \ref{ass_signal} holds, for any fixed $n$, when $d$ is large enough such that
\[(1-\{\E \big|F_\epsilon(\epsilon)-F_\epsilon(-X^\top_{\tau_0}\bm{\beta}_{0,\tau_0})\big|\}^n)^d\lesssim 2^{-cn\tilde c_{\min}}\wedge 2^{-cnc_{\min}},\]
we have
\[\Prob(\tau_0\not\in\mathcal{C})\lesssim 2^{-cn\tilde c_{\min}}\wedge 2^{-cnc_{\min}}.\]
\end{Theorem}
Theorem \ref{thm_candidate} ensures the inclusion of $\tau_0$ in $\mathcal{C}$ as long as Assumption \ref{ass_signal} is satisfied and $d$ is large enough.

Next, we demonstrate that under a stronger signal strength condition, the requirement for the number of repro samples, $d$, can be relaxed.

\begin{Assumption}\label{ass_signal_strong}
For all $\tau$ with $\abs{\tau}\le s,\tau\ne\tau_0$,
\begin{align*}
    \min_{\bm{\beta}_\tau\in\R^{\abs{\tau}},\sigma\ge 0}L_{\bm{\theta}_0}^*(\tau,\bm{\beta}_\tau,\sigma)-L_{\bm{\theta}_0}^*(\tau_0,\bm{\beta}_{0,\tau_0},0)\gtrsim\sqrt{\frac{\abs{\tau}\vee 1}{n}}+\sqrt{\abs{\tau_0\setminus\tau}\wedge(\abs{\tau}\vee1)}\sqrt{\frac{\log p}{n}}.
\end{align*}
\end{Assumption}

Since $\bm{\epsilon}^*$ is independent of $(\bm{X},\bm{y})$, we know that $(\tau,\bm{\beta}_{\tau},\sigma)=(\tau_0,\bm{\beta}_{0,\tau_0},0)$ minimizes the expected prediction error $L_{\bm{\theta}_0}^*(\tau,\bm{\beta}_\tau,\sigma)$. Thus Assumption \ref{ass_signal_strong} assumes that all the wrong models $\tau\ne\tau_0$ with $|\tau|\le s$ have a positive error gap from the true model. Compared to Assumption \ref{ass_signal}, the signal strength in Assumption \ref{ass_signal_strong} scales with $\frac{1}{\sqrt{n}}$ instead of $\frac{1}{n}$ as in Assumption \ref{ass_signal}.

As we will show in the following theorem, if the stronger signal strength Assumption~\ref{ass_signal_strong} holds, then similar to the model selection consistency \citep{zhao2006model,zhang2010nearly,bunea2008honest}, the model candidate set contains the true model support with high probability for any $d \ge 1$. A proof is provided in the Appendix.

\begin{Theorem}\label{thm_candidate_strong_signal}

Denote
\[\tilde c_{\min}^*=\bigg(\inf_{|\tau|\le s,\tau\ne\tau_0,\bm{\beta}_\tau\in\R^{\abs{\tau}},\sigma\ge 0}\dfrac{L_{\bm{\theta}_0}^*(\tau,\bm{\beta}_\tau,\sigma)-L_{\bm{\theta}_0}^*(\tau_0,\bm{\beta}_{0,\tau_0},0)-c\sqrt{\frac{\abs{\tau}+1}{n}}}{\sqrt{\abs{\tau_0\setminus\tau}}}\bigg)^2,\]
\[c_{\min}^*=\bigg(\inf_{|\tau|\le s,\tau\ne\tau_0,\bm{\beta}_\tau\in\R^{\abs{\tau}},\sigma\ge 0}\dfrac{L_{\bm{\theta}_0}^*(\tau,\bm{\beta}_\tau,\sigma)-L_{\bm{\theta}_0}^*(\tau_0,\bm{\beta}_{0,\tau_0},0)-c\sqrt{\frac{\abs{\tau}+1}{n}}}{\sqrt{\abs{\tau}\vee 1}}\bigg)^2.\]
For any $n$ and $d$, the model candidate set satisfies,
\[\Prob(\tau_0\not\in\mathcal{C})\lesssim e^{-\frac{n}{8}\tilde c_{\min}^*+2\log p}\wedge e^{-\frac{n}{8}nc_{\min}^*+\log p}.\]
If Assumption \ref{ass_signal_strong} holds, then
\[\Prob(\tau_0\not\in\mathcal{C})\lesssim e^{-cn\tilde c_{\min}^*}\wedge e^{-cnc_{\min}^*}.\]

\end{Theorem}

\begin{Remark}

    \begin{itemize}
        \item[(1)] Besides the coverage for $\tau_0$, we can also guarantee the consistency of $\mathcal{C}$. Specifically, under Assumption \ref{ass_signal_strong}, using the same notation as in Theorem \ref{thm_candidate_strong_signal}, if we set $d$ such that $\log d\lesssim \log p$, then with high probability, we have $\mathcal{C}$ contains only $\tau_0$,
        \[\Prob(\mathcal{C}\ne\{\tau_0\})\lesssim e^{-cn\tilde c_{\min}^*}\wedge e^{-cnc_{\min}^*}.\]
        Note that to conduct inference for $\tau_0$ and $\bm{\beta}_0$, it is only necessary that $\tau_0\in\mathcal{C}$, but $\mathcal{C}=\{\tau_0\}$ is not required. Therefore, we can set $d$ as large as necessary.
        \item[(2)] Combining Theorem \ref{thm_candidate} and \ref{thm_candidate_strong_signal}, it becomes evident that the model candidate set $\mathcal{C}$ is adaptive to the signal strength. Under the weak signal strength Assumption \ref{ass_signal}, as we discussed in Remark \ref{rem_cmin} and \ref{rem_betamin}, none of the existing work can be guaranteed to find $\tau_0$, but our approach assures $\tau_0\in\mathcal{C}$ as long as $d$ is large enough. Furthermore, if the stronger signal strength Assumption \ref{ass_signal_strong} is satisfied, then $d$ doesn't need to be large at all, since $\tau_0\in\mathcal{C}$ holds for any $d\ge 1$. Moreover, under Assumption \ref{ass_signal_strong}, if $d$ is not too large such that $\log d\lesssim\log p$, it is ensured that $\mathcal{C}=\{\tau_0\}$.
    \end{itemize}

\end{Remark}

\subsection{Inference for \texorpdfstring{$\tau_0$}{tau\_0}}\label{sec_tau}

If we are interested in the inference for the true model $\tau_0$, then $\bm{\beta}_{0,\tau_0}$ is 
a nuisance parameter. As we discussed in Section \ref{sec_repro} Equation \eqref{eq_nuclear_test}, if the nuclear statistic has the form $T(\bm{X}^{obs},\bm{\epsilon}^*,\bm{\theta})=\tilde T(\bm{X}^{obs},\bm{Y}^*,\bm{\theta})$ where $\bm{Y}^*$ is generated by $\bm{X}^{obs},\bm{\epsilon}^*$ and $\bm{\theta}$, then it suffices to check whether $\tilde T(\bm{X}^{obs},\bm{y}^{obs},\bm{\theta})$ is in $B_{\alpha}(\bm{\theta})$. In order to deal with the nuisance parameter, we consider the following form of confidence set for $\tau_0$,
\begin{align*}
    \Gamma_\alpha^{\tau_0}(\bm{X}^{obs},\bm{y}^{obs})=&\{\tau:\exists\bm{\beta}_\tau\in\R^{|\tau|}{\rm~s.t.~}\tilde T(\bm{X}^{obs},\bm{y}^{obs},(\tau,\bm{\beta}_\tau))\in B_\alpha((\tau,\bm{\beta}_\tau))\},
\end{align*} 
with $B_\alpha(\bm{\theta})$ satisfies
$\Prob(\tilde T(\bm{X},\bm{Y}^*,\bm{\theta})\in B_\alpha(\bm{\theta}))\ge\alpha.$

If $1-\tilde T(\bm{X},\bm{Y}^*,\bm{\theta})$ is a $p$-value, then we can take $B_\alpha(\bm{\theta})=(-\infty,\alpha)$ and rewrite $\Gamma^{\tau_0}_\alpha(\bm{X}^{obs},\bm{y}^{obs})$ as
\begin{equation}\label{eq_model_conf}
    \Gamma^{\tau_0}_\alpha(\bm{X}^{obs},\bm{y}^{obs})=\{\tau:\min_{\bm{\beta}_\tau\in\R^{|\tau|}}\tilde T(\bm{X}^{obs},\bm{y}^{obs},(\tau,\bm{\beta}_\tau))<\alpha\}.
\end{equation}
Here, we refer to
$\min_{\bm{\beta}_\tau\in\R^{|\tau|}}\tilde T(\bm{X}^{obs},\bm{Y}^*,(\tau,\bm{\beta}_\tau))$ as a {\it profile nuclear statistic}.

Specifically, we construct the nuclear statistic $\tilde T$ and the model confidence sets as follows. For any given $\bm{\theta}=(\tau,\bm{\beta}_\tau)$ and $\bm{Y}^*\in\{0,1\}^n$ generated by $Y_i^*=\mathbbm{1}\{X_{i,\tau}^{obs\top}\bm{\beta}_{\tau}+\epsilon_i^*> 0\}$ with $\epsilon_i^*\overset{i.i.d.}{\sim}{\rm Logistic}$, we solve
\begin{equation}\label{eq_synthetic_lasso}
    \tilde{\bm{\beta}}(\lambda)=\argmin_{\bm{\beta}\in\R^p}\frac{1}{n}\sum_{i=1}^n\log(1+e^{-(2Y_i^*-1)X_i^{obs\top}\bm{\beta}})+\lambda\norm{\bm{\beta}}_1,
\end{equation}
\[\tilde\lambda(\tau,\bm{\beta}_\tau)=\argmax_{\lambda\ge0}\|\tilde{\bm{\beta}}(\lambda)\|_0,\quad{\rm s.t. }\norm{\tilde{\bm{\beta}}(\lambda)}_0\le\abs{\tau},\]
\[\tilde\tau(\bm{X}^{obs},\bm{Y}^*,\bm{\theta})=\supp(\tilde{\bm{\beta}}(\tilde\lambda(\bm{\theta}))).\]
The model selector $\tilde\tau(\bm{X}^{obs},\bm{Y}^*,\bm{\theta})$ is the largest model with cardinality at most $|\tau|$ in the solution path of Problem \eqref{eq_synthetic_lasso} using the synthetic data $(\bm{X}^{obs},\bm{Y}^*)$. 
Denote
\[P_{\bm{\theta}}(\tau^*)=\Prob_{\bm{\epsilon}^*|\bm{\theta}}(\tilde \tau(\bm{X}^{obs},\bm{Y}^*,\bm{\theta})=\tau^*),\]
where $\Prob_{\bm{\epsilon}^*|\bm{\theta}}$ counts the randomness of $\bm{Y}^*$ given $\bm{X}^{obs}$. Then we consider the nuclear statistic
\[T(\bm{X}^{obs},\bm{\epsilon},\bm{\theta})=\tilde T(\bm{X}^{obs},\bm{y},\bm{\theta})=\Prob_{\bm{\epsilon^*}|\bm{\theta}}\big(P_{\bm{\theta}}(\tilde\tau(\bm{X}^{obs},\bm{Y^*},\bm{\theta}))>P_{\bm{\theta}}(\tilde\tau(\bm{X}^{obs},\bm{y},\bm{\theta}))\big)\]
which is the probability that $\tilde\tau(\bm{X}^{obs},\bm{y},\bm{\theta})$ appears less often than the synthetic model selector $\tilde\tau(\bm{X}^{obs},\bm{Y}^*,\bm{\theta})$ in $P_{\bm{\theta}}(\cdot)$. Since $\tilde T(\bm{X}^{obs},\bm{y},\bm{\theta})$ is also the survival function of random variable $P_{\bm{\theta}}(\tilde \tau(\bm{X}^{obs},\bm{Y}^*,\bm{\theta}))$ evaluated at $P_{\bm{\theta}}(\tilde\tau(\bm{X}^{obs},\bm{y},\bm{\theta}))$, when $\bm{\theta}=\bm{\theta}_0,\bm{y}=\1(\bm{X}^{obs}\bm{\beta}_0+\bm{\epsilon}>0)$, we know that $1-\tilde T(\bm{X}^{obs},\bm{y},\bm{\theta}_0)$ is a p-value with
\[\Prob_{\bm{\epsilon}}(\tilde T(\bm{X}^{obs},\bm{y},\bm{\theta}_0)<\alpha)\ge\alpha.\]
Here $\Prob_{\bm{\epsilon}}$ counts the randomness of $\bm{y}$ given $\bm{X}^{obs}$. Since $\tau_0$ belongs to $\mathcal{C}$ with high probability as guaranteed by Theorem \ref{thm_candidate} and \ref{thm_candidate_strong_signal}, we constrain the model confidence set to be a subset of $\mathcal{C}$. Then according to Equation \eqref{eq_model_conf}, we define the confidence set for $\tau_0$ as
\begin{align*}
    \Gamma_\alpha^{\tau_0}(\bm{X}^{obs},\bm{y}^{obs})=&\{\tau:\exists\bm{\beta}_\tau\in\R^{|\tau|}{\rm~s.t.~}\tilde T(\bm{X}^{obs},\bm{y}^{obs},(\tau,\bm{\beta}_\tau))<\alpha,\tau\in\mathcal{C}\}\\
    =&\{\tau:\min_{\bm{\beta}_\tau\in\R^{|\tau|}}\tilde T(\bm{X}^{obs},\bm{y}^{obs},(\tau,\bm{\beta}_\tau))<\alpha,\tau\in\mathcal{C}\}.
\end{align*}

Since we don't have an explicit expression for $P_{\bm{\theta}}(\tau)$, we apply the Monte Carlo method to approximate it. More specifically, we generate $\{\bm{\epsilon}^{*(j)}:j\in[m]\}$ with $\epsilon_i^{*(j)}\overset{i.i.d.}{\sim}{\rm Logistic}$ for $i\in[n], j\in[m]$, then generate $\{\bm{Y}^{*(j)}:j\in[m]\}$ by $Y_i^{*(j)}=\mathbbm{1}\{X_{i,\tau}^{obs\top}{\bm{\beta}}_\tau+\epsilon_i^{*(j)}> 0\}$. For each $\bm{Y}^{*(j)}$, we calculate the corresponding $\tilde\tau^{(j)}\overset{\triangle}{=}\tilde\tau(\bm{X}^{obs},\bm{Y}^{*(j)},\bm{\theta})$ and estimate $P_{\bm{\theta}}(\tau^*)$ by $\hat P_{\bm{\theta}}(\tau^*)=\frac{1}{m}\sum_{j=1}^m\mathbbm{1}\{\tilde \tau^{(j)}=\tau^*\}$. Denote the estimated profile nuclear statistic as
\[\hat T(\bm{X}^{obs},\bm{y},\tau)=\min_{{\bm{\beta}}_\tau\in\R^{\abs{\tau}}}\frac{\abs{\{j\in[m]:\hat P_{\tau,{\bm{\beta}}_\tau}(\tilde\tau^{(j)})>\hat P_{\tau,{\bm{\beta}}_\tau}(\tilde\tau(\bm{X}^{obs},\bm{y},\bm{\theta}))\}}}{m},\]
then the final confidence set for $\tau_0$ becomes
\[\hat\Gamma_\alpha^{\tau_0}(\bm{X}^{obs},\bm{y}^{obs})=\{\tau:\hat T(\bm{X}^{obs},\bm{y}^{obs},\tau)<\alpha,\tau\in\mathcal{C}\}.\]
We summarize the procedure in Algorithm \ref{alg_model_confidence}. 

\begin{algorithm}
\caption{Model Confidence Set}\label{alg_model_confidence}
\begin{algorithmic}[1]
\State{\bf Input:} Observed data $(\bm{X}^{obs},\bm{y}^{obs})$, model candidate set $\mathcal{C}$, the number of Monte Carlo samples $m$. 
\State{\bf Output:} Model confidence set $\hat\Gamma_\alpha^{\tau_0}(\bm{X}^{obs},\bm{y}^{obs})$.
\For{$\tau\in\mathcal{C}$}
\State Generate $m$ copies of logistic random noises $\{\bm{\epsilon}^{*(j)}:\epsilon_i^{*(j)}\overset{{\rm i.i.d.}}{\sim}{\rm Logistic},i\in[n],j\in[m]\}$.
\State For some ${\bm{\beta}}_\tau$ to be optimized later, compute $\{{\bm Y}^{*(j)}:j\in[m]\}$ with $Y_i^{*(j)}=\1\{X_{i,\tau}^{obs\top}{\bm{\beta}}_\tau+\epsilon_i^{*(j)}>0\}$.
\State For each ${\bm Y}^{*(j)},j\in[m]$, calculate
\[\tilde{\bm{\beta}}^{(j)}(\lambda)=\argmin_{{\bm{\beta}}\in\R^p}\frac{1}{n}\sum_{i=1}^n\log(1+e^{-(2Y_i^{*(j)}-1)X_i^{obs\top}{\bm{\beta}}})+\lambda\norm{{\bm{\beta}}}_1,\]
\[\tilde\tau^{(j)}=\supp(\tilde{\bm{\beta}}^{(j)}(\tilde\lambda^{(j)}(\tau,\bm{\beta}_\tau))),\qquad\tilde\lambda^{(j)}(\tau,\bm{\beta}_\tau)=\argmax_{\lambda\ge0}\norm{\tilde{\bm{\beta}}^{(j)}(\lambda)}_0\quad{\rm s.t. }\norm{\tilde{\bm{\beta}}^{(j)}(\lambda)}_0\le\abs{\tau},\]
and
\[\tilde{\bm{\beta}}(\lambda)=\argmin_{{\bm{\beta}}\in\R^p}\frac{1}{n}\sum_{i=1}^n\log(1+e^{-(2y_i^{obs}-1)X_i^{obs\top}{\bm{\beta}}})+\lambda\norm{{\bm{\beta}}}_1,\]
\[\tilde\tau(\bm{X}^{obs},\bm{y}^{obs},(\tau,\bm{\beta}_\tau))=\supp(\tilde{\bm{\beta}}(\tilde\lambda(\tau,\bm{\beta}_\tau))),\qquad\tilde\lambda(\tau,\bm{\beta}_\tau)=\argmax_{\lambda\ge0}\norm{\tilde{\bm{\beta}}(\lambda)}_0\quad{\rm s.t. }\norm{\tilde{\bm{\beta}}(\lambda)}_0\le\abs{\tau}.\]

\State Calculate
\[\hat T(\bm{X}^{obs},\bm{y}^{obs},\tau)=\min_{{\bm{\beta}}_\tau\in\R^{\abs{\tau}}}\frac{\abs{\{j\in[m]:\hat P_{\tau,{\bm{\beta}}_\tau}(\tilde\tau^{(j)})>\hat P_{\tau,{\bm{\beta}}_\tau}(\tilde\tau(\bm{X}^{obs},\bm{y}^{obs},(\tau,\bm{\beta}_\tau))\}}}{m},\]
with $\hat P_{\tau,{\bm{\beta}}_\tau}(\tau^*)=\frac{1}{m}\sum_{j=1}^m\mathbbm{1}\{\tilde \tau^{(j)}=\tau^*\}$.
\EndFor
\State Construct the model confidence set as
\[\hat\Gamma_\alpha^{\tau_0}(\bm{X}^{obs},\bm{y}^{obs})=\{\tau:\hat T(\bm{X}^{obs},\bm{y}^{obs},\tau)<\alpha,\tau\in\mathcal{C}\}.\]
\end{algorithmic}
\end{algorithm}
Now we formalize the intuition stated above as the following theorem, which guarantees the validity of $\hat\Gamma_\alpha^{\tau_0}(\bm{y}^{obs})$. A proof is given in the Appendix.

\begin{Theorem}\label{thm_tau}

\begin{itemize}
    \item[(1)] If Assumption \ref{ass_signal} holds, $d$ is large enough as required in Theorem \ref{thm_candidate} and $n$ is any fixed number, for $c_{\min},\tilde c_{\min}$ defined in Theorem \ref{thm_candidate}, we have
    \[\Prob(\tau_0\in\hat\Gamma_\alpha^{\tau_0}(\bm{X},\bm{y}))\ge\alpha-\sqrt{\frac{(\frac{ep}{s})^s}{4m}}-\sqrt{\frac{\pi}{8m}}-ce^{-cnc_{\min}}\wedge ce^{-cn\tilde c_{\min}}.\]
    \item[(2)] If Assumption \ref{ass_signal_strong} holds, $n$ and $d$ are any fixed numbers, for $c_{\min}^*,\tilde c_{\min}^*$ defined in Theorem \ref{thm_candidate_strong_signal}, we have
    \[\Prob(\tau_0\in\hat\Gamma_\alpha^{\tau_0}(\bm{X},\bm{y}))\ge\alpha-\sqrt{\frac{(\frac{ep}{s})^s}{4m}}-\sqrt{\frac{\pi}{8m}}-ce^{-cnc_{\min}^*}\wedge ce^{-cn\tilde c_{\min}^*}.\]
\end{itemize}

\end{Theorem}

\begin{Remark}[Practical implementation of Algorithm \ref{alg_model_confidence}]
    Step 4 in Algorithm \ref{alg_model_confidence} involves the optimization for indicator functions which could be computationally challenging. This optimization with respect to ${\bm{\beta}}_\tau$ ensures that under the true model $\tau_0$, the statistic $\hat T(\bm{X}^{obs},\bm{y}^{obs},\tau_0)$ is more conservative than $\frac{\abs{\{j\in[m]:\hat P_{\tau_0,{\bm{\beta}}_{0,\tau_0}}(\tilde\tau^{(j)})>\hat P_{\tau_0,{\bm{\beta}}_{0,\tau_0}}(\tilde\tau(\bm{X}^{obs},\bm{y}^{obs},(\tau_0,\bm{\beta}_{0,\tau_0})))\}}}{m}$ which is the oracle statistic when using ${\bm{\beta}}_{0,\tau_0}$ to generate $\bm{Y}^*$. In practice, for any $\tau\in\mathcal{C}$, MLE of ${\bm{\beta}}_{\tau}$ can also be employed to generate $\bm{Y}^{*(j)}$ since it is a consistent estimator in the low dimensional setting given $\tau_0$. And our numerical results confirm that MLE indeed yields confidence sets with guaranteed coverages and reasonable sizes.
\end{Remark}

\subsection{Inference for \texorpdfstring{$A\bm{\beta}_0$}{Abeta\_0}}\label{sec_Abeta}


In this section, we construct confidence sets for linear combinations of coefficients $A{\bm{\beta}}_0$ for any $A\in\R^{q\times p}$, $q \geq 1$. Without loss of generality, we assume ${\rm rank}(A)=q\le p$. 
Here, our target is $A{\bm{\beta}}_0$, and we treat $\tau_0$ and the remaining $p-q$ parameters as the nuisance parameters. 
In the following, we first provide a brief overview of the intuition for inferring $A\bm{\beta}_0$. Then, we elaborate on this intuition with more details.

Recall that $A_{\cdot\tau}$ is a submatrix of $A$ consisting of all the columns with column indexes in $\tau$, so we have $A\bm{\beta}_0=A_{\cdot\tau_0}\bm{\beta}_{0,\tau_0}$. 
Then we can quantify the uncertainty of estimating $A\bm{\beta}_0$ by considering two components: the uncertainty of estimating the model parameters $A_{\cdot\tau_0}{\bm{\beta}}_{0,\tau_0}$ given the true nuisance parameters and the impact of not knowing the nuisance parameters. 
At first, when the true model $\tau_0$ is known, we consider the low-dimensional data $\{(X_{i,\tau_0}^{obs},y^{obs}_i):i\in[n]\}$ with covariates $\bm{X}^{obs}_{\cdot\tau_0}$ constrained on $\tau_0$ and construct a confidence set for $A_{\cdot\tau_0}\bm{\beta}_{0,\tau_0}$ by employing likelihood ratio tests on $A_{\cdot\tau_0}\bm{\beta}_{0,\tau_0}$. 
To address the impact of unknown nuisance parameters, we consider each $\hat\tau\in\mathcal{C}$ as a possible true model and apply a likelihood ratio test using data $\{(X_{i,\hat\tau}^{obs},y_i^{obs}):i\in[n]\}$, resulting in a set for $A_{\cdot\hat\tau}{\bm{\beta}}_{0,\hat\tau}$, which we refer to as {\it representative set}. If $\hat\tau = \tau_0$,  this resulting set  is a level-$\alpha$ confidence set for $A_{\cdot\tau_0}{\bm{\beta}}_{0}$. However, when $\hat\tau \not = \tau_0$, the confidence statement for the resulting set does not hold, thus we refer it here as a representative set. 
By combining these representative sets, we obtain a valid confidence set for $A{\bm{\beta}}_0$. 
Following the intuition described above, we elaborate on this intuition with more details as follows.

Let us first consider the case where $\tau_0$ is known and derive the test statistic for $A{\bm{\beta}}_0$ by considering the working hypothesis test $H_0: A_{\cdot\tau_0}{\bm{\beta}}_{0,\tau_0}=t,{\bm{\beta}}_{0,\tau_0^c}=\bm{0}$ versus $H_1: A_{\cdot\tau_0}{\bm{\beta}}_{0,\tau_0}\ne t, {\bm{\beta}}_{0,\tau_0^c}=\bm{0}$ for any $t\in\R^q$.
Without loss of generality, we assume $A{\bm{\beta}}_0=t$ and $\tau_0$ are compatible, i.e., $\{b\in\R^s:A_{\cdot\tau_0}b=t\}\ne\emptyset$, otherwise we reject $A{\bm{\beta}}_0=t$. Then we denote ${\rm rank}(A_{\cdot\tau})=r(\tau)$, so $r(\tau)\le q\wedge|\tau|$. We also denote the MLE of ${\bm{\beta}}_{0,\tau_0}$ under $H_0$, $H_0\cup H_1$ given $(\bm{X}^{obs},\bm{y}^{obs})$ to be 
$\hat b_{0}$ and $\hat b_{1}$, respectively. Then the likelihood ratio test statistic is defined as
\[\tilde T(\bm{X}^{obs},\bm{y}^{obs},(\tau_0,t))=-2\big(l(\tau_0,\hat b_{0}\mid\bm{X}^{obs},\bm{y}^{obs})-l(\tau_0,\hat b_{1}\mid\bm{X}^{obs},\bm{y}^{obs})\big).\]
Due to the chi-squared approximation of the likelihood ratio test statistic in moderate dimension, if we let $B_\alpha(\tau_0)=[0,F_{\chi^2_{r(\tau_0)}}^{-1}(\alpha))$ with $F^{-1}_{\chi^2_r}(\alpha)$ to be the $\alpha$-quantile of $\chi^2_r$, then
\[\Prob\big(\tilde T(\bm{X},\bm{y},(\tau_0,A\bm{\beta}_0))\in B_\alpha(\tau_0)\big)\rightarrow\alpha.\]
Therefore we collect all the values of $A{\bm{\beta}}_0$ that are accepted by the test to get a confidence set for $A{\bm{\beta}}_0$. 

Secondly, to deal with the impact of unknown $\tau_0$, we apply the previous procedure to each candidate model pretending it is the true model, then we combine all the sets together to get a level-$\alpha$ confidence set of $A{\bm{\beta}}_0$:
\begin{align*}
    \Gamma_\alpha^{A{\bm{\beta}}_0}(\bm{X}^{obs},\bm{y}^{obs})=&\{t:\tilde T(\bm{X}^{obs},\bm{y}^{obs},(\tau,t))<F^{-1}_{\chi^2_{r(\tau)}}(\alpha),\tau\in\mathcal{C}\}.
\end{align*}
We summarize the above procedure in Algorithm \ref{alg_coefficient_confidence}. 

\begin{algorithm}
\caption{Confidence set for $A{\bm{\beta}}_0$}\label{alg_coefficient_confidence}
\begin{algorithmic}[1]
\State{\bf Input:} Observed data $(\bm{X}^{obs},\bm{y}^{obs})$, model candidate set $\mathcal{C}$. 
\State{\bf Output:} Confidence set $\Gamma_{\alpha}^{A\bm{\beta}_0}(\bm{X}^{obs},\bm{y}^{obs})$ for $A{\bm{\beta}}_0$.
\For{$\tau\in\mathcal{C}$}
\State Fix any null parameter space $\Theta_0=\{{\bm{\beta}}\in\R^p:A{\bm{\beta}}_0=t,{\bm{\beta}}_{\tau^c}=\bm{0}\}$ and the full parameter space $\Theta=\{b\in\R^s:{\bm{\beta}}_{\tau^c}=\bm{0}\}$, calculate the MLE under $\Theta_0$ and $\Theta$ as $\hat {\bm{\beta}}^{mle_0}$ and $\hat{\bm{\beta}}^{mle}$, respectively.
\State Calculate
\[\tilde T(\bm{X}^{obs},\bm{y}^{obs},(\tau, t))=-2(l(\tau,\hat {\bm{\beta}}_\tau^{mle_0}\mid\bm{X}^{obs},\bm{y}^{obs})-l(\tau,\hat{\bm{\beta}}_\tau^{mle}\mid\bm{X}^{obs},\bm{y}^{obs})).\]
\EndFor
\State For $r={\rm rank}(A_{\cdot\tau})$, construct
\[\Gamma_{\alpha}^{A{\bm{\beta}}_0}(\bm{X}^{obs},\bm{y}^{obs})=\{t:\tilde T(\bm{X}^{obs},\bm{y}^{obs},(\tau,t))\le F_{\chi_r^2}^{-1}(\alpha),\tau\in\mathcal{C}\}.\]
\end{algorithmic}
\end{algorithm}
It is worth noting that once we get the confidence set $\Gamma^{A\bm{\beta}_0}_\alpha(\bm{X}^{obs},\bm{y}^{obs})$ for $A\bm{\beta}_0$, it is straightforward to transfer $\Gamma^{A\bm{\beta}_0}_\alpha(\bm{X}^{obs},\bm{y}^{obs})$ into the confidence set $\Gamma^{h(A\bm{\beta}_0)}_\alpha(\bm{X}^{obs},\bm{y}^{obs})$ for a nonlinear transformation $h$ of $A\bm{\beta}_0$, by applying $h$ to each element in $\Gamma^{A\bm{\beta}_0}_\alpha(\bm{X}^{obs},\bm{y}^{obs})$,
\[\Gamma^{h(A\bm{\beta}_0)}_\alpha(\bm{X}^{obs},\bm{y}^{obs})=\{h(t):t\in\Gamma_\alpha^{A\bm{\beta}_0}(\bm{X}^{obs},\bm{y}^{obs})\}.\]

In the following, we provide the theoretical guarantee of Algorithm~\ref{alg_coefficient_confidence} to show the valid coverage of $\Gamma_{\alpha}^{A{\bm{\beta}}_0}(\bm{X}^{obs},\bm{y}^{obs})$ and $\Gamma_{\alpha}^{h(A{\bm{\beta}}_0)}(\bm{X}^{obs},\bm{y}^{obs})$. We first introduce an assumption.

\begin{Assumption}\label{ass_hessian}
    Denote $H=n\E\frac{e^{YX^\top{\bm{\beta}}_0}}{(1+e^{YX^\top{\bm{\beta}}_0})^2}X_{\tau_0}X_{\tau_0}^\top$ to be the negative expected Hessian of likelihood $l(\tau_0,\cdot|\bm{X}^{obs},\bm{y}^{obs})$, we assume
    \[\lambda_{\min}(H)\asymp\lambda_{\max}(H)\asymp n.\]
\end{Assumption}

Assumption \ref{ass_hessian} is on the Hessian matrix under the true model $\tau_0$, rather than the Hessian matrix with respect to the full coefficient vector $\bm{\beta}_0$. Therefore, it is weaker than other commonly imposed conditions on the Hessian matrix \citep{cai2021statistical,van2014asymptotically,fei2021estimation}.

Theorem~\ref{thm_Abeta} below states that $\Gamma_\alpha^{A{\bm{\beta}}_0}(\bm{X}^{obs},\bm{y}^{obs})$ and $\Gamma_{\alpha}^{h(A{\bm{\beta}}_0)}(\bm{X}^{obs},\bm{y}^{obs})$ are level-$\alpha$ confidence sets of $A{\bm{\beta}}_0$ and $h(A\bm{\beta}_0)$, respectively. A proof can be found in the Appendix.

\begin{Theorem}\label{thm_Abeta}
If Assumption \ref{ass_hessian} holds, $\|X_{\tau_0}\|_{\psi_2}\le\xi$, $n\gtrsim (s\log n)^{\frac{3}{2}}$ and $n\gg\frac{s^3}{\sqrt{r(\tau_0)}}$, when one of the following conditions holds
\begin{itemize}
    \item[(1)] $d\rightarrow \infty$ at first, then $n\rightarrow\infty$, and $n, p, s$ satisfy Assumption \ref{ass_signal}, 
    \item[(2)]  fix any $d$, $n\rightarrow \infty$, and $n, p, s$ satisfy Assumption \ref{ass_signal_strong},
\end{itemize}
then the confidence sets $\Gamma_\alpha^{A{\bm{\beta}}_0}(\bm{X}^{obs},\bm{y}^{obs})$ and $\Gamma_{\alpha}^{h(A{\bm{\beta}}_0)}(\bm{X}^{obs},\bm{y}^{obs})$ are asymptotically valid
\[\Prob(A{\bm{\beta}}_0\in\Gamma_{\alpha}^{A{\bm{\beta}}_0}(\bm{X},\bm{y}))\ge\alpha-o(1),\quad \Prob(h(A{\bm{\beta}}_0)\in\Gamma_{\alpha}^{h(A{\bm{\beta}}_0)}(\bm{X},\bm{y}))\ge\alpha-o(1).\]
\end{Theorem}


\cite{shi2019linear} also studied the problem of testing $A\bm{\beta}_0$ but with the assumption that $A$ has only $m$ non-zero columns. 
This implies only $m$ elements $\bm{\beta}_{0,M}$ of $\bm{\beta}_0$ are involved in $A\bm{\beta}_0$, for some $M\subset[p]$ with $|M|=m$. 
They developed asymptotically valid tests using partial penalized Wald, score and likelihood ratio statistics, respectively. 
However, the validity of their proposed tests relies on two conditions. 
On the one hand, they suppose $m\ll n^{\frac{1}{3}}$, which restricts the number of coefficients in the test and excludes many important cases such as $A\bm{\beta}_0=\bm{\beta}_0$. 
On the other hand, their approach requires a signal strength condition on the coefficients $\bm{\beta}_{0,M^c}$ that are not involved in the hypothesis, which can be more stringent than the $\beta$-min condition in cases like testing a single large coefficient.

Marginal inference for single coefficients $\beta_{0,j}$ and joint inference for the whole vector ${\bm{\beta}}_0$ are usually of particular interest. 
Additionally, simultaneous inference for the case probabilities of a set of new observations plays an important role in many cases, such as electronic health record data analysis \citep{guo2021inference}.
Equipped with the general result in Theorem \ref{thm_Abeta}, we can address these special cases by setting $A=e_j^\top$, $A=I_p$, and $A=\bm{X}_{\rm new}\in\R^{n_{\rm new}\times p}$, respectively.

\subsubsection{Inference for \texorpdfstring{$\beta_{0,j}$}{beta\_0,j}} \label{sec_betaj}
Following the general framework described in Section \ref{sec_Abeta} with $A=e_j^\top$, to construct a confidence set for $\beta_{0,j}$, we apply the likelihood ratio test to $\beta_{0,j}$ under each candidate model. Concretely, given any candidate model $\tau\in\mathcal{C}$, we test the working hypothesis $H_0:\beta_{0,j}=\beta_j,{\bm{\beta}}_{0,\tau^c}=\bm{0}$ versus $H_1:\beta_{0,j}\ne \beta_j,{\bm{\beta}}_{0,\tau^c}=\bm{0}$.
Without loss of generality, we assume $j\in\tau$, otherwise, if $j\not\in\tau$ and $\beta_j=0$, we accept $H_0$ and if $j\not\in\tau$, $\beta_j\ne 0$, we reject $H_0$. Denote the MLEs of ${\bm{\beta}}_{0,\tau}$ under $H_0$ and $H_0\cup H_1$ to be $\hat b_{0}$ and $\hat b_{1}$, respectively. Then the likelihood ratio test statistic is
\[\tilde T(\bm{X}^{obs},\bm{y}^{obs},(\tau,\beta_j))=-2\big(l(\tau,\hat b_{0}\mid\bm{X}^{obs},\bm{y}^{obs})-l(\tau,\hat b_{1}\mid\bm{X}^{obs},\bm{y}^{obs})\big).\]
Finally, we combine the likelihood ratio test statistics corresponding to each candidate model and define the level-$\alpha$ confidence set for $\beta_{0,j}$ as
\begin{align*} 
\Gamma_\alpha^{\beta_{0,j}}(\bm{X}^{obs},\bm{y}^{obs}) = \{\beta_j: \tilde T(\bm{X}^{obs},\bm{y}^{obs},(\tau,\beta_j)) <F^{-1}_{\chi^2_1}(\alpha), \tau \in \mathcal{C}\}.
\end{align*}
Following Theorem \ref{thm_Abeta}, we can show $\Gamma_{\alpha}^{\beta_{0,j}}(\bm{X}^{obs},\bm{y}^{obs})$ is a valid asymptotic level-$\alpha$ confidence set for $\beta_{0,j}$.
\begin{Corollary}
If Assumption \ref{ass_hessian} holds, $\|X_{\tau_0}\|_{\psi_2}\le\xi$, $n\gtrsim (s\log n)^{\frac{3}{2}}$ and $n\gg s^3$, for any $j\in[p]$, when one of the following conditions holds
\begin{itemize}
    \item[(1)] $d\rightarrow \infty$ at first, then $n\rightarrow\infty$, and $n, p, s$ satisfy Assumption \ref{ass_signal}, 
    \item[(2)]  fix any $d$, $n\rightarrow \infty$, and $n, p, s$ satisfy Assumption \ref{ass_signal_strong},
\end{itemize}
then
\[\Prob(\beta_{0,j}\in\Gamma_\alpha^{\beta_{0,j}}(\bm{X},\bm{y}))\ge \alpha-o(1).\]
\end{Corollary}

The debiasing methods for high-dimensional logistic regression models \citep{cai2021statistical,van2014asymptotically} have been proposed for inferring single coefficients. 
These methods require a constant lower bound for the smallest eigenvalue, of either the Hessian matrix with respect to $\bm{\beta}_0$ or the covariance matrix $\E XX^\top$. Such assumptions can be violated if, for instance, two non-informative covariates are identical. 
However, since Assumption~\ref{ass_hessian} only involves $X_{\tau_0}$, our results remain valid in such cases.

The confidence sets generated by debiasing methods are intervals for any $\beta_{0,j}$, regardless of whether $\beta_{0,j}$ is zero. In contrast, the confidence sets produced by our method are unions of intervals. Specifically, if a candidate model contains the index $j$, the confidence set for $\beta_{0,j}$ will encompass the interval derived under that candidate model. If no candidate model includes $j$, then we are confident that $\beta_{0,j}=0$ and the confidence set for $\beta_{0,j}$ reduces to a singleton $\{0\}$. Therefore our method is more flexible and can adapt to the uncertainties of model selection.

\begin{Remark}\label{rem_weak_signal}
    The inclusion of $\tau_0$ in $\mathcal{C}$ is derived under the signal strength condition in Assumption \ref{ass_signal}, when this condition is violated, the model candidate set may not contain the true support, thus the proposed inference method fails. However, if we divide the set of models $\{\tau\subset[p]:\abs{\tau}\le s\}$ into two parts $\mathcal{T}_1,\mathcal{T}_2$ with $\tau_0\in\mathcal{T}_2$,
\[\inf_{\tau\in\mathcal{T}_1,{\bm{\beta}}_\tau\in\R^{\abs{\tau}},\sigma\ge 0}L_{\Prob_{\bm{\theta}_0}}(\tau,{\bm{\beta}}_\tau,\sigma|\bm{y},\bm{\epsilon})\gtrsim \min\bigg\{\abs{\tau_0\setminus\tau}\dfrac{\log p}{n}+(\abs{\tau}+1)\dfrac{\log \frac{n}{\abs{\tau}+1}}{n},(\abs{\tau}+1)\dfrac{\log p}{n}\bigg\},\]
\[\max_{\tau\in\mathcal{T}_2}\inf_{{\bm{\beta}}_\tau\in\R^{\abs{\tau}},\sigma\ge 0}L_{\Prob_{\bm{\theta}_0}}(\tau,{\bm{\beta}}_\tau,\sigma|\bm{y},\bm{\epsilon})\lesssim\min\bigg\{\abs{\tau_0\setminus\tau}\dfrac{\log p}{n}+(\abs{\tau}+1)\dfrac{\log \frac{n}{\abs{\tau}+1}}{n},(\abs{\tau}+1)\dfrac{\log p}{n}\bigg\},\]
here $\mathcal{T}_1$ contains all the models that have bad prediction performance, while $\mathcal{T}_2$ consists of all models that are close to the true model in terms of prediction. Then using the same argument as Lemma \ref{lem_oracle} and Theorem \ref{thm_candidate}, we are sure that the model candidate set $\mathcal{C}$ contains at least one model in $\mathcal{T}_2$ which is close to $\tau_0$. Therefore, to make inference for $\beta_{0,j}$, we apply the following heuristic: for any $\hat\tau\in\mathcal{C}$, instead of using the low-dimensional model $\hat\tau$ based on $\{(X_{i,\hat\tau}^{obs},y^{obs}_i):i\in[n]\}$, we consider an augmented model $\hat\tau\cup\{j\}$ and use the data $\{(X_{i,\hat\tau\cup\{j\}}^{obs},y^{obs}_i):i\in[n]\}$, then apply the approach above to test $\beta_j$. 

The intuition is that for $j\in\tau_0$, if $\abs{\beta_{0,j}}$ is large, then $\hat\tau$ is likely to contain $j$, so the model augmentation is not harmful, however, if $\abs{\beta_{0,j}}$ is small, then $\mathcal{C}$ could omit $j$. In this case $\hat\tau\cup\{j\}$ may result in better coverage than $\hat\tau$ since the former includes more informative covariates. However, for $j\not\in\tau_0$, since $j\in\hat\tau\cup\{j\}$, using the augmented model, our decision on $\beta_j$ is always reached using likelihood ratio test, therefore the confidence set of $\beta_{0,j}$ always has positive size. While using original $\hat\tau$ may leads to a confidence set $\{0\}$, since if $j\not\in\hat\tau$ for $\forall \hat\tau\in\mathcal{C}$, the original procedure concludes that $\beta_{0,j}=0$. Therefore although the augmented models may lead to better coverage, it could also result in a wider confidence set for the coefficients of non-informative covariates.

\end{Remark}

\subsubsection{Inference for \texorpdfstring{$\bm{\beta}_0$}{beta\_0}} 

Following the general framework in Section \ref{sec_Abeta} with $A=I_p$, to construct a confidence set for $\bm{\beta}_0$, we apply the likelihood ratio test to $\bm{\beta}_0$ under each candidate model. Particularly, for each candidate model $\tau\in\mathcal{C}$, we consider the working hypothesis $H_0:{\bm{\beta}}_{0,\tau}={\bm{\beta}}_{\tau},{\bm{\beta}}_{0,\tau^c}=\bm{0}$ versus $H_1:{\bm{\beta}}_{0,\tau}\ne{\bm{\beta}}_{\tau},{\bm{\beta}}_{0,\tau^c}=\bm{0}$.
We calculate the MLE of ${\bm{\beta}}_{0,\tau}$ as $\hat b_{1}$, then the likelihood ratio test statistic for ${\bm{\beta}}_{\tau}$ is 
\[\tilde T(\bm{X}^{obs},\bm{y}^{obs},(\tau,\bm{\beta}_{\tau}))=-2\big(l(\tau,{\bm{\beta}}_{\tau}\mid\bm{X}^{obs},\bm{y}^{obs})-l(\tau,\hat b_{1}\mid\bm{X}^{obs},\bm{y}^{obs})\big).\]
Given the likelihood ratio test statistics corresponding to each candidate model, the final level-$\alpha$ confidence set for ${\bm{\beta}}_0$ is
\[\Gamma_\alpha^{{\bm{\beta}}_0}(\bm{X}^{obs},\bm{y}^{obs})=\{{\bm{\beta}}:\tilde T(\bm{X}^{obs},\bm{y}^{obs},(\tau,{\bm{\beta}}_\tau))<F^{-1}_{\chi^2_s}(\alpha),{\bm{\beta}}_{\tau^c}=\bm{0},\tau\in\mathcal{C}\}.\]
Similarly, we have the following corollary stating that $\Gamma_\alpha^{{\bm{\beta}}_0}(\bm{X}^{obs},\bm{y}^{obs})$ has asymptotic coverage $\alpha$.
\begin{Corollary}\label{cor_joint}
If Assumption \ref{ass_hessian} holds, $\|X_{\tau_0}\|_{\psi_2}\le\xi$, $n\gtrsim(s\log n)^{\frac{3}{2}}$ and $n\gg s^{5/2}$, when one of the following conditions holds
\begin{itemize}
    \item[(1)] $d\rightarrow \infty$ at first, then $n\rightarrow\infty$, and $n, p, s$ satisfy Assumption \ref{ass_signal}, 
    \item[(2)]  fix any $d$, $n\rightarrow \infty$, and $n, p, s$ satisfy Assumption \ref{ass_signal_strong},
\end{itemize}
then
\[\Prob({\bm{\beta}}_{0}\in\Gamma_\alpha^{{\bm{\beta}}_0}(\bm{X},\bm{y}))\ge \alpha-o(1).\]
\end{Corollary}

\cite{zhang2017simultaneous} also studied the simultaneous inference for $\bm{\beta}_0$ based on the debiasing method \citep{van2014asymptotically}. 
Their approach produces an asymptotically valid test for $\bm{\beta}_0$, provided the smallest eigenvalue of the Hessian matrix of the log-likelihood with respect to $\bm{\beta}_0$ exceeds a positive constant. 
However, this assumption fails to hold if there is collinearity among the non-informative covariates. 
In contrast, our method remains valid in such cases. Moreover, instead of being a full-dimensional ellipsoid, our constructed confidence set is a union of low-dimensional ellipsoids with many coefficients to be exactly zero. Therefore, our method can adapt to the uncertainty of model selection.

\subsubsection{Simultaneous inference for case probabilities}

Logistic regression has been widely applied to detect infectious diseases based on information of patients \citep{ravi2019detection, chadwick2006distinguishing}. 
Statistical inference for patients' case probabilities is critical for identifying those at risk, enabling early intervention.
However, individual-level inference lacks the capacity for group-wise error control and, therefore fails to control disease transmission due to interconnected infection dynamics. 
Consequently, there is an imperative need for simultaneous inference methods for case probabilities of a group of patients. 

Given the fixed covariates $\{X_{{\rm new},i}\in\R^p: i\in[n_{\rm new}]\}$ of an arbitrary group of new patients, we assume the unknown infection statuses $\{Y_{{\rm new},i}\in\{0,1\}:i\in[n]\}$ share the same conditional distribution with the training data, i.e.,
\[\Prob(Y_{{\rm new},i}=1|X_{{\rm new},i})=\frac{1}{1+e^{-X_{{\rm new},i}^\top\bm{\beta}_0}}=h(X_{{\rm new},i}^\top\bm{\beta}_0).\]
Then the case probabilities $\{h(X_{{\rm new},i}^\top\bm{\beta}_0):i\in[n_{\rm new}]\}$ measure the confidence for labeling each new patient as infected. 
Denote $\bm{X}_{\rm new}=(X_{{\rm new},1},\ldots,X_{{\rm new},n_{\rm new}})^\top\in\R^{n_{\rm new}\times p}$, $h(\bm{X}_{\rm new}^\top\bm{\beta}_0)=(h(X_{{\rm new},1}^\top\bm{\beta}_0),\ldots,h(X_{{\rm new},n_{\rm new}}^\top\bm{\beta}_0))^\top\in\R^{n_{\rm new}}$.
To quantify the uncertainty of predicting $Y_{{\rm new}, i}$'s, we aim to conduct statistical inference for all the case probabilities $h(\bm{X}_{\rm new}^\top\bm{\beta}_0)$ of these $n_{\rm new}$ new patients simultaneously. To this end, we construct a confidence set for the vector $h(\bm{X}_{\rm new}^{\top}\bm{\beta}_0)$ and the matrix $A$ in Section~\ref{sec_Abeta} equals $\bm{X}_{\rm new}^\top$.
Then it suffices to form a confidence set for $\bm{X}_{\rm new}^\top\bm{\beta}_0$. 

Following the strategy described in Section~\ref{sec_Abeta} with $A=\bm{X}^\top_{\rm new}$, to construct a confidence set for $\bm{X}_{\rm new}^\top\bm{\beta}_0$, we apply the likelihood ratio test to $\bm{X}_{\rm new}^\top\bm{\beta}_0$ under each candidate model. Specifically, for any candidate model $\tau\in\mathcal{C}$, we consider the working hypotheses $H_0:\bm{X}_{{\rm new},\cdot\tau}\bm{\beta}_{0,\tau}=t, \bm{\beta}_{0,\tau^c}=\bm{0}$ versus $H_1:\bm{X}_{{\rm new},\cdot\tau}\bm{\beta}_{0,\tau}\ne t, \bm{\beta}_{0,\tau^c}=\bm{0}$,
with $\bm{X}_{{\rm new},\cdot\tau}$ to be a submatrix consisting of the columns of $\bm{X}_{\rm new}$ with indexes in $\tau$. Without loss of generality, we assume the existence of $\bm{\beta}$ such that $\bm{X}_{{\rm new},\cdot\tau}\bm{\beta}_{\tau}=t$, otherwise we reject $H_0$. We denote ${\rm rank}(\bm{X}_{{\rm new},\cdot\tau})=r(\tau)$ and let the MLE of $\bm{\beta}_{0,\tau}$ under $H_0$ and $H_0\cup H_1$ to be $\hat b_0$ and $\hat b_1$, respectively. Then the likelihood ratio test statistic is
\[\tilde T(\bm{X}^{obs},\bm{y}^{obs},(\tau,t))=-2\big(l(\tau,\hat b_0|\bm{X}^{obs},\bm{y}^{obs})-l(\tau,\hat b_1|\bm{X}^{obs},\bm{y}^{obs})\big).\]
Given the likelihood ratio test statistics corresponding to each candidate model, we define the final confidence set for $h(\bm{X}^\top_{\rm new}\bm{\beta}_0)$ to be
\[\Gamma_\alpha^{h(\bm{X}_{\rm new}\bm{\beta}_0)}(\bm{X}^{obs},\bm{y}^{obs})=\{h(t):\tilde T(\bm{X}^{obs},\bm{y}^{obs},(\tau,t))<F^{-1}_{\chi^2_{r(\tau)}}(\alpha),\tau\in\mathcal{C}\}.\]
According to Theorem~\ref{thm_Abeta}, we know $\Gamma^{h(\bm{X}^\top_{\rm new}\bm{\beta}_0)}_\alpha(\bm{X}^{obs},\bm{y}^{obs})$ is asymptotically valid.

\begin{Corollary}\label{cor_case_probability}
    If Assumption \ref{ass_hessian} holds, $\|X_{\tau_0}\|_{\psi_2}\le\xi$, $n\gtrsim (s\log n)^{\frac{3}{2}}$ and $n\gg\frac{s^3}{\sqrt{r(\tau_0)}}$, when one of the following conditions holds
\begin{itemize}
    \item[(1)] $d\rightarrow \infty$ at first, then $n\rightarrow\infty$, and $n, p, s$ satisfy Assumption \ref{ass_signal}, 
    \item[(2)]  fix any $d$, $n\rightarrow \infty$, and $n, p, s$ satisfy Assumption \ref{ass_signal_strong},
\end{itemize}
then the confidence set $\Gamma_\alpha^{h(\bm{X}_{\rm new}\bm{\beta}_0)}(\bm{X}^{obs},\bm{y}^{obs})$ is asymptotically valid
\[\Prob(h(\bm{X}_{\rm new}\bm{\beta}_0)\in\Gamma_\alpha^{h(\bm{X}_{\rm new}\bm{\beta}_0)}(\bm{X},\bm{y}))\ge \alpha-o(1).\]
\end{Corollary}

In comparison, \cite{guo2021inference} pioneered the study of statistical inference for case probabilities in high-dimensional logistic regression models. 
However, their method can only be applied to one observation, in contrast, our method enables simultaneous inference for the case probabilities of an arbitrary set of new observations.

\section{Numerical Results}\label{sec_numerical}
In this section, we illustrate the performance of the proposed methods using both synthetic data and real data.

\subsection{Synthetic Data}\label{sec_simulation}
In this subsection, we demonstrate the performance of the proposed methods based on synthetic data. Throughout this subsection, for any sample size $n$, model dimension $p$, sparsity $s$, number of repro samples $d$ and the regression coefficients ${\bm{\beta}}_0$ to be specified later, we generate $n$ i.i.d. copies $\{X_i:i\in[n]\}$ of $X\in\R^p$ from normal distribution $N(\bm{0},\Sigma)$ with mean vector $\bm{0}$ and covariance matrix $\Sigma\in\R^{p\times p}$ satisfying $\Sigma_{ij}=0.2^{\abs{i-j}}$. We set 
$\tau_0=[s]$, then generate $\{Y_i:i\in[n]\}$ from the logistic regression model with
\[\Prob(Y_i=1|X_i)=1-\Prob(Y_i=0|X_i)=\frac{1}{1+e^{-X_{i,\tau_0}^\top{\bm{\beta}}_{0,\tau_0}}}.\]
Specifically, we consider the following four choices of $(n,p,s,d,{\bm{\beta}}_{0,\tau_0})$.
\begin{itemize}
    \item[(M1)] $n=400, p=1000, s=4, d=4000,{\bm{\beta}}_{0,\tau_0}=(5,4,3,2)^\top$.
    \item[(M2)] $n=500, p=1000, s=4, d=2000, {\bm{\beta}}_{0,\tau_0}=(5,4,3,2)^\top$.
    \item[(M3)] $n=700, p=1000, s=4, d=10000, {\bm{\beta}}_{0,\tau_0}=(5,4,3,1)^\top$.
    \item[(M4)] $n=900, p=1000, s=4, d=10000, {\bm{\beta}}_{0,\tau_0}=(5,4,3,1)^\top$.
\end{itemize}
The Model (M1) has strong signals and a relatively small sample size while Model (M3) has weaker signal strength but a relatively larger sample size. Then we consider Model (M2) and (M4) to investigate the influence of sample size by increasing the sample sizes of Model (M1) and (M3), respectively. In the following subsections, we choose the confidence level $\alpha$ to be 0.95 and report the performance of the proposed methods for each model, all the results are calculated based on 100 replications.

\subsubsection{Model candidate set}\label{sec_simulation_candidate}
In this section, we study the coverage of our proposed model candidate set for the unknown true model $\tau_0$. When applying Algorithm \ref{alg_candidate} for the model candidate set, we replace the $\ell_0$ constrained empirical 0-1 risk minimization problem in Step 2 by the following computationally efficient surrogate
\[(\hat{\bm{\beta}}^{(j)}(\lambda_j),\hat\sigma^{(j)}(\lambda_j))=\argmin_{{\bm{\beta}}\in\R^p,\sigma\in\R}\sum_{i=1}^nL((2y_i^{obs}-1)(X_i^{obs\top}{\bm{\beta}}+\sigma\epsilon_i^{*(j)}))+\lambda_j\frac{\norm{{\bm{\beta}}}_1}{\|\tilde{\bm{\beta}}^{(j)}\|_1},\]
\[\hat\tau(\bm{\epsilon}^{*(j)},\lambda_j)=\supp\{\hat{\bm{\beta}}^{(j)}(\lambda_j)\},\]
where we take $L$ to be either the logistic loss $L_l$ or hinge loss $L_h$ defined as
\[L_l(t)=\log(1+e^{-t}),\quad L_h(t)=\max\{0,1-t\},\]
and we choose $\tilde{\bm{\beta}}^{(j)}$ as the solution of
\[(\tilde{\bm{\beta}}^{(j)}(\tilde\lambda_j),\tilde\sigma^{(j)}(\tilde\lambda_j))=\argmin_{{\bm{\beta}}\in\R^p,\sigma\in\R}\sum_{i=1}^nL((2y_i^{obs}-1)(X_i^{obs\top}{\bm{\beta}}+\sigma\epsilon_i^{*(j)}))+\tilde\lambda_j\norm{{\bm{\beta}}}_2,\]
for $\tilde\lambda_j$ chosen by 3-fold cross validation. The tuning parameter $\lambda_j$ is selected using the extended BIC \citep{chen2008extended}
\begin{align*}
    {\rm EBIC}_{j,\xi}(\lambda)=&-2\sum_{i=1}^nL((2y_i^{obs}-1)(X_i^{obs\top}\hat{\bm{\beta}}^{(j)}(\lambda)+\hat\sigma^{(j)}(\lambda)\epsilon_i^{*(j)}))\\
    &+\abs{\hat\tau(\bm{\epsilon}^{*(j)},\lambda)}\log n+2\xi\log{p \choose\abs{\hat\tau(\bm{\epsilon}^{*(j)},\lambda)}}.
\end{align*}
Here we choose $\lambda_j(\xi)$ to maximize ${\rm EBIC}_{j,\xi}(\lambda)$ for each $\xi\in[0,1]$. Therefore for each $\bm{\epsilon}^{*(j)}$, we collect all models $\{\hat\tau(\bm{\epsilon}^{*(j)},\lambda_j(\xi)):\xi\in[0,1]\}$. Then the final model candidate set becomes
\[\{\hat\tau(\bm{\epsilon}^{*(j)},\lambda_j(\xi)):j\in[d],\xi\in[0,1]\}.\]

For the logistic loss $L_l$ and hinge loss $L_h$, we calculate the model candidate sets with 100 replications and report the averaged coverage of $\tau_0$ and the averaged cardinality of the candidate sets with standard deviations in the parentheses in Table \ref{tab_candidate}. From Table \ref{tab_candidate}, we see the proposed method performs well for all four models with high coverage and small cardinality. These results are in line with our theoretical results in Section \ref{sec_candidate}. If the signal is strong enough such that Assumption \ref{ass_signal} or Assumption \ref{ass_signal_strong} are satisfied, then the candidate set is guaranteed to contain the true model. Comparing (M1) with (M2) and (M3) with (M4), we find the coverages increase as sample sizes increase, this coincides with Theorems \ref{thm_candidate} and \ref{thm_candidate_strong_signal} in that the probability of missing $\tau_0$ is decreasing in $n$. It is not surprising that the cardinalities decrease when $n$ increases, since as stated in Theorem \ref{thm_candidate_strong_signal}, when $n$ increases, the stronger signal strength Assumption \ref{ass_signal_strong} becomes weaker, and if this stronger assumption is satisfied, the candidate set are likely to be smaller or even to be precisely $\{\tau_0\}$. Comparing (M1) (M2) with (M3) (M4), we see when ${\bm{\beta}}_{0,\tau_0}$ changes from $(5,4,3,2)^\top$ to $(5,4,3,1)$ and $n$ increases from $(400,500)$ to $(700,900)$, the coverage in (M3) (M4) are comparable to that in (M1) (M2). This demonstrates that if the signals become weaker, the quantities $c_{\min}$ and $\tilde c_{\min}$ in Theorems \ref{thm_candidate} and \ref{thm_candidate_strong_signal} become smaller, therefore the performance will be worse. Then more samples are required to make $nc_{\min}, n\tilde c_{\min}$ still comparable in these two settings, which leads to comparable coverages.

\begin{table}[t]
    \centering
    \begin{tabular*}{\columnwidth}{@{\extracolsep\fill}lcccc@{\extracolsep\fill}}\hline
         & \multicolumn{4}{c}{Losses}\\\cline{2-5}
        & \multicolumn{2}{c}{Hinge} & \multicolumn{2}{c}{Logistic} \\\cline{2-3}\cline{4-5}
        Models  & Coverage & Cardinality & Coverage & Cardinality \\\hline
         M1 & 1.00(0.00) & 6.99(3.37) & 0.98(0.14) & 7.81(4.99)\\
         M2 & 1.00(0.00) & 5.14(2.12) & 0.99(0.10) & 4.91(3.12)\\
         M3 & 0.98(0.14) & 5.54(2.45) & 0.97(0.17) & 3.95(2.07)\\
         M4 & 0.97(0.17) & 4.11(2.05) & 0.99(0.10) & 2.48(1.51)\\\hline
    \end{tabular*}
    \caption{Comparison of performance of the model candidate sets. Here ``Coverage" means the probability for the model candidate set $\mathcal{C}$ to contain $\tau_0$, and ``Cardinality" indicates the number of models in $\mathcal{C}$.}
    \label{tab_candidate}
\end{table}

\subsubsection{Inference for $\tau_0$}
In this subsection, we study the performance of the model confidence set proposed in Section \ref{sec_tau}. When applying Algorithm \ref{alg_model_confidence}, in Step 4, for each $\tau\in\mathcal{C}$, we need to solve an optimization problem for a discrete function which can be hard. In practice, we use the MLE of ${\bm{\beta}}_\tau$ to generate $\bm{Y}^{*(j)}$. We also report the results when the profile method in Step 4 is solved by the \texttt{optim} function in \texttt{R} using the method in \cite{nelder1965simplex}. Here we choose the number $m$ of Monte Carlo samples to be 100 for all settings. The coverages and cardinalities of the model confidence sets are reported in Table \ref{tab_tau_conf} where we deal with the nuisance parameter ${\bm{\beta}}_{0,\tau}$ using both the MLE and profile method. From Table \ref{tab_tau_conf}, we find the model confidence sets are smaller than the model candidate sets in all settings while the coverages of the model confidence sets are the same as the model candidate sets. Due to the discreteness of the nuclear statistic, the model confidence sets are conservative, however, they are still able to reject some models in the model candidate sets and produce smaller sets of models.

\begin{table}[t]
    \centering
    \begin{tabular*}{\columnwidth}{@{\extracolsep\fill}lccccc@{\extracolsep\fill}}\hline
        & & \multicolumn{4}{c}{$\beta_\tau$}\\\cline{3-6}
        & & \multicolumn{2}{c}{Profile} & \multicolumn{2}{c}{MLE}\\\cline{3-4}\cline{5-6}
        Models & Losses   & Coverage & Cardinality  & Coverage & Cardinality \\\hline
        M1 & Hinge & 1.00(0.00) & 4.88(2.68) & 1.00(0.00) & 4.02(2.34)\\
        & Logistic & 0.98(0.14) & 5.19(3.05) & 0.98(0.14) & 4.30(2.37)\\
        M2 & Hinge & 1.00(0.00) & 3.57(1.88) & 1.00(0.00) & 3.02(1.58)\\
        & Logistic & 0.99(0.10) & 3.79(2.28) & 0.99(0.10) & 3.10(1.82)\\
        M3 & Hinge & 0.98(0.14) & 3.78(1.61) & 0.98(0.14) & 3.21(1.45)\\
         & Logistic & 0.97(0.17) & 3.05(1.69) & 0.97(0.17) & 2.62(1.47)\\
         M4 & Hinge & 0.97(0.17) & 2.80(1.52) & 0.97(0.17) & 2.53(1.40)\\
         & Logistic & 0.99(0.10) & 1.94(1.11) & 0.99(0.10) & 1.82(0.98)\\\hline
    \end{tabular*}
   \caption{Comparison of performance of the model confidence sets. Here ``Coverage" means the probability for the model confidence set $\Gamma_\alpha^{\tau_0}(\bm{X}^{obs},\bm{y}^{obs})$ to contain $\tau_0$, and ``Cardinality" indicates the number of models in $\Gamma_\alpha^{\tau_0}(\bm{X}^{obs},\bm{y}^{obs})$.}
    \label{tab_tau_conf}
\end{table}

\subsubsection{Inference for \texorpdfstring{$\beta_{0,j}$}{beta\_0,j}}

In this subsection, we study the performance of the confidence sets for individual coefficients $\beta_{0,j}$ for $j\in[p]$. We also compare our method with the Debiased Lasso method in \cite{van2014asymptotically} implemented using the \texttt{lasso.proj} function in \texttt{hdi} package and also the oracle likelihood ratio test 
assuming $\tau_0$ were known. For Model (M1)-(M4), for each replication, we calculate the coverage and size of confidence sets for each $\beta_{0,j},j\in[p]$, then we average the performance over $j\in\tau_0$ and $j\in[p]\setminus\tau_0$, respectively. Note that the proposed confidence sets for $\beta_{0,j}$ is a union of intervals, so we report the Lebesgue measure of the confidence sets. Then the final results reported in Table \ref{tab_betai_conf} contain the averaged coverages and sizes of confidence sets over 100 replications with standard deviations in the parentheses. As we discussed in Section \ref{sec_simulation_candidate}, we consider two losses, logistic loss and hinge loss, for Step 2 in Algorithm \ref{alg_candidate}. Hereafter, we use the abbreviations ``Repro-Logistic" and ``Repro-Hinge" to denote the repro samples method with logistic loss and hinge loss, respectively. We also use ``Debias" to denote the Debiased Lasso method and use ``Oracle" to denote the oracle likelihood ratio test with the knowledge of $\tau_0$. From Table \ref{tab_betai_conf}, we see that for $j\in\tau_0$, the proposed methods Repro-Hinge and Repro-Logistic have the desired coverage of  0.95 for all four models, while the Oracle method undercovers in (M1) and the Debiased method couldn't cover the nonzero coefficients in all the settings. In terms of size, the confidence sets produced by Repro-Hinge and Repro-Logistic are wider than that of the Oracle method in Model (M1) (M2). The reasons are two folded, on the one hand, the model candidate sets in these two models are relatively large, therefore the confidence sets of $\beta_{0,j},j\in\tau_0$ contain intervals calculated using likelihood ratio test under some wrong models; on the other hand, due to the smaller sample size $(400,500)$, the likelihood ratio tests are not accurate in finite sample, therefore the intervals under wrong models are also wide and unstable, resulting in a large set after taking the union. The sizes of confidence sets created by Repro-Hinge and Repro-Logistic are comparable to that of the Oracle methods in Model (M3) (M4), because the model candidate sets are smaller, and due to the larger sample size $(700,900)$, the likelihood ratio tests are more accurate, so the intervals calculated under supermodels of $\tau_0$ are likely to be close to the interval under $\tau_0$. The sizes of the intervals calculated by the Debiased Lasso method are even shorter than that of the Oracle method, so are likely to be undercover. For the zero coefficients with $j\in[p]\setminus\tau_0$, Repro-Hinge, Repro-Logistic and Debiased Lasso all have coverage rates approaching 1, but the sizes corresponding to Repro-Hinge and Repro-Logistic are shorter than the sizes correspond to Debiased Lasso, the reason is Repro-Hinge and Repro-Logistic also make use of the uncertainty of the selected models. When no models in the candidate set contain $\beta_{0,j}$, we estimate $\beta_{0,j}$ by 0 with confidence 1.

\begin{table}[t]
    \centering
    \begin{tabular*}{\columnwidth}{@{\extracolsep\fill}lccccc@{\extracolsep\fill}}\hline
        &&\multicolumn{2}{c}{$\beta_{0,j},j\in\tau_0$}&\multicolumn{2}{c}{$\beta_{0,j},j\in[p]\setminus\tau_0$}\\\cline{3-4}\cline{5-6}
        Model & Method & Coverage & Length & Coverage & Length \\\hline
        M1 & Repro-Hinge & 0.94(0.19) & 3.20(1.23) & 1.00(0.00) & 0.01(0.00)\\
        & Repro-Logistic & 0.94(0.19) & 3.42(1.40) & 1.00(0.00) & 0.01(0.00)\\
        & Debias & 0.09(0.22) & 0.98(0.19) & 1.00(0.00) & 0.82(0.15)\\
        & Oracle & 0.90(0.23) & 2.23(0.50) & &\\
        M2 & Repro-Hinge & 0.94(0.18) & 2.55(0.72) & 1.00(0.00) & 0.00(0.00)\\
        & Repro-Logistic & 0.94(0.19) & 2.69(1.06) & 1.00(0.00) & 0.00(0.00)\\
        & Debias & 0.13(0.27) & 0.89(0.18) & 1.00(0.00) & 0.74(0.14)\\
        & Oracle & 0.93(0.20) & 1.98(0.34) & &\\
        M3 & Repro-Hinge & 0.97(0.11) & 1.95(0.43) & 1.00(0.00) & 0.00(0.00)\\
        & Repro-Logistic & 0.96(0.12) & 1.86(0.38) & 1.00(0.00) & 0.00(0.00)\\
        & Debias & 0.15(0.24) & 0.73(0.10) & 1.00(0.00) & 0.59(0.08)\\
        & Oracle & 0.94(0.14) & 1.53(0.21) & &\\
        M4 & Repro-Hinge & 0.96(0.14) & 1.52(0.24) & 1.00(0.00) & 0.00(0.00)\\
        & Repro-Logistic & 0.96(0.12) & 1.43(0.21) & 1.00(0.00) & 0.00(0.00)\\
        & Debias & 0.15(0.25) & 0.64(0.07) & 0.99(0.00) & 0.51(0.05)\\
        & Oracle & 0.94(0.17) & 1.30(0.16) & &\\\hline
    \end{tabular*}
    \caption{Comparison of performance of the confidence sets of $\beta_{0,j}$. Here ``Coverage" means the probability for $\Gamma_\alpha^{\beta_{0,j}}(\bm{X}^{obs},\bm{y}^{obs})$ to contain $\beta_{0,j}$, and ``Length" means the Lebesgue measure of $\Gamma_\alpha^{\beta_{0,j}}(\bm{X}^{obs},\bm{y}^{obs})$. The third and fourth columns correspond to $j\in\tau_0$, and the last two columns correspond to $j\in[p]\setminus\tau_0$.}
    \label{tab_betai_conf}
\end{table}

\subsubsection{Inference for \texorpdfstring{$\bm{\beta}_0$}{beta\_0}}

We also study the performance of the proposed method for simultaneous inference for the vector parameter ${\bm{\beta}}_0$. Since it is hard to calculate the Lebesgue measure of the confidence sets, we report only the coverage rates of Repro-Hinge, Repro-Logistic, and the Oracle method with known $\tau_0$ in Table \ref{tab_beta_conf}. From Table \ref{tab_beta_conf}, we can see the Repro-Hinge and Repro-Logistic have similar performance to that of the Oracle method, and the coverage rates are all close to the desired 0.95.

\begin{table}[t]
    \centering
    \begin{tabular*}{\columnwidth}{@{\extracolsep\fill}lccc@{\extracolsep\fill}}\hline
        & \multicolumn{3}{c}{Coverage}\\\cline{2-4}
         Models & Repro-Hinge & Repro-Logistic & Oracle \\\hline
         M1 & 0.98(0.14) & 0.96(0.20) & 0.98(0.14) \\
         M2 & 0.93(0.26) & 0.92(0.27) & 0.93(0.26)\\
         M3 & 0.93(0.26) & 0.91(0.29) & 0.94(0.24)\\
         M4 & 0.92(0.27) & 0.94(0.24) & 0.95(0.22)\\\hline
    \end{tabular*}
    \caption{Comparison of performance of the confidence sets of $\bm{\beta}_0$. Here ``Coverage" means the probability for $\Gamma_\alpha^{\bm{\beta}_0}(\bm{X}^{obs},\bm{y}^{obs})$ to contain $\bm{\beta}_0$.}
    \label{tab_beta_conf}
\end{table}

\subsubsection{Simultaneous inference for case probabilities}

To evaluate the empirical performance of our proposed method for simultaneous inference for case probabilities, we construct $\bm{X}_{\rm new}$ as follows. For (M1)-(M4), the number of new observations is consistently set to $n_{\rm new}=2~{\rm or}~2000$. Then for each of the four models, we generate $X_{{\rm new},i}\in\R^p$ to be i.i.d. random vectors from normal distribution $N(\bm{0},\Sigma_{\rm new})$ with the covariance matrix $\Sigma_{\rm new}$ satisfying $\Sigma_{{\rm new},ij}=0.3^{|i-j|}$. Since it is hard to measure the volume of the confidence sets, we instead report the coverage rates of Repro-Hinge, Repro-Logistic, and the Oracle method with known $\tau_0$ in Table~\ref{tab_case_conf}. The results in Table~\ref{tab_case_conf} reveal that both Repro-Hinge and Repro-Logistic have performance comparable to the Oracle method, with coverage rates close to the nominal value of 0.95. Notably, when $n_{\rm new}=2000>p$, it is typical that ${\rm rank}(\bm{X}_{\rm new})=p$. Therefore testing $\bm{X}_{\rm new}\bm{\beta}_0$ is equivalent to testing $\bm{\beta}_0$. And the coverages for $h(\bm{X}_{\rm new}\bm{\beta}_0)$, as listed in Table~\ref{tab_case_conf} are indeed the same as the coverages for $\bm{\beta}_0$ in Table~\ref{tab_beta_conf}.

\begin{table}[t]
    \centering
    \begin{tabular*}{\columnwidth}{@{\extracolsep\fill}lcccc@{\extracolsep\fill}}\hline
        & & \multicolumn{3}{c}{Coverage}\\\cline{3-5}
         $n_{\rm new}$ & Models & Repro-Hinge & Repro-Logistic & Oracle \\\hline
         2 & M1 & 0.97(0.17) & 0.97(0.17) & 0.93(0.26)\\
         & M2 & 0.96(0.20) & 0.97(0.17) & 0.93(0.26)\\
         & M3 & 0.95(0.22) & 0.94(0.24) & 0.95(0.22)\\
         & M4 & 0.94(0.24) & 0.94(0.24) & 0.92(0.27)\\
         2000 & M1 & 0.98(0.14) & 0.96(0.20) & 0.98(0.14) \\
         & M2 & 0.93(0.26) & 0.92(0.27) & 0.93(0.26)\\
         & M3 & 0.93(0.26) & 0.91(0.29) & 0.94(0.24)\\
         & M4 & 0.92(0.27) & 0.94(0.24) & 0.95(0.22)\\\hline
    \end{tabular*}
    \caption{Comparison of performance of the confidence sets of $h(\bm{X}_{\rm new}\bm{\beta}_0)$. Here ``Coverage" means the probability for $\Gamma_\alpha^{h(\bm{X}_{\rm new}\bm{\beta}_0)}(\bm{X}^{obs},\bm{y}^{obs})$ to contain $h(\bm{X}_{\rm new}\bm{\beta}_0)$.}
    \label{tab_case_conf}
\end{table}

\subsubsection{Inference for \texorpdfstring{$\beta_{0,j}$}{beta\_0,j} with weak signals}

In this subsection, we consider such a difficult setting 
with several weak signals and a small sample size:
\begin{itemize}
    \item[(M5)] $n=300, p=500, s=6, d=10000, {\bm{\beta}}_{0,\tau_0}=(5,4,3,1,0.5,0.2)^\top$.
\end{itemize}
Under this setting, the signal strength Assumption \ref{ass_signal} or \ref{ass_signal_strong} may be violated, and then the model candidate set is no longer guaranteed to contain the true model. This is evidenced in our simulation results, where the candidate sets obtained using either Repro-Logistic or Repro-Hinge fail to contain the true model and both have coverage rates equal to 0. However, we can still get desirable results for the single coefficient  $\beta_{0,j}$ using the augmented method proposed in Section \ref{sec_betaj} Remark \ref{rem_weak_signal}.

To make inference for $\beta_{0,j}$ with weak signals, we study the performance of the heuristic methods augmented Repro-Hinge and augmented Repro-Logistic proposed in Section \ref{sec_betaj} Remark \ref{rem_weak_signal} under Model (M5). In Table \ref{tab_betai_conf_weak}, we compare augmented Repro-Hinge and augmented Repro-Logistic with the Debiased Lasso method and the oracle likelihood ratio test with known $\tau_0$. The confidence intervals of the Oracle method for $\beta_{0,j}, j\in[p]\setminus\tau_0$ are calculated by performing the likelihood ratio test under the model $\tau_0\cup\{j\}$. In Table \ref{tab_betai_conf_weak}, we divide the coefficients into three groups, the first group is $j\in[3]$ which corresponds to the coefficients with large value $(5,4,3)^\top$, the second group is $j\in\tau_0\setminus[3]$ corresponding to small nonzero coefficients $(1,0.5,0.2)^\top$, and the last group $j\in[p]\setminus\tau_0$ contains all the zero coefficients. From Table \ref{tab_betai_conf_weak}, we find Repro-Hinge, and Repro-Logistic have the same performance with their augmented version augmented Repro-Hinge, and augmented Repro-Logistic when $j\in[3]$. Compared to the Oracle method, they have much better coverage and slightly larger sizes. This is due to the weak signals and small sample size. The model candidate set is large in (M5), therefore the confidence intervals produced by the repro samples method contain the intervals of the likelihood ratio test under the wrong models. Due to the small sample size, the chi-squared approximation of the likelihood ratio test is not accurate, so the intervals corresponding to the wrong models are wide and unstable, resulting in a large confidence set after taking the union. The small sample size is also the reason for the Oracle method to be undercover. The Debiased Lasso method still creates confidence intervals with too short lengths and, therefore is likely to be undercover. For the weak signals with $j\in\tau_0\setminus[3]$, the performances of Repro-Hinge and Repro-Logistic are poor since the model candidate set fails to identify these weak signals, however, the augmented Repro-Hinge and augmented Repro-Logistic methods have the desired coverages and the sizes are comparable to that of the Oracle method. The reason is when testing $\beta_{0,j}$, we artificially include $j$ into each candidate model, although the candidate models may still miss some other small nonzero coefficients, they do not have much influence on the data, therefore augmented Repro-Hinge and augmented Repro-Logistic are still able to cover $\beta_{0,j}$. For the zero coefficients with $j\in[p]\setminus\tau_0$, Repro-Hinge, Repro-Logistic, and Debiased Lasso all have coverage 1 while the sizes corresponding to Repro-Hinge and Repro-Logistic are nearly 0. The performances of augmented Repro-Hinge and augmented Repro-Logistic are more similar to that of the Oracle method. They all have coverages approaching 0.95 and comparable sizes.

\begin{table}[t]
\footnotesize
    \centering
    \begin{tabular*}{\columnwidth}{@{\extracolsep\fill}lccccccc@{\extracolsep\fill}}\hline
        & & \multicolumn{6}{c}{Methods}\\\cline{3-8}
        & & \multicolumn{2}{c}{Repro-Hinge} & \multicolumn{2}{c}{Repro-Logistic}\\\cline{3-4}\cline{5-6}
        Coefficients& & Non-augmented & Augmented & Non-augmented & Augmented & Debias & Oracle \\\hline
         $j\in[3]$ & Coverage & 0.98(0.14) & 0.98(0.14) & 0.98(0.13) & 0.98(0.13) & 0.11(0.32) & 0.89(0.32)\\
        & Length & 4.43(2.27) & 4.43(2.27) & 5.21(3.06) & 5.21(3.06) & 1.11(0.27) & 3.14(1.27)\\
         \multirow{2}*{$j\in\tau_0\setminus[3]$} & Coverage & 0.30(0.46) & 0.96(0.20) & 0.34(0.48) & 0.95(0.22) & 0.72(0.45) & 0.93(0.26)\\
         & Length & 0.61(1.00) & 1.71(0.61) & 0.83(1.26) & 1.99(1.13) & 0.89(0.20) & 1.30(0.27)\\
         \multirow{2}*{$j\in[p]\setminus\tau_0$} & Coverage & 1.00(0.00) & 0.98(0.01) & 1.00(0.00) & 0.98(0.10) & 1.00(0.00) & 0.94(0.01)\\
         & Length & 0.01(0.01) & 1.54(0.40) & 0.02(0.03) & 1.79(1.03) & 0.88(0.20) & 1.17(0.15)\\\hline
    \end{tabular*}
    \caption{Comparison of performance of the confidence sets of $\beta_{0,j}$ under model (M5). We consider the following 6 methods: Repro-Hinge, augmented Repro-Hinge, Repro-Logistic, augmented Repro-Logistic, Debiased Lasso, and the Oracle method with $\tau_0$ to be known. We use the terms non-augmented Repro-Hinge, and non-augmented Repro-Logistic to distinguish Repro-Hinge and Repro-Logistic, respectively, from their augmented counterparts. Here ``Coverage" means the probability for $\Gamma_\alpha^{\beta_{0,j}}(\bm{X}^{obs},\bm{y}^{obs})$ to contain $\beta_{0,j}$, and ``Length" means the Lebesgue measure of $\Gamma_\alpha^{\beta_{0,j}}(\bm{X}^{obs},\bm{y}^{obs})$.}
    \label{tab_betai_conf_weak}
\end{table}

\subsection{Real Data}

In this section, we consider a high-dimensional real data analysis. Note that most existing methods focus on statistical inference for single coefficients, but our method can also quantify the uncertainty of model selection. As will be demonstrated, the Debiased Lasso method identifies only one variable as significant. In contrast, our model confidence sets find several variables that have been shown as important by many existing studies.

Specifically, we apply the proposed repro samples method to the single-cell RNA-seq data from \cite{shalek2014single}. This data comprises gene expression profiles for 27723 genes across 1861 primary mouse bone-marrow-derived dendritic cells spanning several experimental conditions. Specifically, we focus on a subset of the data consisting of 96 cells stimulated by the pathogenic component PIC (viral-like double-stranded RNA) and 96 control cells without stimulation, with gene expressions measured six hours after stimulation. In our study, we label each cell with 0 and 1 to indicate ``unstimulated" and ``stimulated" statues, respectively. Our goal is to investigate the association between gene expressions and stimulation status. Similar to \cite{cai2021statistical}, we filter out genes that are not expressed in more than 80\% of the cells and discard the bottom 90\% genes with the lowest variances. Subsequently, we log-transform and normalize the gene expressions to have mean 0 and unit variance.

Using the same parameter tuning strategy as detailed in Section \ref{sec_simulation}, Repro-Hinge and Repro-Logistic identify 7 and 12 models, respectively, in both the model candidate and model confidence sets. We list all models within the model confidence sets in Table \ref{tab_real}. Most of the identified genes have been previously associated with immune systems. RSAD2 is involved in antiviral innate immune responses, and is also a powerful stimulator of adaptive immune response mediated via mDCs \citep{jang2018rsad2}. IFIT1 inhibits viral replication by binding viral RNA that carries PPP-RNA \citep{pichlmair2011ifit1}. IFT80 is known to be an essential component for the development and maintenance of motile and sensory cilia \citep{wang2018ift80}, while ciliary machinery is repurposed by T cell to focus the signaling protein LCK at immune synapse \citep{stephen2018ciliary}. BC044745 has been identified as significant in MRepro-Logistic/MpJ mouse, which exhibits distinct gene expression patterns involved in immune response \citep{podolak2015transcriptome}. ACTB has shown associations with immune cell infiltration, immune checkpoints, and other immune modulators in most cancers \citep{gu2021pan}. HMGN2 has been validated to play an important role in the innate immune system during pregnancy and development in mice \citep{deng2012chromosomal}. Finally, IFI47, also known as IRG47, has been proven to be vital for immune defense against protozoan and bacterial infections \citep{collazo2001inactivation}.

Regarding confidence sets for individual genes, we compare the proposed Repro-Hinge and Repro-Logistic methods with the debaised approach. 
Repro-Hinge identifies RSAD2 and AK217941 as significant, while both Repro-Logistic and Debiased Lasso only identify RSAD2 as significant. 
While RSAD2 plays an important role in antiviral innate immune responses, AK217941, though not studied in the literature, deserves further attention as it has been identified in both model confidence sets and single coefficient confidence sets.

\begin{table}
    \centering
    \begin{tabular*}{\columnwidth}{@{\extracolsep\fill}lccccccccccccccccccc@{\extracolsep\fill}}\hline
        & \multicolumn{19}{c}{Methods}\\\cline{2-20}
        Genes & \multicolumn{7}{c}{Repro-Hinge} & \multicolumn{12}{c}{Repro-Logistic} \\\cline{1-1}\cline{2-8}\cline{9-20}
         RSAD2 & \solidcirc & \solidcirc & \solidcirc & \solidcirc & \solidcirc & \solidcirc & \solidcirc & \hollowcirc & \hollowcirc & \hollowcirc & \hollowcirc & \hollowcirc & \hollowcirc & \hollowcirc & \hollowcirc & \hollowcirc & \hollowcirc & \hollowcirc & \hollowcirc\\
         AK217941 & \solidcirc & \solidcirc & \solidcirc & \solidcirc & \solidcirc & \solidcirc & \solidcirc & \hollowcirc & \hollowcirc & \hollowcirc & \hollowcirc & & \hollowcirc & \hollowcirc & \hollowcirc & \hollowcirc & \hollowcirc & \hollowcirc & \hollowcirc\\
         IFIT1 & &&&& \solidcirc& \solidcirc & &\hollowcirc &\hollowcirc &  &\hollowcirc &  &\hollowcirc &  &\hollowcirc  &\hollowcirc  &\hollowcirc &\hollowcirc &\hollowcirc\\
         IFT80 &&&&&&&&& \hollowcirc & & \hollowcirc & & \hollowcirc& \hollowcirc& \hollowcirc& \hollowcirc& \hollowcirc& \hollowcirc& \hollowcirc\\
         BC044745 & & & \solidcirc & \solidcirc & \solidcirc & \solidcirc & \solidcirc & &&& \hollowcirc & & && \hollowcirc & & & & \hollowcirc\\
         ACTB & & \solidcirc & \solidcirc & & & \solidcirc & \solidcirc & & & & & & & & & & & \hollowcirc &\hollowcirc\\
        HMGN2 & & & & & & & \solidcirc & & & & & & & & \hollowcirc & \hollowcirc & & & \\
         IFI47 & & & & & & & & & & & & & \hollowcirc & & & & & &\\
         ZFP488 & & & & & & & & & & & & & & & & & \hollowcirc & &\\\hline
    \end{tabular*}
    \caption{All the models in the model confidence sets. Each row stands for a gene while each column corresponds to a model. The circle in the $i$-th row and $j$-th column indicates that the $i$-th gene appears in the $j$-th model.}
    \label{tab_real}
\end{table}

\section{Conclusions and Discussions}\label{sec:disc}

In this article, we develop a novel method for statistical inference for both the model support and regression coefficients in high-dimensional logistic regression models based on the repro samples method. For model support, we provide the first effective approach for constructing model confidence sets, assuming only a weak signal strength condition. For the regression coefficients, we develop a comprehensive approach for inference for any group of linear combinations of coefficients without relying on data splitting, sure screening, or sparse Hessian matrix assumptions. Our simulation results demonstrate that the proposed method produces valid and small model confidence sets and achieves better coverage for regression coefficients than the debiasing methods. It is worth noting that under the proposed framework, the logistic distribution of the random noise $\epsilon$ is not crucial, 
so our method could be extended to other generalized linear models with binary outcomes. 

To ensure the coverage of model candidate sets, we introduce a signal strength condition. However, as discussed in Remark \ref{rem_weak_signal}, in the presence of weak signals, the constructed model candidate sets can only be guaranteed to contain at least one model with distribution close to the true data generating distribution. An interesting direction for future exploration would be to devise a method for constructing model candidate sets that contain the true model without relying on any signal strength condition.

\bibliography{reference}
\bibliographystyle{plainnat}

\appendix

\section{Proofs}\label{sec_proof}
This section includes all the proofs of the theoretical results in previous sections.

\subsection{Connection to $\beta_{\rm min}$}
\begin{Lemma}\label{lem_cmin}
For any $\tau_1,\tau_2\subset[p], {\bm{\beta}}_1\in\R^{\abs{\tau_1}},{\bm{\beta}}_2\in\R^{\abs{\tau_2}}$, we have
\[\Prob(\mathbbm{1}\{X_{\tau_1}^\top{\bm{\beta}}_1+\epsilon>0\}\ne\mathbbm{1}\{X_{\tau_2}^\top{\bm{\beta}}_2+\epsilon>0\})= {\rm TV}(\Prob_{X,Y|\tau_1,{\bm{\beta}}_{1}},\Prob_{X,Y|\tau_2,{\bm{\beta}}_{2}}).\]

\end{Lemma}
\begin{proof}

Denote $F_\epsilon$ to be the CDF of $\epsilon$, we have
\begin{align*}
    &\Prob(\mathbbm{1}\{X_{\tau_1}^\top{\bm{\beta}}_1+\epsilon>0\}\ne\mathbbm{1}\{X_{\tau_2}^\top{\bm{\beta}}_2+\epsilon>0\})\\
    =&\E\Prob(-X_{\tau_1}^\top{\bm{\beta}}_{1}<\epsilon\le-X_{\tau_2}^\top{\bm{\beta}}_2|X)+\E\Prob(-X_{\tau_2}^\top{\bm{\beta}}_{2}<\epsilon\le-X_{\tau_1}^\top{\bm{\beta}}_1|X)\\
    =&\E\left|F_\epsilon(-X_{\tau_1}^\top{\bm{\beta}}_1)-F_\epsilon(-X_{\tau_2}^\top{\bm{\beta}}_2)\right|\\
    =&\E\abs{\Prob_{Y|X,(\tau_1,{\bm{\beta}}_1)}(Y=1|X)-\Prob_{Y|X,(\tau_2,{\bm{\beta}}_2)}(Y=1|X))}\\
    =&{\rm TV}(\Prob_{(\tau_1,{\bm{\beta}}_1)},\Prob_{(\tau_2,{\bm{\beta}}_2)}).
\end{align*}

\end{proof}

\begin{Lemma}\label{lem_betamin}
    Denote $\beta_{\min}=\min_{j\in\tau_0}\abs{\beta_{0,j}}$. Assume $\norm{X}_{\psi_2}\le\xi$, $\norm{{\bm{\beta}}_0}_2\le B$ and $\lambda_{\min}(\E XX^\top)\ge C$ with $\xi,B$ and $C$ are positive constants,
    then
    \[\inf_{\abs{\tau}\le s,\tau\ne\tau_0,{\bm{\beta}}_\tau\in\R^{\abs{\tau}}}\frac{{\rm TV}(\Prob_{\bm{\theta}_0},\Prob_{(\tau,{\bm{\beta}}_\tau)})}{\sqrt{\abs{\tau_0\setminus\tau}}}\gtrsim\beta_{\min}.\]
\end{Lemma}

\begin{proof}
    
    By Lemma \ref{lem_cmin},
    \begin{align*}
        &{\rm TV}(\Prob_{\bm{\theta}_0},\Prob_{(\tau,{\bm{\beta}}_\tau)})\\
        =&\E\abs{\frac{1}{1+e^{-X_{\tau_0}^\top{\bm{\beta}}_{0,\tau_0}}}-\frac{1}{1+e^{-X_\tau^\top{\bm{\beta}}_\tau}}}\\
    =&\E\left|\dfrac{1}{1+e^{-X_{\tau_0}^\top{\bm{\beta}}_{0,\tau_0}}}-\dfrac{1}{1+e^{-X_\tau^\top{\bm{\beta}}_\tau}}\right|\{\mathbbm{1}\{|X_{\tau}^\top{\bm{\beta}}_\tau|\le2|X_{\tau_0}^\top{\bm{\beta}}_{0,\tau_0}|\}+\mathbbm{1}\{|X_{\tau}^\top{\bm{\beta}}_\tau|>2|X_{\tau_0}^\top{\bm{\beta}}_{0,\tau_0}|\}\}\\
    \ge &\E|X_{\tau_0}^\top{\bm{\beta}}_{0,\tau_0}-X_\tau^\top{\bm{\beta}}_\tau|\dfrac{e^{-2|X_{\tau_0}^\top{\bm{\beta}}_{0,\tau_0}|}}{(1+e^{-2|X_{\tau_0}^\top{\bm{\beta}}_{0,\tau_0}|})^2}\mathbbm{1}\{|X_{\tau}^\top{\bm{\beta}}_\tau|\le2|X_{\tau_0}^\top{\bm{\beta}}_{0,\tau_0}|\}\\
    &+\E|X_{\tau_0}^\top{\bm{\beta}}_{0,\tau_0}|\dfrac{e^{-2|X_{\tau_0}^\top{\bm{\beta}}_{0,\tau_0}|}}{(1+e^{-2|X_{\tau_0}^\top{\bm{\beta}}_{0,\tau_0}|})^2}\mathbbm{1}\{|X_{\tau}^\top{\bm{\beta}}_\tau|>2|X_{\tau_0}^\top{\bm{\beta}}_{0,\tau_0}|\}\\
    \ge&\dfrac{(\E\min\{|X_{\tau_0}^\top{\bm{\beta}}_{0,\tau_0}-X_\tau^\top{\bm{\beta}}_\tau|,|X_{\tau_0}^\top{\bm{\beta}}_{0,\tau_0}|\}^{1/2})^2}{\E(1+e^{-2|X_{\tau_0}^\top{\bm{\beta}}_{0,\tau_0}|})^2e^{2|X_{\tau_0}^\top{\bm{\beta}}_{0,\tau_0}|}}\\
    \ge&\beta_{\min}\inf_{\abs{\tau}\le s,\tau\ne\tau_0,{\bm{\beta}}_\tau\in\R^{\abs{\tau}}}\left(\E\min\left\{\abs{X_{\tau_0}^\top\frac{{\bm{\beta}}_{0,\tau_0}}{\beta_{\min}}-X_\tau^\top{\bm{\beta}}_\tau},\abs{X_{\tau_0}^\top\frac{{\bm{\beta}}_{0,\tau_0}}{\beta_{\min}}}\right\}^{1/2}\right)^2e^{-c\|{\bm{\beta}}_0\|_2^2\xi^2}/4.
    \end{align*}
    For $\tau\ne\tau_0,\abs{\tau}\le s$, there exists $b\in\R^p,\norm{b}_0\le 2s$ such that $X_{\tau_0}^\top\frac{{\bm{\beta}}_{0,\tau_0}}{\beta_{\min}}-X_\tau^\top{\bm{\beta}}_\tau=X^\top b$. For any $j\in\tau_0\setminus\tau$, we have $\abs{b_j}=\frac{\abs{\beta_{0,j}}}{\beta_{\min}}\ge 1$, therefore $\norm{b}_2\ge \sqrt{\abs{\tau_0\setminus\tau}}$. Similarly $\norm{\frac{{\bm{\beta}}_{0,\tau_0}}{\beta_{\min}}}_2\ge \sqrt{s}$. Note that 
    \begin{align*}
        \sup_{\norm{{\bm{\beta}}}_0\le 2s,\norm{{\bm{\beta}}}_1=1}\Prob(\abs{X^\top{\bm{\beta}}}\le\frac{\sqrt{C}}{2})\le\sup_{\norm{{\bm{\beta}}}_0\le 2s,\norm{{\bm{\beta}}}_1=1}\frac{C^2}{4{\bm{\beta}}^\top\E XX^\top{\bm{\beta}}}\le\frac{1}{4}.
    \end{align*}
    Then
    \begin{align*}
        &\inf_{\abs{\tau}\le s,\tau\ne\tau_0,{\bm{\beta}}_\tau\in\R^{\abs{\tau}}}\E\min\left\{\abs{X_{\tau_0}^\top\frac{{\bm{\beta}}_{0,\tau_0}}{\beta_{\min}}-X_\tau^\top{\bm{\beta}}_\tau},\abs{X_{\tau_0}^\top\frac{{\bm{\beta}}_{0,\tau_0}}{\beta_{\min}}}\right\}^{1/2}\\
        \ge&\inf_{\abs{\tau}\le s,\tau\ne\tau_0,{\bm{\beta}}_\tau\in\R^{\abs{\tau}}}\E\min\left\{\abs{X_{\tau_0}^\top\frac{{\bm{\beta}}_{0,\tau_0}}{\beta_{\min}}-X_\tau^\top{\bm{\beta}}_\tau},\abs{X_{\tau_0}^\top\frac{{\bm{\beta}}_{0,\tau_0}}{\beta_{\min}}}\right\}^{1/2}\\
        &\qquad\cdot\1\left\{\abs{X_{\tau_0}^\top\frac{{\bm{\beta}}_{0,\tau_0}}{\beta_{\min}}-X_\tau^\top{\bm{\beta}}_\tau}>\sqrt{\abs{\tau_0\setminus\tau}}\frac{\sqrt{C}}{2},\abs{X_{\tau_0}^\top\frac{{\bm{\beta}}_{0,\tau_0}}{\beta_{\min}}}>\sqrt{s}\frac{\sqrt{C}}{2}\right\}\\
        \ge&\frac{C^{1/4}}{\sqrt{2}}\abs{\tau_0\setminus\tau}^{1/4}\inf_{\abs{\tau}\le s,\tau\ne\tau_0,{\bm{\beta}}_\tau\in\R^{\abs{\tau}}}\bigg(1-\Prob\bigg(\abs{X_{\tau_0}^\top\frac{{\bm{\beta}}_{0,\tau_0}}{\beta_{\min}}-X_\tau^\top{\bm{\beta}}_\tau}\le \sqrt{\abs{\tau_0\setminus\tau}}\frac{\sqrt{C}}{2}\bigg)\\
        &-\Prob\bigg(\abs{X_{\tau_0}^\top\frac{{\bm{\beta}}_{0,\tau_0}}{\beta_{\min}}}\le \sqrt{s}\frac{\sqrt{C}}{2}\bigg)\bigg)\\
        \ge&\frac{C^{1/4}}{\sqrt{2}}\abs{\tau_0\setminus\tau}^{1/4}(1-2\sup_{\norm{{\bm{\beta}}}_0\le 2s,\norm{{\bm{\beta}}}_2=1}\Prob(|X^\top{\bm{\beta}}|\le \frac{\sqrt{C}}{2}))\\
        \ge& \frac{C^{1/4}}{2\sqrt{2}}\abs{\tau_0\setminus\tau}^{1/4}.
    \end{align*}
    Combining terms completes the proof.
\end{proof}

\subsection{Proofs in Section~\ref{sec_candidate}}

The following lemma follows from the Fundamental Theorem of Learning Theory \citep{shalev2014understanding}
\begin{Lemma}\label{lem_pac_upper}
For any $\tau\subset[p]$, we have
\begin{align*}
    &\Prob(\exists{\bm{\beta}}_\tau\in\R^{\abs{\tau}},\sigma\ge 0 ~{\rm s.t.}~ L_n(\tau,{\bm{\beta}}_\tau,\sigma|\bm{X},\bm{y},\bm{\epsilon})=0, L_{\bm{\theta}_0}(\tau,{\bm{\beta}}_\tau,\sigma)\ge\eta)\\
    \le&(1-e^{-\frac{n\eta}{8}})^{-1}\bigg\{2^{\abs{\tau}+1}\vee\bigg(\dfrac{2en}{\abs{\tau}+1}\bigg)^{\abs{\tau}+1}\bigg\}2^{-\frac{n\eta}{2}}.
\end{align*}
\end{Lemma}

\begin{proof}[Proof of Lemma \ref{lem_pac_upper}]
Suppose we have another sample $\tilde S=\{(\tilde X_i,\tilde \epsilon_i,\tilde Y_i):i\in[n]\}$ that is i.i.d. with $S=\{(X_i^{obs},\epsilon_i^{rel},y_i^{obs}):i\in[n]\}$. Denote
\[A=\{\exists{\bm{\beta}}_\tau\in\R^{\abs{\tau}},\sigma\ge 0 ~{\rm s.t.}~ L_n(\tau,{\bm{\beta}}_\tau,\sigma|\bm{X},\bm{y},\bm{\epsilon})=0, L_{\bm{\theta}_0}(\tau,{\bm{\beta}}_\tau,\sigma)\ge\eta\},\]
\[B=\{\exists{\bm{\beta}}_\tau\in\R^{\abs{\tau}},\sigma\ge 0 ~{\rm s.t.}~ L_n(\tau,{\bm{\beta}}_\tau,\sigma|\bm{X},\bm{y},\bm{\epsilon})=0, L_n(\tau,{\bm{\beta}}_\tau,\sigma|\tilde{\bm{X}},\tilde{\bm{y}},\tilde{\bm{\epsilon}})\ge\dfrac{\eta}{2}\}.\]
Conditioning on event $A$, we denote $\hat{\bm{\beta}}_\tau\in R^{\abs{\tau}},\hat\sigma\ge 0$ to be the coefficients satisfy $A$. Given $S$ and $A$, $\mathbbm{1}\{\tilde Y\ne\mathbbm{1}\{\tilde X_\tau^\top\hat{\bm{\beta}}_{\tau}+\hat\sigma\tilde \epsilon>0\}\}$ is a Bernoulli random variable with parameter $\rho=L_{\bm{\theta}_0}(\tau,\hat{\bm{\beta}}_\tau,\hat\sigma)\ge\eta$, using Chernoff bound in multiplicative form \citep{hoeffding1994probability}, we have
\[\Prob(B^c|A)\le\Prob(L_n(\tau,\hat{\bm{\beta}}_\tau,\hat\sigma|\tilde{\bm{X}},\tilde{\bm{Y}},\tilde{\bm{\epsilon}})\le\frac{1}{2}L_{\bm{\theta}_0}(\tau,\hat{\bm{\beta}}_\tau,\sigma)|A)\le\E e^{-\frac{n\rho}{8}}\le e^{-\frac{n\eta}{8}}.\]
Then
\[\Prob(B)\ge\Prob(B|A)\Prob(A)\ge (1-e^{-\frac{n\eta}{8}})\Prob(A).\]
Now conditioning on $S\cup\tilde S$, we construct $T$ and $\tilde T$ by randomly partitioning $S\cup\tilde S$ into two sets with equal sizes. We also denote
\[L_n(\tau,{\bm{\beta}}_\tau,\sigma|T)=\dfrac{1}{n}\sum_{(X,\epsilon,Y)\in T}\mathbbm{1}\{Y\ne\mathbbm{1}\{X_{\tau}^\top{\bm{\beta}}_{\tau}+\sigma\epsilon>0\}\},\]
\[L_n(\tau,{\bm{\beta}}_\tau,\sigma|\tilde T)=\dfrac{1}{n}\sum_{(X,\epsilon,Y)\in \tilde T}\mathbbm{1}\{Y\ne\mathbbm{1}\{X_{\tau}^\top{\bm{\beta}}_{\tau}+\sigma\epsilon>0\}\},\]
then 
\begin{align*}
    \Prob(B)=&\E_{S\cup\tilde S}\Prob(\exists{\bm{\beta}}_\tau\in\R^{\abs{\tau}},\sigma\ge 0 ~{\rm s.t.}~ L_n(\tau,{\bm{\beta}}_\tau,\sigma|\bm{X},\bm{y},\bm{\epsilon})=0, L_n(\tau,{\bm{\beta}}_\tau,\sigma|\tilde{\bm{X}},\tilde{\bm{Y}},\tilde{\bm{\epsilon}})\ge\dfrac{\eta}{2}|S\cup\tilde S)\\
    =&\E_{S\cup\tilde S}\Prob(\exists{\bm{\beta}}_\tau\in\R^{\abs{\tau}},\sigma\ge 0 ~{\rm s.t.}~ L_n(\tau,{\bm{\beta}}_\tau,\sigma|T)=0, L_n(\tau,{\bm{\beta}}_\tau,\sigma|\tilde T)\ge\dfrac{\eta}{2}|S\cup\tilde S).
\end{align*}
Conditioning on $S\cup\tilde S$, instead of considering ${\bm{\beta}}_\tau$ directly, we study the evaluation of the classifiers $\mathbbm{1}(X_\tau^\top{\bm{\beta}}_\tau+\sigma\epsilon>0)$ on samples in $S\cup\tilde S$, then by Sauer's Lemma \citep{shalev2014understanding}, the total number of labellings of $\mathbbm{1}\{X_\tau^\top{\bm{\beta}}_\tau+\sigma\epsilon>0\},\forall{\bm{\beta}}_\tau\in\R^{\abs{\tau}},\sigma\ge 0$ on $S\cup\tilde S$ is less than $2^{\abs{\tau}+1}\vee\big(\frac{2en}{\abs{\tau}+1}\big)^{\abs{\tau}+1}$
\begin{align*}
    &\Prob(\exists{\bm{\beta}}_\tau\in\R^{\abs{\tau}},\sigma\ge 0 ~{\rm s.t.}~ L_n(\tau,{\bm{\beta}}_\tau,\sigma|T)=0, L_n(\tau,{\bm{\beta}}_\tau,\sigma|\tilde T)\ge\dfrac{\eta}{2}|S\cup\tilde S)\\
    \le&\bigg\{2^{\abs{\tau}+1}\vee\bigg(\dfrac{2en}{\abs{\tau}+1}\bigg)^{\abs{\tau}+1}\bigg\}\sup_{{\bm{\beta}}_\tau\in\R^{\abs{\tau}},\sigma\ge 0}\Prob(L_n(\tau,{\bm{\beta}}_\tau,\sigma|T)=0, L_n(\tau,{\bm{\beta}}_\tau,\sigma|\tilde T)\ge\dfrac{\eta}{2}|S\cup\tilde S)\\
    \le&\bigg\{2^{\abs{\tau}+1}\vee\bigg(\dfrac{2en}{\abs{\tau}+1}\bigg)^{\abs{\tau}+1}\bigg\}2^{-\frac{n\eta}{2}},
\end{align*}
where to derive the last inequality, we assume the total number of errors of ${\bm{\beta}}_\tau$ on $S\cup\tilde S$ to be $m\in[\frac{n\eta}{2},n]$, then the probability that all the $m$ wrong samples are in $\tilde T$ is $\binom{n}{m}/\binom{2n}{m}\le 2^{-m}\le 2^{-\frac{n\eta}{2}}$. 

In conclusion, we have 
\[\Prob(A)\le(1-e^{-\frac{n\eta}{8}})^{-1}\bigg\{2^{\abs{\tau}+1}\vee\bigg(\dfrac{2en}{\abs{\tau}+1}\bigg)^{\abs{\tau}+1}\bigg\}2^{-\frac{n\eta}{2}}.\]
\end{proof}

\begin{proof}[Proof of Lemma \ref{lem_oracle}]
Since $\tau_0$ is one of the minimizers of problem \eqref{eq_oracle}, we know the minimum is 0. Denote
\[\tilde c_{\min}=\min_{|\tau|\le s,\tau\not\supset\tau_0,{\bm{\beta}}_\tau\in\R^{\abs{\tau}},\sigma\ge 0}\frac{L_{\bm{\theta}_0}(\tau,{\bm{\beta}}_\tau,\sigma)-\frac{2\abs{\tau}+2}{n}\log_2\frac{2en}{\abs{\tau}+1}}{\abs{\tau_0\setminus\tau}},\]
\[c_{\min}=\min_{|\tau|\le s,\tau\not\supset\tau_0,{\bm{\beta}}_\tau\in\R^{\abs{\tau}},\sigma\ge 0}\frac{L_{\bm{\theta}_0}(\tau,{\bm{\beta}}_\tau,\sigma)-\frac{2\abs{\tau}+2}{n}\log_2\frac{2en}{\abs{\tau}+1}}{\abs{\tau}\vee 1},\]
then
\begin{align*}
    &\Prob(\inf_{\tau\not\supset\tau_0,|\tau|\le s,{\bm{\beta}}\in\R^p,\sigma\ge 0} L_n(\tau,{\bm{\beta}}_\tau,\sigma|\bm{X},\bm{y},\bm{\epsilon})=0)\\
    = &\Prob(\exists \tau\not\supset\tau_0, |\tau|\le s, {\bm{\beta}}_\tau\in\R^{\abs{\tau}},\sigma\ge 0 ~{\rm s.t.}~ L_n(\tau,{\bm{\beta}}_\tau,\sigma|\bm{X},\bm{y},\bm{\epsilon})=0,\\
    &\qquad L_{\bm{\theta}_0}(\tau,{\bm{\beta}}_\tau,\sigma)\ge\inf_{{\bm{\beta}}_\tau\in\R^{\abs{\tau}},\sigma\ge 0}L_{\bm{\theta}_0}(\tau,{\bm{\beta}}_\tau,\sigma))\\
    \le&\sum_{\tau\not\supset\tau_0,|\tau|\le s}\Prob(\exists{\bm{\beta}}_\tau\in\R^{\abs{\tau}},\sigma\ge 0 ~{\rm s.t.}~ L_n(\tau,{\bm{\beta}}_\tau,\sigma|\bm{X},\bm{y},\bm{\epsilon})=0,\\
    &\qquad L_{\bm{\theta}_0}(\tau,{\bm{\beta}}_\tau,\sigma)\ge\inf_{{\bm{\beta}}_\tau\in\R^{\abs{\tau}},\sigma\ge 0}L_{\bm{\theta}_0}(\tau,{\bm{\beta}}_\tau,\sigma))\\
    \overset{\triangle}{=}&T.
\end{align*}
On the one hand, noting that $\sum_{l=0}^r\binom{p-s}{l}\le (\frac{e(p-s)}{r})^r,\binom{s}{r}\le s^r,s(p-s)\le\frac{p^2}{4}$, if we divide $\abs{\tau}$ into $j=\abs{\tau_0\cap\tau}$ and $l=\abs{\tau\setminus\tau_0}$, then applying Lemma \ref{lem_pac_upper} gives
\begin{align*}
    T\lesssim &\sum_{\tau\not\supset\tau_0,\abs{\tau}\le s}2^{-\frac{1}{2}n\inf_{{\bm{\beta}}_\tau\in\R^{\abs{\tau}},\sigma\ge 0}L_{\bm{\theta}_0}(\tau,{\bm{\beta}}_\tau,\sigma)+(\abs{\tau}+1)\log_2 \frac{2en}{\abs{\tau}+1}}\\
    \le&\sum_{j=0}^{s-1}\sum_{l=0}^{s-j}\binom{s}{j}\binom{p-s}{l}2^{-\frac{1}{2}n(s-j)\tilde c_{\min}}\\
    \overset{r=s-j}{\le} & \sum_{r=1}^s s^r2^{-\frac{1}{2}nr\tilde c_{\min}}\sum_{l=0}^{r}\binom{p-s}{l}\\
    \le&\sum_{r=1}^s2^{-r(\frac{1}{2}n\tilde c_{\min}-\log_2 (es(p-s)))}\\
    \le & \dfrac{2^{-\frac{1}{2}n\tilde c_{\min}+\log_2 (es(p-s))}}{1-2^{-\frac{1}{2}n\tilde c_{\min}+\log_2 (es(p-s))}}\\
    \le & 2^{-\frac{1}{2}n\tilde c_{\min}+\log_2(es(p-s))+1}\\
    \lesssim&2^{-\frac{1}{2}n\tilde c_{\min}+2\log_2 p}.
\end{align*}

On the other hand, similarly we denote $j=\abs{\tau}$, then
\begin{align*}
    T\lesssim &\sum_{\tau\not\supset\tau_0,\abs{\tau}\le s}2^{-\frac{1}{2}n\inf_{{\bm{\beta}}_\tau\in\R^{\abs{\tau}},\sigma\ge 0}L_{\bm{\theta}_0}(\tau,{\bm{\beta}}_\tau,\sigma)+(\abs{\tau}+1)\log_2 \frac{2en}{\abs{\tau}+1}}\\
    \le&\sum_{j=0}^s\binom{p}{j}2^{-\frac{1}{2}n(j\vee 1)c_{\min}}\\
    \le&\sum_{j=0}^s2^{-\frac{1}{2}n(j\vee 1)c_{\min}+j\log_2 p}\\
    \lesssim &2^{-\frac{1}{2}nc_{\min}+\log_2 p}.
\end{align*}
\end{proof}

\begin{proof}[Proof of Theorem \ref{thm_candidate}]

If we denote
\[A=\{\bm{\epsilon}^*:-X^\top_{i,\tau_0}\bm{\beta}_{0,\tau_0}<\epsilon^*_i\le\epsilon_i~{\rm or}~\epsilon_i\le\epsilon_i^*\le-X^\top_{i,\tau_0}\bm{\beta}_{0,\tau_0}\},\]
then we have the following decomposition
\begin{align*}
    \Prob(\tau_0\not\in\mathcal{C})
    \le\Prob(\{\tau_0\not\in\mathcal{C})\}\cap(\cup_{j\in[d]}\{\bm{\epsilon}^{*(j)}\in A\}))+\Prob(\cap_{j\in[d]}\{\bm{\epsilon}^{*(j)}\not\in A\})
    =T_1+T_2.
\end{align*}
Note that for any $\bm{\epsilon}^*\in A$, we have 
\[y_i=\1(X^\top_{i,\tau_0}\bm{\beta}_{0,\tau_0}+\epsilon_i^*>0),\quad \epsilon^*_i-\epsilon_i\left\{\begin{array}{cc}
    \le 0 & {\rm if~} y_i=1, \\
    \ge 0 & {\rm if~} y_i=0.
\end{array}\right.\]
Then for all $\tau\subset[p],\bm{\beta}_\tau\in\R^{|\tau|},\sigma\ge 0$,
\begin{align*}
    &\1(y_i\ne\1(X^\top_{i,\tau}\bm{\beta}_\tau+\sigma\epsilon^*_i>0))\\
    =&\1(y_i=1,X^\top_{i,\tau}\bm{\beta}_\tau+\sigma\epsilon^*_i\le 0)+\1(y_i=0,X^\top_{i,\tau}\bm{\beta}_\tau+\sigma\epsilon^*_i>0)\\
    =&\1(y_i=1,X^\top_{i,\tau}\bm{\beta}_\tau+\sigma\epsilon_i+\sigma(\epsilon^*_i-\epsilon_i)\le 0)+\1(y_i=0,X^\top_{i,\tau}\bm{\beta}_\tau+\sigma\epsilon_i+\sigma(\epsilon_i^*-\epsilon_i)>0)\\
    \ge&\1(y_i=1,X^\top_{i,\tau}\bm{\beta}_\tau+\sigma\epsilon_i\le 0)+\1(y_i=0,X^\top_{i,\tau}\bm{\beta}_\tau+\sigma\epsilon_i>0)\\
    =&\1(y_i\ne\1(X^\top_{i,\tau}\bm{\beta}_\tau+\sigma\epsilon_i>0)),
\end{align*}
then we can control term $T_1$ as
\begin{align*}
    T_1\le&\Prob(\exists\bm{\epsilon}^*\in A~{\rm s.t.}~\tau_0\ne\argmin_{\abs{\tau}\le s}\min_{{\bm{\beta}}_\tau\in\R^{\abs{\tau}},\sigma\ge0}L_n(\tau,{\bm{\beta}}_\tau,\sigma|\bm{X},\bm{y},\bm{\epsilon}^*))\\
    \le&\Prob(\exists\bm{\epsilon}^*\in A~{\rm s.t.}~\inf_{\tau\ne\tau_0,\abs{\tau}\le s,{\bm{\beta}}_\tau\in\R^{\abs{\tau}},\sigma\ge0}L_n(\tau,{\bm{\beta}}_\tau,\sigma|\bm{X},\bm{y},\bm{\epsilon}^*)\le L_n(\tau_0,{\bm{\beta}}_{0,\tau_0},1|\bm{X},\bm{y},\bm{\epsilon}^*))\\
    =&\Prob(\exists\bm{\epsilon}^*\in A~{\rm s.t.}~\inf_{\tau\ne\tau_0,\abs{\tau}\le s,{\bm{\beta}}_\tau\in\R^{\abs{\tau}},\sigma\ge0}L_n(\tau,{\bm{\beta}}_\tau,\sigma|\bm{X},\bm{y},\bm{\epsilon}^*)=0)\\
    \le&\Prob(\inf_{\tau\ne\tau_0,\abs{\tau}\le s,{\bm{\beta}}_\tau\in\R^{\abs{\tau}},\sigma\ge0}L_n(\tau,{\bm{\beta}}_\tau,\sigma|\bm{X},\bm{y},\bm{\epsilon})=0)\\
    \lesssim&2^{-\frac{1}{2}n\tilde c_{\min}+2\log_2p}\wedge2^{-\frac{1}{2}nc_{\min}+\log_2p},
\end{align*}
where we have used Lemma \ref{lem_oracle} in the last inequality.

For term $T_2$, denote $F_\epsilon(z)=(1+e^{-z})^{-1}$ to be the CDF of logistic distribution, then
\begin{align*}
    T_2=&(1-\Prob(\bm{\epsilon}^*\in A))^d\\
    =&(1-\{\Prob(-X^\top_{\tau_0}\bm{\beta}_{0,\tau_0}<\epsilon^*\le\epsilon~{\rm or}~\epsilon\le\epsilon^*\le-X^\top_{\tau_0}\bm{\beta}_{0,\tau_0})\}^n)^d\\
    =&(1-\{\E \big|F_\epsilon(\epsilon)-F_\epsilon(-X^\top_{\tau_0}\bm{\beta}_{0,\tau_0})\big|\}^n)^d,
\end{align*}
where in the last equation, we have used the fact that $\epsilon^*$ is independent of $Y,X$. Combining terms completes the proof.
\end{proof}

\begin{proof}[Proof of Theorem \ref{thm_candidate_strong_signal}]
    For any $\bm{\epsilon}^*$ independent of the observed data, by Theorem 4.10 and Example 5.24 in \cite{wainwright2019high}, given any $\tau\subset[p]$, we have
    \begin{align*}
        &\Prob(\sup_{{\bm{\beta}}_\tau\in\R^{\tau},\sigma\ge0}\abs{L_n-L_{\bm{\theta}_0}}(\tau,{\bm{\beta}}_\tau,\sigma|\bm{X},\bm{y},\bm{\epsilon}^*)\vee\abs{L_n-L_{\bm{\theta}_0}}(\tau_0,{\bm{\beta}}_{0,\tau_0},0|\bm{X},\bm{y},\bm{\epsilon}^*)\\
        &\qquad\ge c\sqrt{\frac{\abs{\tau}+1}{n}}+\delta)\\
        \le&e^{-\frac{n\delta^2}{2}}.
    \end{align*}
    Then we can control the probability of false model selection as
    \begin{align*}
        &\Prob(\hat \tau(\bm{\epsilon}^*)\ne\tau_0)\\
        \le&\Prob(\inf_{\tau\ne\tau_0,\abs{\tau}\le s,{\bm{\beta}}_\tau\in\R^{\abs{\tau}},\sigma\ge 0}L_n(\tau,{\bm{\beta}}_\tau,\sigma|\bm{y},\bm{\epsilon}^*)\le L_{\bm{\theta}_0,n}(\tau_0,{\bm{\beta}}_{0,\tau_0},0|\bm{y},\bm{\epsilon}^*))\\
        \le&\sum_{\tau\ne\tau_0,\abs{\tau}\le s}\Prob(\inf_{{\bm{\beta}}_\tau\in\R^{\abs{\tau}},\sigma\ge 0}L_{\bm{\theta}_0}(\tau,{\bm{\beta}}_\tau,\sigma|\bm{y},\bm{\epsilon}^*)-L_{\bm{\theta}_0}(\tau_0,{\bm{\beta}}_{0,\tau_0},0|\bm{y},\bm{\epsilon}^*)\\
        &\qquad\le2\sup_{{\bm{\beta}}_{\tau}\in\R^{\abs{\tau}},\sigma\ge 0}\abs{L_n-L_{\bm{\theta}_0}}(\tau,{\bm{\beta}}_{\tau},\sigma|\bm{y},\bm{\epsilon}^*)\vee\abs{L_n-L_{\bm{\theta}_0}}(\tau_0,{\bm{\beta}}_{0,\tau_0},0|\bm{y},\bm{\epsilon}^*))\\
        \le&\sum_{\tau\ne\tau_0,\sigma\ge 0}e^{-\frac{n}{8}\big\{\inf_{{\bm{\beta}}_\tau\in\R^{\abs{\tau}},\sigma\ge 0}L_{\bm{\theta}_0}(\tau,{\bm{\beta}}_\tau,\sigma|\bm{y},\bm{\epsilon}^*)-L_{\bm{\theta}_0}(\tau_0,{\bm{\beta}}_{0,\tau_0},\sigma|\bm{y},\bm{\epsilon}^*)-c\sqrt{\frac{\abs{\tau}+1}{n}}\big\}^2}\\
        =&T.
    \end{align*}
    Similar with the proof of Lemma \ref{lem_oracle}, on the one hand, if we denote $j=\abs{\tau_0\cap\tau},l=\abs{\tau\setminus\tau_0}$, then
    \begin{align*}
        T\le\sum_{j=0}^{s-1}\sum_{l=0}^{s-j}\binom{s}{j}\binom{p-s}{l}e^{-\frac{n}{8}(s-j)\tilde c_{\min}}
        \lesssim e^{-\frac{1}{8}n\tilde c_{\min}+2\log p}.
    \end{align*}
    On the other hand, if we denote $j=\abs{\tau}$, then
    \begin{align*}
        T\le\sum_{j=0}^s\binom{p}{j}e^{-\frac{1}{8}n(j\vee 1)c_{\min}}
        \lesssim e^{-\frac{1}{8}nc_{\min}+\log p}.
    \end{align*}
    Suppose $\mathcal{C}=\{\hat\tau(\bm{\epsilon}^{*(j)}):\epsilon_i^{*(j)}\overset{{\rm i.i.d.}}{\sim}{\rm Logistic},i\in[n],j\in[d]\}$, then
    \[\Prob(\tau_0\not\in\mathcal{C})\le\Prob(\hat\tau(\bm{\epsilon}^{*})\ne\tau_0)\lesssim e^{-\frac{n}{8}\tilde c_{\min}+2\log p}\wedge e^{-\frac{n}{8}nc_{\min}+\log p},\]
    \[\Prob(\mathcal{C}\ne\{\tau_0\})\le\sum_{j=1}^d\Prob(\hat\tau(\bm{\epsilon}^{*(j)})\ne\tau_0)\lesssim e^{-\frac{n}{8}\tilde c_{\min}+2\log p+\log d}\wedge e^{-\frac{n}{8}c_{\min}+\log p+\log d}.\]
\end{proof}

\subsection{Proof of Theorem~\ref{thm_tau}}
To prove Theorem \ref{thm_tau}, we need the following lemma from \cite{berend2012convergence}.
\begin{Lemma}[Theorem 2 in \cite{berend2012convergence}]\label{lem_TV_tail}
For every $k\in\mathbb{N}$, distribution $\Prob$ with support on $k$ elements and its empirical version $\hat\Prob$ based on $m$ samples, we have
\[\Prob({\rm TV}(\Prob,\hat\Prob)>\delta)\le \exp\bigg(-\frac{m}{2}\bigg(2\delta-\sqrt{\frac{k}{m}}\bigg)^2\bigg),\quad2\delta\ge\sqrt{\frac{k}{m}}.\]
\end{Lemma}
\begin{proof}[Proof of Theorem \ref{thm_tau}]
    \begin{align*}
        &\Prob(\tau_0\not\in\Gamma_{\alpha}^{\tau_0}(\bm{X},\bm{y}))\\
        \le&\Prob(\tau_0\not\in\Gamma_{\alpha}^{\tau_0}(\bm{X},\bm{y}),\tau_0\in\mathcal{C})+\Prob(\tau_0\not\in\mathcal{C})\\
        \le&\Prob(\tilde T(\bm{X},\bm{y},\tau_0)\ge\alpha)+\Prob(\tau_0\not\in\mathcal{C})\\
        =&T_1+T_2.
    \end{align*}
    Denote $\Prob_\tau,\hat\Prob_\tau$ to be the probabilities with respect to $\tau$ with distribution $P_{\tau_0,{\bm{\beta}}_{0,\tau_0}},\hat P_{\tau_0,{\bm{\beta}}_{0,\tau_0}}$, respectively. And for simplicity, we also compress $P_{\tau_0,{\bm{\beta}}_{0,\tau_0}},\hat P_{\tau_0,{\bm{\beta}}_{0,\tau_0}}$ as $P,\hat P$, respectively. Then the nuclear statistic $\tilde T(\bm{X}^{obs},\bm{y}^{obs},\tau_0)$ can be expressed as
    \[\tilde T(\bm{X}^{obs},\bm{y}^{obs},\tau_0)=\hat\Prob_{\tau}(\hat P(\tau)>\hat P(\tilde\tau(\bm{X}^{obs},\bm{y}^{obs},\tau_0))).\]
    Note that $P$ only counts the randomness of $\bm{\epsilon}^{rel}$ conditional on $\bm{X}$, therefore it is still random. Since $\tilde\tau(\bm{X}^{obs},\bm{y}^{obs},\tau_0)$ is a realization from $P_{\tau_0,{\bm{\beta}}_{0,\tau_0}}$, then we can rewrite $T_1$ as
    \[T_1=\E\Prob_{\tilde\tau}(\hat\Prob_{\tau}(\hat P(\tau)>\hat P(\tilde\tau))\ge\alpha),\]
    where $\E$ includes the randomness of both the Monte Carlo samples $\{\bm{\epsilon}^{*(j)}:j\in[m]\}$ and the covariates $\bm{X}$.
    By the property of CDF, we know
    \begin{align*}
        \hat\Prob_{\tilde\tau}(\hat\Prob_\tau(\hat P(\tau)>\hat P(\tilde\tau))\ge\alpha)
        =\hat\Prob_{\tilde\tau}(\hat\Prob_{\tau}(\hat P(\tau)\le\hat P(\tilde\tau))\le1-\alpha)
        \le1-\alpha\quad{\rm a.s.}.
    \end{align*}
    Then
    \begin{align*}
        T_1=&\E\Prob_{\tilde\tau}(\hat\Prob_{\tau}(\hat P(\tau)>\hat P(\tilde\tau))\ge\alpha)-\E\hat\Prob_{\tilde\tau}(\hat\Prob_\tau(\hat P(\tau)>\hat P(\tilde\tau))\ge\alpha)
        +\E\hat\Prob_{\tilde\tau}(\hat\Prob_\tau(\hat P(\tau)>\hat P(\tilde\tau))\ge\alpha)\\
        \le&\E{\rm TV(\Prob_\tau,\hat\Prob_\tau)}+1-\alpha.
    \end{align*}
    By construction, $P$ and $\hat P$ have supports on at most $\sum_{j=0}^s{p\choose j}\le (\frac{ep}{s})^s$ elements, then we can control the expected total variation term using Lemma \ref{lem_TV_tail} as
    \begin{align*}
        &\E{\rm TV}(\Prob,\hat\Prob)\\
        =&\E{\rm TV}(\Prob,\hat\Prob)\1\bigg\{{\rm TV}(\Prob,\hat\Prob)\le\sqrt{\frac{(\frac{ep}{s})^s}{4m}}\bigg\}+\E{\rm TV}(\Prob,\hat\Prob)\1\bigg\{{\rm TV}(\Prob,\hat\Prob)>\sqrt{\frac{(\frac{ep}{s})^s}{4m}}\bigg\}\\
        \le&\sqrt{\frac{(\frac{ep}{s})^s}{4m}}+\int_{\sqrt{\frac{(\frac{ep}{s})^s}{4m}}}^\infty\Prob({\rm TV}(\Prob_\tau,\hat\Prob_\tau)>t)dt\\
        =&\sqrt{\frac{(\frac{ep}{s})^s}{4m}}+\int_{\sqrt{\frac{(\frac{ep}{s})^s}{4m}}}^\infty e^{-\frac{m}{2}(2t-\sqrt{\frac{(\frac{ep}{s})^s}{m}})^2}dt\\
        =&\sqrt{\frac{(\frac{ep}{s})^s}{4m}}+\sqrt{\frac{\pi}{8m}}.
    \end{align*}
    Therefore
    \[\Prob(\tau_0\not\in\Gamma_{\alpha}^{\tau_0}(\bm{X},\bm{y}))\le1-\alpha+\sqrt{\frac{(\frac{ep}{s})^s}{4m}}+\sqrt{\frac{\pi}{8m}}+\Prob(\tau_0\not\in\mathcal{C}).\]
    Together with Theorem \ref{thm_candidate} and Theorem \ref{thm_candidate_strong_signal}, we conclude the proof.
\end{proof}

\subsection{Proof of Theorem~\ref{thm_Abeta}}

\begin{Lemma}\label{lem_tensor}

For any fixed $B\in\R^{d_1\times s}$ with $BB^\top=I_{d_1}$, denote $\mathcal{B}_{d_1}=\{a\in\R^{d_1}:\norm{a}_2=1\}$ to be the unit sphere in $\R^{d_1}$, then
\[\sup_{a,b,c\in\mathcal{B}_{d_1}}\sum_{i=1}^n\abs{a^\top BX_{i,\tau_0} b^\top BX_{i,\tau_0} c^\top BX_{i,\tau_0}}=O_P(n+\sqrt{nd_1}+(d_1\log n)^{\frac{3}{2}}).\]
\end{Lemma}
\begin{proof}[Proof of Lemma \ref{lem_tensor}]
    Note that for any $a,b,c\in\mathcal{B}_{d_1}$, we have
    \[\norm{a^\top BX_{\tau_0} b^\top BX_{\tau_0} c^\top BX_{\tau_0}}_{\psi_{2/3}}\le\xi^3,\]
    \[\norm{\big(a^\top BX_{\tau_0} b^\top BX_{\tau_0} c^\top BX_{\tau_0}\big)^2}_{\psi_{1/3}}\le\xi^6.\]
    Denote $\mathcal{N}_{\frac{1}{4}}$ to be the $\frac{1}{4}$-net of $\mathcal{B}_{d_1}$, then $\log|\mathcal{N}_{\frac{1}{4}}|\lesssim d_1$. For any $a,b,c\in\mathcal{B}_{d_1}$, there exist $\tilde a,\tilde b,\tilde c\in\mathcal{N}_{\frac{1}{4}}$ such that $\|a-\tilde a\|_2,\|b-\tilde b\|_2,\|c-\tilde c\|_2\le\frac{1}{4}$, and
    \begin{align*}
        &\sum_{i=1}^n|a^\top BX_{i,\tau_0}b^\top BX_{i,\tau_0}c^\top BX_{i,\tau_0}|\\
        \le&\sum_{i=1}^n|\tilde a^\top BX_{i,\tau_0}\tilde b^\top BX_{i,\tau_0}\tilde c^\top BX_{i,\tau_0}|+\frac{61}{64}\sup_{a,b,c\in\mathcal{B}_{d_1}}\sum_{i=1}^n|a^\top BX_{i,\tau_0}b^\top BX_{i,\tau_0}c^\top BX_{i,\tau_0}|.
    \end{align*}
    Taking supremum over $a,b,c\in\mathcal{B}_{d_1}$, we get
    \[\sup_{a,b,c\in\mathcal{B}_{d_1}}\sum_{i=1}^n|a^\top BX_{i,\tau_0}b^\top BX_{i,\tau_0}c^\top BX_{i,\tau_0}|\le\frac{64}{3}\max_{a,b,c\in\mathcal{N}_{\frac{1}{4}}}\sum_{i=1}^n|a^\top BX_{i,\tau_0}b^\top BX_{i,\tau_0}c^\top BX_{i,\tau_0}|.\]
    By Theorem 3.4 in \cite{kuchibhotla2022moving}, we have
    \begin{align*}
        &\max_{a,b,c\in\mathcal{N}_{\frac{1}{4}}}\sum_{i=1}^n\bigg\{|a^\top BX_{i,\tau_0}b^\top BX_{i,\tau_0}c^\top BX_{i,\tau_0}|-\E|a^\top BX_{i,\tau_0}b^\top BX_{i,\tau_0}c^\top BX_{i,\tau_0}|\bigg\}\\
        =&O_P(\sqrt{nd_1}+(\log n)^{\frac{3}{2}}d_1^{\frac{3}{2}}).
    \end{align*}
    Then
    \begin{align*}
        &\sup_{a,b,c\in\mathcal{B}_{d_1}}\sum_{i=1}^n\abs{a^\top BX_{i,\tau_0} b^\top BX_{i,\tau_0} c^\top BX_{i,\tau_0}}\\
        \le&\frac{64}{3}\max_{a,b,c\in\mathcal{N}_{\frac{1}{4}}}\sum_{i=1}^n|a^\top BX_{i,\tau_0}b^\top BX_{i,\tau_0}c^\top BX_{i,\tau_0}|\\
        \le&\frac{64}{3}\max_{a,b,c\in\mathcal{N}_{\frac{1}{4}}}n\E|a^\top BX_{\tau_0} b^\top BX_{\tau_0} c^\top BX_{\tau_0}|\\
        &+\max_{a,b,c\in\mathcal{N}_{\frac{1}{4}}}\frac{64}{3}\sum_{i=1}^n\bigg\{|a^\top BX_{i,\tau_0}b^\top BX_{i,\tau_0}c^\top BX_{i,\tau_0}|-\E|a^\top BX_{i,\tau_0}b^\top BX_{i,\tau_0}c^\top BX_{i,\tau_0}|\bigg\}\\
        =&O_P(n+\sqrt{nd_1}+(d_1\log n)^{\frac{3}{2}}).
    \end{align*}
\end{proof}

\begin{Lemma}\label{lem_hessian}
For $B\in\R^{d_1\times s},d_1\in[s]$ with $BB^\top=I_{d_1}$, denote
\[\hat H_{\tau_0}=\sum_{i=1}^n\frac{e^{y_iX_i^\top{\bm{\beta}}_0}}{(1+e^{y_iX_i^\top{\bm{\beta}}_0})^2}BX_{i,\tau_0}X_{i,\tau_0}^\top B^\top,\]
then with probability at least $1-2e^{-cd_1}$,
\[\norm{\hat H_{\tau_0}-\E\hat H_{\tau_0}}_{\rm op}\lesssim\sqrt{(c+1)d_1n}+(c+1)d_1.\]
\end{Lemma}
\begin{proof}[Proof of Lemma \ref{lem_hessian}]
    Take $\mathcal{N}_{d_1}$ to be the $\frac{1}{4}$-nets of the unit ball $\mathcal{B}_{d_1}$ in $\R^{d_1}$. As in \cite{fan2021shrinkage}, if we denote
    \[\Phi(A)=\max_{(u,v)\in\mathcal{N}_{d_1}\times\mathcal{N}_{d_1}}u^\top Av,\]
    we have 
    \[\norm{A}_{\rm op}\le\frac{16}{7}\Phi(A).\]
    To see this, for any $(u,v)\in\mathcal{B}_{d_1}\times\mathcal{B}_{d_1}$, there exist $(u_1,v_1)\in\mathcal{N}_{d_1}\times\mathcal{N}_{d_1}$ such that $\norm{u-u_1}_2\le\frac{1}{4},\norm{v-v_1}_2\le\frac{1}{4}$,
    \begin{align*}
        u^\top Av=&u_1^\top Av_1+(u-u_1)^\top Av_1+u_1^\top A(v-v_1)+(u-u_1)^\top A(v-v_1)\\
        \le&\Phi(A)+(\frac{1}{4}+\frac{1}{4}+\frac{1}{16})\norm{A}_{\rm op}.
    \end{align*}
    Taking supremum on both sides yields the result.
    
    Fix any $(u,v)\in\mathcal{N}_{d_1}\times\mathcal{N}_{d_1}$ and denote $G_i=\frac{e^{y_iX_i^\top{\bm{\beta}}_0}}{(1+e^{y_iX_i^\top {\bm{\beta}}_0})^2}u^\top BX_{i,\tau_0}X_{i,\tau_0}^\top B^\top v$, then
    \[\norm{G_i}_{\psi_1}\le\norm{u^\top BX_{i,\tau_0}}_{\psi_2}\norm{v^\top BX_{i,\tau_0}}_{\psi_2}\le\xi^2.\]
    By Bernstein's inequality,
    \[\Prob(\abs{\hat H_{\tau_0}-\E\hat H_{\tau_0}}\gtrsim (\sqrt{nt}+t)\xi^2)\le 2e^{-t}.\]
    Applying union bound over $(u,v)\in\mathcal{N}_{d_1}\times\mathcal{N}_{d_1}$ and take $t=d_1(2\log 8+c)$, we have with probability at least $1-2e^{-cd_1}$,
    \[\norm{\hat H_{\tau_0}-\E\hat H_{\tau_0}}_{\rm op}\le\frac{16}{7}\Phi(\hat H_{\tau_0}-\E\hat H_{\tau_0})\lesssim \sqrt{(c+1)d_1n}+(c+1)d_1.\]
\end{proof}

\begin{proof}[Proof of Theorem \ref{thm_Abeta}]
    Given the true model $\tau_0$, without loss of generality, we assume $A_{\cdot\tau_0^c}=\bf{0}$, $A_{\cdot\tau_0}A_{\cdot\tau_0}^\top=I_r$ with $r=r(\tau_0)$. Denote $\tilde A=A_{\cdot\tau_0}$, $B\in\R^{(s-r)\times s}$ satisfies $(\tilde A^\top, B^\top)$ to be an orthogonal matrix. Then the null parameter space is 
    \begin{align*}
        \Theta_0=&\{{\bm{\beta}}\in\R^p:A{\bm{\beta}}=t,{\bm{\beta}}_{\tau_0^c}=\bf{0}\}
        =\{{\bm{\beta}}\in\R^p:{\bm{\beta}}_{\tau_0}=\tilde A^\top t+B^\top v,{\bm{\beta}}_{\tau_0^c}=\bf{0}\}.
    \end{align*}
    For ${\bm{\beta}}\in\Theta_0$ with ${\bm{\beta}}_{\tau_0}=\tilde A^\top t+B^\top v,{\bm{\beta}}_{\tau_0^c}=\bf{0}$,
    \[\nabla_vl({\bm{\beta}})=\sum_{i=1}^n\frac{y_i BX_{i,\tau_0}}{1+e^{y_iX_i^\top{\bm{\beta}}}},\]
    \[\nabla_{vv}^2l({\bm{\beta}})=-\sum_{i=1}^n\frac{e^{y_iX_i^\top{\bm{\beta}}}}{(1+e^{y_iX_i^\top{\bm{\beta}}})^2}BX_{i,\tau_0}X_{i,\tau_0}^\top B^\top.\]
    For the MLE $\hat b_0\in\Theta_0$ under null with $\hat b_{0,\tau_0}=\tilde A^\top t+B^\top\hat v$, we have
    \[\nabla_vl(\hat b_0)=0.\]
    Suppose ${\bm{\beta}}_{0,\tau_0}=\tilde A^\top t+B^\top v_0$, to control the error of $\hat v$ for $v_0$, note that if $l({\bm{\beta}}_0)>l({\bm{\beta}})$ for all ${\bm{\beta}}\in\mathcal{B}_c=\{{\bm{\beta}}:{\bm{\beta}}_{\tau_0^c}=0, {\bm{\beta}}_{\tau_0}={\bm{\beta}}_{0,\tau_0}+B^\top \Delta_v,\norm{\Delta_v}_2=c\sqrt{\frac{s-r}{n}}\}$, by the concavity of $l$, we have $\norm{\hat v-v_0}_2\le c\sqrt{\frac{s-r}{n}}$. For ${\bm{\beta}}\in\mathcal{B}_c$, there exists $\tilde{\bm{\beta}}$ between ${\bm{\beta}}$ and ${\bm{\beta}}_0$, such that
    \begin{align*}
        &l({\bm{\beta}})-l({\bm{\beta}}_0)\\
        =&\nabla_v^\top l({\bm{\beta}}_0)\Delta_v+\frac{1}{2}\Delta_v^\top\nabla_{vv}^2(\tilde {\bm{\beta}})\Delta_v\\
        \le&\norm{\nabla_v l({\bm{\beta}}_0)}_2c\sqrt{\frac{s-r}{n}}-\frac{1}{2}\lambda_{\min}(-\nabla_{vv}^2({\bm{\beta}}_0))c^2\frac{s-r}{n}+\frac{1}{2}\Delta_v^\top(\nabla_{vv}^2(\tilde{\bm{\beta}})-\nabla_{vv}^2({\bm{\beta}}_0))\Delta_v\\
        =&T_1+T_2+T_3.
    \end{align*}
    For $T_1$, since
    \[\norm{\frac{y_iBX_{i,\tau_0}}{1+e^{y_iX_i^\top{\bm{\beta}}_0}}}_{\psi_2}\le\norm{BX_{i,\tau_0}}_{\psi_2}\le\xi,\quad \E\nabla_vl(\bm{\beta}_0)=0,\]
    it follows from sub-Gaussian concentration inequality that
    \[\norm{\nabla_vl({\bm{\beta}}_0)}_{\psi_2}\lesssim \sqrt{n}.\]
    Then, by Lemma 1 in \cite{jin2019short}, we know with high probability
    \[T_1\lesssim c(s-r).\]
    By Lemma \ref{lem_hessian}, we can control $T_2$ as follows with high probability
    \[\abs{T_2}\gtrsim c^2(s-r).\]
    Denote $g(z)=\frac{e^z}{(1+e^z)^2}$, we have $\abs{g'(z)}\le \frac{1}{4}$. By Lemma 1 in \cite{jin2019short},
    \begin{align*}
        \norm{\norm{BX_{i,\tau_0}}_2}_{\psi_2}\le&\norm{\norm{BX_{i,\tau_0}-\E BX_{i,\tau_0}}_2}_{\psi_2}+\norm{\norm{\E BX_{i,\tau_0}}_2}_{\psi_2}\\
        \lesssim&\sqrt{s-r}.
    \end{align*}
    Since $(s\log n)^{\frac{3}{2}}\lesssim n$, by Lemma \ref{lem_tensor}, we have that with high probability
    \begin{align*}
        T_3=&\frac{1}{2}\Delta_v^\top(\nabla_{vv}^2(\tilde{\bm{\beta}})-\nabla_{vv}^2({\bm{\beta}}_0))\Delta_v\\
        =&\frac{1}{2}\sum_{i=1}^{n_0}(g(y_iX_i^\top\tilde{\bm{\beta}})-g(y_iX_i^\top{\bm{\beta}}_0))(X_{i,\tau_0}^\top B^\top\Delta_v)^2\\
        \le&\frac{1}{8}\sum_{i=1}^{n_0}\abs{X_{i,\tau_0}^\top B^\top\Delta_v}^3\\
        \lesssim&\big(n+\sqrt{n(s-r)}+((s-r)\log n)^{\frac{3}{2}}\big)\norm{\Delta_v}_2^3\\
        \lesssim&c^3\sqrt{\frac{(s-r)^3}{n}}.
    \end{align*}
    Then
    \[l({\bm{\beta}})-l({\bm{\beta}}_0)\lesssim c(s-r)-c^2(s-r)+c^3\sqrt{\frac{(s-r)^3}{n}}.\]
    Since $n\gg s-r$, by choosing $c\asymp 1$ large enough, with high probability, it happens
    \[l({\bm{\beta}})<l({\bm{\beta}}_0),\]
    which implies
    \[\norm{\hat v-v_0}_2=O_P(\sqrt{\frac{s-r}{n}}),\]
    \[\norm{\hat b_0-{\bm{\beta}}_0}_2=\norm{B^\top(\hat v-v_0)}_2=O_P(\sqrt{\frac{s-r}{n}}).\]
    Note that for ${\bm{\beta}}\in\Theta_0$,
    \begin{align*}
        \nabla_{{\bm{\beta}}_{\tau_0}} l({\bm{\beta}})=&\sum_{i=1}^n\frac{y_iX_{i,\tau_0}}{1+e^{y_iX_i^\top{\bm{\beta}}}}
        =\sum_{i=1}^n\frac{y_i\tilde A^\top\tilde AX_{i,\tau_0}}{1+e^{y_iX_i^\top{\bm{\beta}}}}+\frac{y_i B^\top BX_{i,\tau_0}}{1+e^{y_iX_i^\top{\bm{\beta}}}}.
    \end{align*}
    \[\nabla_{{\bm{\beta}}_{\tau_0}{\bm{\beta}}_{\tau_0}}^2l({\bm{\beta}})=-\sum_{i=1}^n\frac{e^{y_iX_i^\top{\bm{\beta}}}}{(1+e^{y_iX_i^\top{\bm{\beta}}})^2}X_{i,\tau_0}X_{i,\tau_0}^\top.\]
    If we denote $u=\sum_{i=1}^n\frac{y_i\tilde AX_{i,\tau_0}}{1+e^{y_iX_i^\top\hat b_0}}$, then
    \[\nabla_{{\bm{\beta}}_{\tau_0}}l(\hat b_0)=\tilde A^\top u.\]
    By Taylor expansion,
    \begin{equation}\label{eq_null_taylor}
        \tilde A^\top u=\nabla_{{\bm{\beta}}_{\tau_0}}l(\hat b_0)=\nabla_{{\bm{\beta}}_{\tau_0}}l({\bm{\beta}}_0)+\nabla_{{\bm{\beta}}_{\tau_0}{\bm{\beta}}_{\tau_0}}^2l({\bm{\beta}}_0)(\hat b_{0,\tau_0}-{\bm{\beta}}_{0,\tau_0})+R
    \end{equation}
    with
    \[R_j=(\nabla_{{\bm{\beta}}_{\tau_0}{\bm{\beta}}_{\tau_0}}^2l(\tilde{\bm{\beta}}^{(j)})-\nabla_{{\bm{\beta}}_{\tau_0}{\bm{\beta}}_{\tau_0}}^2l({\bm{\beta}}_0))_{j\cdot}(\hat b_{0,\tau_0}-{\bm{\beta}}_{0,\tau_0})\]
    for some $\tilde{\bm{\beta}}^{(j)}$ between $\hat b_0$ and ${\bm{\beta}}_0$, $j\in[s]$. Then for some $t_i^{(j)}$ between $y_iX_i^\top \tilde{\bm{\beta}}^{(j)}$ and $y_iX_i^\top{\bm{\beta}}_0$, $i\in[n], j\in[s]$, by Lemma \ref{lem_tensor},
    \begin{align*}
        \norm{R}_\infty=&\max_{j\in[s]}\sum_{i=1}^n\abs{(g(y_iX_{i,\tau_0}^\top\tilde {\bm{\beta}}_{\tau_0}^{(j)})-g(y_iX_{i,\tau_0}^\top{\bm{\beta}}_{0,\tau_0}))X_{i,\tau_{0,j}}X_{i,\tau_0}^\top(\hat b_{0,\tau_0}-{\bm{\beta}}_{0,\tau_0})}\\
        =&\max_{j\in[s]}\sum_{i=1}^n\abs{g'(t_i^{(j)})X_{i,\tau_{0,j}}(\tilde{\bm{\beta}}_{\tau_0}^{(j)}-{\bm{\beta}}_{0,\tau_0})^\top X_{i,\tau_0}(\hat b_{0,\tau_0}-{\bm{\beta}}_{0,\tau_0})^\top X_{i,\tau_0}}\\
        \le&\frac{1}{4}\max_{j\in[s]}\sum_{i=1}^n\abs{X_{i,\tau_{0,j}}(\hat b_{0,\tau_0}-{\bm{\beta}}_{0,\tau_0})^\top X_{i,\tau_0}(\hat b_{0,\tau_0}-{\bm{\beta}}_{0,\tau_0})^\top X_{i,\tau_0}}\\
        =&O_P(s-r),
    \end{align*}
    therefore
    \begin{align*}
        \norm{R}_2\le&\sqrt{s}\norm{R}_\infty
        =O_P(\sqrt{s}(s-r)).
    \end{align*}
    Denote $H=-\E\nabla_{{\bm{\beta}}_{\tau_0}{\bm{\beta}}_{\tau_0}}^2l({\bm{\beta}}_0), \hat H=-\nabla_{{\bm{\beta}}_{\tau_0}{\bm{\beta}}_{\tau_0}}^2l({\bm{\beta}}_0)$, by the same argument with Lemma \ref{lem_hessian}, we know
    \[\norm{H}_{\rm op}\asymp n,\qquad\norm{\hat H-H}_{\rm op}=O_P(\sqrt{sn}).\]
    Then by Equation \eqref{eq_null_taylor},
    \begin{equation}\label{eq_null_clt}
        \begin{aligned}
        &\sqrt{n}(\hat b_{0,\tau_0}-{\bm{\beta}}_{0,\tau_0})\\
        =&\sqrt{n}H^{-1}\nabla_{{\bm{\beta}}_{\tau_0}}l({\bm{\beta}}_0)-\sqrt{n}H^{-1}\tilde A^\top u+\sqrt{n}H^{-1}(\hat H-H)(\hat b_{0,\tau_0}-{\bm{\beta}}_{0,\tau_0})+\sqrt{n}H^{-1}R\\
        =&\sqrt{n}H^{-1}\nabla_{{\bm{\beta}}_{\tau_0}}l({\bm{\beta}}_0)-\sqrt{n}H^{-1}\tilde A^\top u+\tilde R,
    \end{aligned}
    \end{equation}
    where
    \[\norm{\tilde R}_2=O_P(\sqrt{\frac{s}{n}}(s-r)).\]
    Multiplying both sides of Equation \eqref{eq_null_clt} from the left by $\tilde A$, since $\tilde A\hat b_{0,\tau_0}=t=\tilde A\bm{\beta}_{0,\tau_0}$, we get
    \[u=(\tilde AH^{-1}\tilde A^\top)^{-1}\tilde AH^{-1}\nabla_{{\bm{\beta}}_{\tau_0}}l({\bm{\beta}}_0)+\frac{1}{\sqrt{n}}(\tilde AH^{-1}\tilde A^\top)^{-1}\tilde A\tilde R.\]
    Plugging back to Equation \eqref{eq_null_clt}, if we denote
    \[P=H^{-1/2}\tilde A^\top(\tilde AH^{-1}\tilde A^\top)^{-1}\tilde AH^{-1/2},\qquad Q=I_s-P,\]
    we get
    \begin{align*}
        \sqrt{n}(\hat b_{0,\tau_0}-{\bm{\beta}}_{0,\tau_0})=\sqrt{n}H^{-1/2}QH^{-1/2}\nabla_{{\bm{\beta}}_{\tau_0}}l({\bm{\beta}}_0)+H^{-1/2}QH^{1/2}\tilde R.
    \end{align*}
    Similarly, if we denote $\hat b_1$ as the MLE in $\Theta=\{{\bm{\beta}}:{\bm{\beta}}_{\tau_0^c}=\bf{0}\}$, then
    \[\norm{\hat b_1-{\bm{\beta}}_0}_2=O_P(\sqrt{\frac{s}{n}}),\quad\sqrt{n}(\hat b_{1,\tau_0}-{\bm{\beta}}_{0,\tau_0})=\sqrt{n}H^{-1}\nabla_{{\bm{\beta}}_{\tau_0}}l({\bm{\beta}}_0)+\hat R,\]
    with
    \[\norm{\hat R}_2=O_P(\sqrt{\frac{s^3}{n}}).\]
    Therefore
    \[\sqrt{n}(\hat b_{1,\tau_0}-\hat b_{0,\tau_0})=\sqrt{n}H^{-1/2}PH^{-1/2}\nabla_{{\bm{\beta}}_{\tau_0}}l({\bm{\beta}}_0)+\hat R-H^{-1/2}QH^{1/2}\tilde R,\]
    if we denote the SVD of $P=UU^\top$ with $U^\top U=I_r$,
    \begin{align*}
        \norm{\hat b_1-\hat b_0}_2\le&\norm{H^{-1/2}PH^{-1/2}\nabla_{{\bm{\beta}}_{\tau_0}}l({\bm{\beta}}_0)}_2+\frac{1}{\sqrt{n}}\norm{\hat R}_2+\frac{1}{\sqrt{n}}\norm{H^{-1/2}QH^{1/2}\tilde R}_2\\
        =&O_P(\frac{1}{\sqrt{n}}\norm{U^\top H^{-1/2}\nabla_{{\bm{\beta}}_{\tau_0}}l({\bm{\beta}}_0)}_2+\frac{\sqrt{s^3}}{n}+\frac{(s-r)\sqrt{s}}{n}).
    \end{align*}
    Since
    \[\norm{\frac{y_iU^\top H^{-1/2}X_{i,\tau_0}}{1+e^{y_iX_i^\top{\bm{\beta}}_0}}}_{\psi_2}\le\norm{H^{-1/2}}_{\rm op}\xi\lesssim \frac{1}{\sqrt{n}},\]
    we know from sub-Gaussian concentration inequality that
    \[\norm{U^\top H^{-1/2}\nabla_{{\bm{\beta}}_{\tau_0}}l({\bm{\beta}}_0)}_{\psi_2}\lesssim 1,\]
    then Lemma 1 in \cite{jin2019short} implies
    \[\norm{U^\top H^{-1/2}\nabla_{{\bm{\beta}}_{\tau_0}}l({\bm{\beta}}_0)}_2=O_P(\sqrt{r}).\]
    So
    \[\norm{\hat b_1-\hat b_0}_2=O_P(\sqrt{\frac{r}{n}}+\frac{\sqrt{s^3}}{n}).\]
    Then for some $\tilde{\bm{\beta}}$ between $\hat b_0$ and $\hat b_1$, we can control the log-likelihood ratio test statistic as
    \begin{align*}
        &-2(l(\hat b_0)-l(\hat b_1))\\
        =&-(\hat b_{1,\tau_0}-\hat b_{0,\tau_0})^\top\nabla_{{\bm{\beta}}_{\tau_0}{\bm{\beta}}_{\tau_0}}^2l(\tilde{\bm{\beta}})(\hat b_{1,\tau_0}-\hat b_{0,\tau_0})\\
        =&(\hat b_{1,\tau_0}-\hat b_{0,\tau_0})^\top H(\hat b_{1,\tau_0}-\hat b_{0,\tau_0})-(\hat b_{1,\tau_0}-\hat b_{0,\tau_0})^\top(H+\nabla_{{\bm{\beta}}_{\tau_0}{\bm{\beta}}_{\tau_0}}^2l(\tilde{\bm{\beta}}))(\hat b_{1,\tau_0}-\hat b_{0,\tau_0})\\
        =&\nabla_{{\bm{\beta}}_{\tau_0}}^\top l({\bm{\beta}}_0)H^{-1/2}PH^{-1/2}\nabla_{{\bm{\beta}}_{\tau_0}} l({\bm{\beta}}_0)+\tilde{\tilde{R}},
    \end{align*}
    where
    \begin{align*}
        \tilde{\tilde{R}}=&\frac{2}{\sqrt{n}}(\hat R-H^{-1/2}QH^{-1/2}\tilde R)H^{1/2}PH^{-1/2}\nabla_{{\bm{\beta}}_{\tau_0}}l({\bm{\beta}}_0)\\
        &+\frac{1}{n}(\hat R-H^{-1/2}QH^{-1/2}\tilde R)^\top H(\hat R-H^{-1/2}QH^{-1/2}\tilde R)\\
        &-(\hat b_{1,\tau_0}-\hat b_{0,\tau_0})^\top(H+\nabla_{{\bm{\beta}}_{\tau_0}{\bm{\beta}}_{\tau_0}}^2l(\tilde{\bm{\beta}}))(\hat b_{1,\tau_0}-\hat b_{0,\tau_0}).
    \end{align*}
    By Lemma \ref{lem_tensor}, we can find $t_i$ between $y_iX_i^\top{\bm{\beta}}_0$ and $y_iX_i^\top\tilde{\bm{\beta}}$ such that
    \begin{align*}
        &\abs{(\hat b_{1,\tau_0}-\hat b_{0,\tau_0})^\top(\nabla_{{\bm{\beta}}_{\tau_0}{\bm{\beta}}_{\tau_0}}^2l({\bm{\beta}}_0)-\nabla_{{\bm{\beta}}_{\tau_0}{\bm{\beta}}_{\tau_0}}^2l(\tilde{\bm{\beta}}))(\hat b_{1,\tau_0}-\hat b_{0,\tau_0})}\\
        \le&\sum_{i=1}^n\abs{g'(t_i)({\bm{\beta}}_{0,\tau_0}-\tilde{\bm{\beta}}_{\tau_0})^\top X_{i,\tau_0}(\hat b_{1,\tau_0}-\hat b_{0,\tau_0})^\top X_{i,\tau_0}(\hat b_{1,\tau_0}-\hat b_{0,\tau_0})^\top X_{i,\tau_0}}\\
        \le&\frac{1}{4}\sum_{i=1}^n\abs{({\bm{\beta}}_{0,\tau_0}-\tilde{\bm{\beta}}_{\tau_0})^\top X_{i,\tau_0}(\hat b_{1,\tau_0}-\hat b_{0,\tau_0})^\top X_{i,\tau_0}(\hat b_{1,\tau_0}-\hat b_{0,\tau_0})^\top X_{i,\tau_0}}\\
        =&O_P(\frac{r\sqrt{s}}{\sqrt{n}}+\frac{s^{7/2}}{n^{3/2}}),
    \end{align*}
    then if $n\gg\frac{s^3}{r^{1/2}}$, we have
    \[\abs{\tilde{\tilde{R}}}=o_P(\sqrt{r}).\]

Denote $Z_i=\frac{\sqrt{n}y_iUH^{-1/2}X_{i,\tau_0}}{1+e^{y_iX_i^\top{\bm{\beta}}_0}}$, we have $\Var(Z_i)=I_{r}$ $\norm{Z_i}_{\psi_2}\lesssim\xi$, therefore for any $j_1,j_2,j_3,j_4\in[r]$,
\[\E\abs{Z_{i,j_1}Z_{i,j_2}Z_{i,j_3}Z_{i,j_4}}<\infty.\]
Denote $G\sim N(0,I_r)$, $Z_i^{\otimes 3}=Z_i\otimes Z_i\otimes Z_i$ to be a tensor in $\R^{r^{\otimes 3}}$, by Corollary 4.10 in \cite{wang2017operator},
\begin{align*}
    \norm{\E Z_i^{\otimes 3}}_{\rm F}\le r\norm{\E Z_i^{\otimes 3}}_{\rm sp}=r\sup_{a,b,c\in\mathcal{B}_s}\E a^\top Z_i b^\top Z_i c^\top Z_i\lesssim r.
\end{align*}
By Lemma 1 in \cite{jin2019short}, we know $\norm{\norm{Z_i}_2}_{\psi_2}\lesssim \sqrt{r}$, then
\[\E\norm{Z_i}_2^4\lesssim r^2.\]
Then Corollary 1 in \cite{zhilova2022new} implies
\begin{align*}
    \sup_{c\ge 0}\abs{\Prob\bigg(\norm{\frac{1}{\sqrt{n}}\sum_{i=1}^n Z_i}_2^2\le c\bigg)-\Prob(\norm{G}_2^2\le c)}\lesssim\frac{r}{\sqrt{n}}.
\end{align*}
Now we can approximate the likelihood ratio test statistic as
\begin{align*}
    &\sup_{c\ge 0}\Prob(-2(l(\hat b_0)-l(\hat b_1))\le c\sqrt{r})-\Prob(\norm{G}_2^2\le c\sqrt{r})\\
    =&\sup_{c\ge0}\Prob\bigg(\norm{\frac{1}{\sqrt{n}}\sum_{i=1}^n Z_i}_2^2+\tilde{\tilde{R}}\le c\sqrt{r}\bigg)-\Prob(\norm{G}_2^2\le c\sqrt{r})\\
    \le&\sup_{c\ge0}\Prob\bigg(\norm{\frac{1}{\sqrt{n}}\sum_{i=1}^n Z_i}_2^2\le (c+\delta)\sqrt{r}\bigg)+\Prob(\tilde{\tilde{R}}\le-\delta \sqrt{r})\\
    &-\Prob(\norm{G}_2^2\le (c+\delta)\sqrt{r})+\Prob(c\sqrt{r}<\norm{G}_2^2\le (c+\delta)\sqrt{r})\\
    \le&\sup_{c\ge 0}\abs{\Prob\bigg(\norm{\frac{1}{\sqrt{n}}\sum_{i=1}^n Z_i}_2^2\le (c+\delta)\sqrt{r}\bigg)-\Prob(\norm{G}_2^2\le (c+\delta)\sqrt{r})}\\
    &+\Prob(\tilde{\tilde{R}}\le-\delta \sqrt{r})+\sup_{c\ge 0}\Prob(c\sqrt{r}<\norm{G}_2^2\le (c+\delta)\sqrt{r})\\
    \rightarrow&0,
\end{align*}
where the last step follows from Lemma S.7 in \cite{shi2019linear} by letting $n$ go to infinity then $\delta$ go to 0 from the right. Similarly,
\begin{align*}
    &\sup_{c\ge 0}\Prob(\norm{G}_2^2\le c\sqrt{r})-\Prob(-2(l(\hat b_0)-l(\hat b_1))\le c\sqrt{r})\\
    =&\sup_{c\ge0}\Prob\bigg(\norm{\frac{1}{\sqrt{n}}\sum_{i=1}^n Z_i}_2^2+\tilde{\tilde{R}}> c\sqrt{r}\bigg)-\Prob(\norm{G}_2^2> c\sqrt{r})\\
    \le&\sup_{c\ge0}\Prob\bigg(\norm{\frac{1}{\sqrt{n}}\sum_{i=1}^n Z_i}_2^2> (c-\delta)\sqrt{r}\bigg)+\Prob(\tilde{\tilde{R}}>\delta \sqrt{r})\\
    &-\Prob(\norm{G}_2^2> (c-\delta)\sqrt{r})+\Prob((c-\delta)\sqrt{r}<\norm{G}_2^2\le c\sqrt{r})\\
    \le&\sup_{c\ge 0}\abs{\Prob\bigg(\norm{\frac{1}{\sqrt{n}}\sum_{i=1}^n Z_i}_2^2>(c-\delta)\sqrt{r}\bigg)-\Prob(\norm{G}_2^2>(c-\delta)\sqrt{r})}\\
    &+\Prob(\tilde{\tilde{R}}>\delta \sqrt{r})+\sup_{c\ge 0}\Prob((c-\delta)\sqrt{r}<\norm{G}_2^2\le c\sqrt{r})\\
    \rightarrow&0.
\end{align*}
Now the confidence set satisfies
\begin{align*}
    \Prob(A{\bm{\beta}}_0\in\Gamma_{\alpha}^{A{\bm{\beta}}_0}(\bm{X},\bm{y}))\ge\alpha-o(1)-\Prob(\tau_0\not\in\mathcal{C})\rightarrow\alpha.
\end{align*}
\end{proof}

\end{document}